\documentclass[10pt, aps, prx,superscriptaddress,footinbib]{revtex4-2}
\usepackage{graphicx}
\usepackage{mwe}     

\usepackage[utf8]{inputenc}
\usepackage[T1]{fontenc}
\usepackage{amsmath}
\usepackage{amsthm}
\usepackage{amsfonts}
\usepackage{physics}
\usepackage{hyperref}
\usepackage[ruled, vlined]{algorithm2e}
\usepackage[nameinlink]{cleveref}
\usepackage{tikz}
\usetikzlibrary{arrows.meta, positioning, calc}
\usepackage{verbatim}
\usepackage{bm}
\RestyleAlgo{ruled} %
\usepackage{enumitem}
\usepackage{placeins}
\usepackage{subcaption} 
\usepackage[disable]{todonotes}
\usepackage{quantikz}
\usepackage[normalem]{ulem}
\usepackage[commandnameprefix=ifneeded]{changes}
\usepackage{xcolor}

\newcommand{\widebar}{\overline}

\SetKwProg{Procedure}{procedure}{}{} %
\SetKw{KwTo}{to}
\SetKw{KwAnd}{and}

\newtheorem{theorem}{Theorem}
\newtheorem{observation}{Observation}
\newtheorem{lemma}{Lemma}
\newtheorem{definition}{Definition}
\newtheorem{problem}{Problem}
\newtheorem{corollary}{Corollary}

\theoremstyle{plain}
\usetikzlibrary{calc}

\DeclareMathOperator{\poly}{poly}

\DeclareMathOperator{\spanop}{span}
\DeclareMathOperator{\swapop}{swap}

\newcommand{\anc}{\mathrm{anc}}

\newcommand{\slack}{\mathrm{slack}}

\let\originalleft\left
\let\originalright\right
\renewcommand{\left}{\mathopen{}\mathclose\bgroup\originalleft}
\renewcommand{\right}{\aftergroup\egroup\originalright}

\usepackage{mathtools}
\DeclarePairedDelimiter{\dpdparens}{\lparen}{\rparen}
\NewDocumentCommand{\parens}{s o m} {
	\IfBooleanTF{#1}
	{\dpdparens{#3}} %
	{\IfNoValueTF{#2} {\dpdparens*{#3}} {\dpdparens[#2]{#3}}}
}
\DeclarePairedDelimiter{\dpdbracks}{\lbrack}{\rbrack}
\NewDocumentCommand{\bracks}{s o m} {
	\IfBooleanTF{#1}
	{\dpdbracks{#3}} %
	{\IfNoValueTF{#2} {\dpdbracks*{#3}} {\dpdbracks[#2]{#3}}}
}

\newcommand{\eps}{\varepsilon}
\newcommand{\floor}[1]{\left\lfloor #1 \right\rfloor}
\newcommand{\RR}{\mathbb{R}}
\newcommand{\bx}{\mathbf{x}}

\begin{document}

\title{ Quantum Algorithms for Heterogeneous PDEs: \\
The Neutron Diffusion Eigenvalue Problem}

\author{Andrew M.\ Childs}
\affiliation{Joint Center for Quantum Information and Computer Science, NIST/University of Maryland, College Park, Maryland 20742, USA}
\affiliation{Department of Computer Science and Institute for Advanced Computer Studies, University of Maryland, College Park, Maryland 20742, USA}
\author{Lincoln Johnston}
\affiliation{Department of Nuclear Engineering and Radiological Sciences, University of Michigan, Ann Arbor, Michigan 48105, USA}
\author{Brian Kiedrowski}
\affiliation{Department of Nuclear Engineering and Radiological Sciences, University of Michigan, Ann Arbor, Michigan 48105, USA}
\author{Mahathi Vempati}
\affiliation{Joint Center for Quantum Information and Computer Science, NIST/University of Maryland, College Park, Maryland 20742, USA}
\affiliation{Department of Computer Science and Institute for Advanced Computer Studies, University of Maryland, College Park, Maryland 20742, USA}
\author{Jeffery Yu}
\affiliation{Joint Center for Quantum Information and Computer Science, NIST/University of Maryland, College Park, Maryland 20742, USA}
\affiliation{Joint Quantum Institute, NIST/University of Maryland, College Park, Maryland 20742, USA}

\date{\today}

\begin{abstract}
We develop a quantum algorithm to solve a type of linear reaction-diffusion equation, the neutron diffusion (generalized) \(k\)-eigenvalue problem that establishes nuclear criticality. The algorithm handles an equation with piecewise constant coefficients, describing a problem in a heterogeneous medium. We apply uniform finite elements and show that the quantum algorithm provides significant polynomial end-to-end speedup over  {classical uniform finite element methods { in the cases tested,} although nonuniform classical methods can perform much better than uniform ones.}
 {Our work} leverages recent advances in quantum linear systems---fast inversion and quantum preconditioning---and uses Hamiltonian simulation as a subroutine.
Our results suggest that quantum algorithms may provide speedups for heterogeneous PDEs, though the extent of this advantage over the fastest classical algorithm depends on the effectiveness of other classical approaches such as nonuniform or adaptive meshing for a given problem instance. 
\end{abstract}

\maketitle
\listoftodos
\tableofcontents
\section{Introduction}
\label{sec:intro}

One of the main goals of quantum computing research is to identify practical problems that can be solved significantly faster by quantum computers than by classical ones \cite{babbush2025grandchallengequantumapplications}. Partial differential equations (PDEs) \cite{PDE_in_20th_Century} are a natural target for two reasons. First, they arise in many domains---fluid dynamics, heat transfer, electromagnetism, structural mechanics, etc.---and have wide commercial and scientific importance. Most practically relevant PDEs do not have analytical solutions and are currently solved by computationally intensive classical numerical methods. Second, when discretized, many PDEs reduce to performing linear algebra on large, sparse matrices whose inputs are functionally provided \cite{Montanaro_2016_quantum_algorithms_finite_element}, enabling the use of quantum algorithms for linear systems \cite{HHL09}.

The possibility of end-to-end quantum speedups for PDEs was first investigated by Montanaro and Pallister \cite{Montanaro_2016_quantum_algorithms_finite_element}. They consider one of the standard methods for solving PDEs, the finite element method (FEM), and compare the performance of quantum and classical uniform FEM algorithms for obtaining a functional of the solution of   { second-order elliptic boundary value problems.} In the uniform FEM, the domain is discretized into a mesh with uniformly spaced cells, and solving the PDE is reduced to solving a linear system of equations. Previous work that considered applying the quantum linear systems algorithm to the FEM had claimed an exponential quantum speedup in the number of mesh elements \(N\). \cite{Clader_2013_preconditioned}. However, Montanaro and Pallister note that \(N\) is not an independent parameter, and depends on the desired accuracy of the solution \(\epsilon\) as \(N = O(\poly(1/\epsilon))\) \cite[Section II.B]{Montanaro_2016_quantum_algorithms_finite_element}. Moreover, as estimating an arbitrary observable for a given quantum state to \(\epsilon\) error requires \(\Omega(1/\epsilon)\) uses of a state-preparation unitary in the black-box model (as a consequence of the amplitude estimation lower bound \cite{Nayak_Wu_1999, Wang_lower_bounds_2025}), quantum algorithms likely give at most a polynomial speedup over classical algorithms in the parameter \(1/\epsilon\). Assuming optimal preconditioning and that the relevant solution norms are constant, \cite{Montanaro_2016_quantum_algorithms_finite_element} show that the classical and quantum complexities to solve this problem with piecewise linear elements are \(\tilde O\parens{\epsilon^{-d/2}}\) and \(\tilde O(\epsilon^{-1})\), respectively, where \(d\) is the spatial dimension of the PDE. Thus, for most physically relevant PDEs, which are two- or three-dimensional, the quantum algorithm does not seem to offer a significant speedup. 

However, the exponent \(d/2\) in the classical complexity arises from the regularity of the PDE solution, which occurs for PDEs with certain structure, such as with constant coefficients and a convex domain \cite{Petzoldt2001}. In many realistic applications, and in computationally hard PDE instances, this is not the case. Wave propagation in geophysics and seismology \cite{abdulle2017multiscale}, reaction-diffusion equations that model spatial distributions in populations \cite{cantrell2003spatial}, petroleum reservoir simulations  \cite{Christie_2001_SPE_Comparative_Solution_Project}, and metamaterial design \cite{Chung2025multiscalemethodswavepropagation} all involve PDEs with spatially varying coefficients, often with strong heterogeneity, and in many cases, discontinuous material interfaces. In such cases, a finer mesh may be required to obtain a given accuracy, and the exponent in the classical complexity for uniform FEM can be much larger than \(d/2\) depending on the degree of irregularity of the PDE solution \cite{Petzoldt2001,Nochetto_2010}. If the quantum complexity remains \(\tilde O(\epsilon^{-1})\) in these cases, quantum algorithms implementing the uniform FEM could offer significant speedup over their classical counterparts.

A common class of differential equations featuring varying coefficients is the class of linear reaction-diffusion equations. Consider the following eigenvalue form of this equation:
\begin{equation} \label{eq:reaction-diffusion}
\left ( -\nabla \cdot (D(\bx) \nabla) + f(\bx) \right ) \phi(\bx) = \lambda \phi(\bx).
\end{equation}
Linear reaction-diffusion equations are used to model systems in a broad variety of fields. For example, \cite{cantrell2003spatial} presents a linear reaction-diffusion eigenvalue equation
to model the spatial distribution of individuals of a species in an environment. The goal is to obtain the principal eigenvalue, which determines the rate of growth or decay of the total population of a species. 
In a similar manner, \cite{Dockery1998evolution} use this equation to model the prevalence of phenotypes expressed in members of a population within an environment. 
The linear reaction-diffusion eigenvalue equation can also be used to determine the basic reproduction number, \(R_0\), of an infectious disease \cite{Wang2012epidemic}. The sign of the principal eigenvalue determines whether \(R_0\) is above or below unity, which in turn determines whether a disease survives or dies out, respectively. 

In this work, we investigate the quantum versus classical uniform FEM for a generalized eigenvalue form of \Cref{eq:reaction-diffusion}: the neutron diffusion \(k\)-eigenvalue problem. This is the simplest meaningful model of the neutron transport problem, which is a key computational problem in nuclear reactor design because it establishes the nuclear criticality of a system. Here, one considers a domain (a nuclear reactor) containing several materials (the fuel, moderator, control rods, etc.). Each material has different values of a diffusion coefficient \(D(\bx)\) as a function of the location $\bx\in\RR^3$ that indicates how easily neutrons can move through the material, an absorption cross section \(\Sigma_a(\bx)\), and the product of a fission cross section and an average number of neutrons produced per fission event \(\nu\Sigma_f(\bx)\). In steady state, the neutron flux \(\phi(\bx)\) in the reactor satisfies a balance equation, \Cref{eq:diffusion_equation_first} below, where losses due to diffusion and absorption are balanced by gains due to fission multiplied by a factor \(1/k\). The largest value of \(k\) for which a non-trivial solution exists is called the effective multiplication factor, and indicates whether the reactor is subcritical (\(k < 1\)), critical (\(k = 1\)), or supercritical (\(k > 1\)). Finding this value to high accuracy is a key challenge in reactor design \cite{hamilton_2018_thesis_k_eigenvalue,Calloo_2023_anderson_acceleration}.

We consider the simplest case of this problem with Dirichlet boundary conditions on a unit cube domain \([0, 1]^3\) with no neutron energy dependence, as follows.

\begin{problem}
\label{prob:neutron_diffusion_problem}
Let $\Omega = [0, 1]^3$. Given positive piecewise-constant functions \(D,\Sigma_a \colon \Omega \to \RR_{>0}\) and non-negative piecewise-constant function \(\nu\Sigma_f \colon \Omega \to \RR_{\ge 0}\), find \(k_{\max}\), the largest value \(k\) satisfying
\begin{equation}
    \label{eq:diffusion_equation_first}
\left ( -\nabla \cdot (D(\bx) \nabla) + \Sigma_a(\bx) \right ) \phi(\bx) = \frac{1}{k} \nu \Sigma_f(\bx) \phi(\bx)
\end{equation}
for some {nonzero function} \(\phi(\bx)\) such that \(\phi(\bx) = 0\) on the boundary of the cube \(\Omega\). The functions \(D(\bx)\), \(\Sigma_a(\bx)\), and \(\nu\Sigma_f(\bx)\) are constant on { axis-aligned cuboidal subdomains} and are provided as a list of region boundaries, and the values they take in each region. {We assume that the subset of \(\Omega\) on which \(\nu\Sigma_f(\bx) = 0\) has volume at most \(\rho \cdot \,\operatorname{vol}(\Omega)\), where \(0 \leq \rho < 1\) is a constant.}
\end{problem}

To solve \Cref{prob:neutron_diffusion_problem} using the uniform FEM, one discretizes the domain $\Omega$ into a mesh of  $N$ cells and solves the resulting discrete matrix eigenvalue problem using standard linear algebra techniques (e.g., power iteration). Due to low solution regularity \cite{Petzoldt2001}, the best available upper bound on the number of mesh elements in the 3D case to achieve \(\epsilon\) error is \(N = O(\epsilon^{-3 \pi/\gamma}) \), where \(\gamma = \sqrt{D_{\min}/D_{\max}}\) and \(D_{\min}\) and \(D_{\max}\) are the minimum and maximum values of the diffusion coefficient \(D(\bx)\), respectively. The classical complexity of the scheme is bottlenecked by the mesh size, even with optimal preconditioning. Moreover, in  \Cref{sec:numerical_experiments} we provide numerical evidence that this type of slow convergence occurs in practice for previously identified hard problem instances \cite{Nochetto_2010,Petzoldt2001}.

In contrast, we show that a quantum algorithm implementing the same FEM scheme can solve \Cref{prob:neutron_diffusion_problem} end-to-end in \(\tilde O(\epsilon^{-1})\) gates. A simplified, graphical representation of the quantum algorithm is shown in \Cref{fig:quantum-algorithm-flowchart}. {Essentially, we cast this problem as a Hamiltonian eigenvalue problem where the Hamiltonian involves discretizations of the operators in \Cref{prob:neutron_diffusion_problem}, and includes an inversion of these operators.} Standard quantum linear algebra techniques to solve this problem scale linearly with the condition number \(\kappa = 1/h^2 = \Omega(\epsilon^{-2 \pi/\gamma})\). { We note that while our quantum algorithm can give a significant speedup over classical methods running the same scheme, nonuniform classical methods can lower this speedup significantly as shown in \Cref{sec:nonuniform_classical_grid_scaling}.}

\begin{figure*}[t]
\centering
\begin{tikzpicture}[
    >=Latex,
    box/.style={
        draw,
        rounded corners=3pt,
        minimum width=3.0cm,
        minimum height=0.8cm,
        align=center,
        font=\small,
        fill=white
    },
    arrow/.style={thick,->},
    annot/.style={font=\footnotesize, align=center}
]

\fill[blue!8]  (-1.2,-4.5) rectangle (15.5,1.8);
\fill[gray!12] (-1.2,-5.5) rectangle (15.5,-4.5);

\node[font=\bfseries\small] at (-0.15,-1.35) {Quantum};
\node[font=\bfseries\small] at (-0.2,-5.0) {Classical};

\node[box] (prep) at (7.3,0.9)
{Prepare uniform superposition state as estimate of fine-mesh eigenvector};

\node[box] (qpe) at (7.3,-1.5)
{Apply phase estimation to an operator constructed using \\
quantum preconditioning (\Cref{sec:preconditioner_construction}) \\
and block encoding (\Cref{sec:block_encoding_hamiltonian})};

\node[box] (measure) at (7.3,-3.8)
{Measure and report the eigenvalue};

\draw[thick] (prep) -- ++(0.0,-0.7);
\draw[arrow] (7.3,-0.1) -- (qpe); %

\draw[thick] (qpe) -- ++(0.0,-1.05);
\draw[arrow] (7.3,-2.9) -- (measure); %

\draw[arrow] (measure) -- ++(0.0,-1.0); %

\node[annot] at (7.3,0.05) {fine-mesh eigenvector estimate quantum state}; %

\node[annot] at (7.3,-2.75) {fine-grid eigenvalue quantum state}; %

\node[annot] at (7.3,-5.1) {fine-grid eigenvalue}; %

\end{tikzpicture}
\caption{High-level workflow for the QPE-based eigenvalue algorithm.}
\label{fig:quantum-algorithm-flowchart}
\end{figure*}

{
Our main contributions are as follows: 
\begin{enumerate}
\item First, for the uniform FEM, we rigorously bound the size of the mesh required given the regularity and the desired eigenvalue accuracy. The analysis applies regularity results from \cite{Petzoldt2001} combined with eigenvalue convergence results from \cite{babuska_1991_eigenvalue_problems}. 
\item Our algorithm involves inverting a high-condition-number matrix, for which we combine the fast-inversion technique of \cite{Tong_2021_fast_inversion} and the preconditioning technique of \cite{deiml2025quantumrealizationfiniteelement}. However, \cite{deiml2025quantumrealizationfiniteelement} does not explicitly construct a block encoding for the preconditioner. Our work requires an implementation of this block encoding, which we construct. 
\item Our algorithm uses quantum phase estimation, which uses an initial state with sufficient overlap with the principal eigenvector. We show that for our matrix, the uniform superposition state has constant overlap with the principal eigenvector, and thus can be used as the initial state. 
\item We construct block encodings of the matrices that encode the initial data and of the operators in the discretized diffusion system that depend on this data.
\item We numerically investigate the performance of classical uniform and nonuniform methods, as well as the probability of success for the initial state preparation in the quantum algorithm. 
\end{enumerate}
}

At a high level, our approach is as follows:
\begin{enumerate}
\item We find it suffices to take { mesh cell size} \(h = c \cdot \epsilon^{\pi/\gamma}\) for some constant \(c\).

\item Next, 
we rearrange the discretized problem to become a standard Hamiltonian eigenvalue problem of the form \(H \psi = k \psi\), where \(H = C^{1/2} (L+A)^{-1} C^{1/2}\) and \(\psi = C^{1/2}\phi\). Here \(L\) and \(A\) correspond to the diffusion and absorption terms on the left-hand side of \Cref{eq:diffusion_equation_first} and \(C\) corresponds to the fission term on the right-hand side. Now, we can use Hamiltonian simulation and quantum phase estimation to find the eigenvalue \cite{Chakraborty_2019_Block_Encoded_Matrix_Powers,shao_2021_generalized_eigenvalue_ode}. This involves two steps: 
\begin{enumerate}
\item Constructing a block encoding of the Hamiltonian \(H\). 
\item Constructing an initial state that has sufficient overlap with the leading eigenstate of \(H\). 
\end{enumerate}
\item A challenge with the block-encoding step (a) is that \(H\) contains the inverse of \(L+A\), whose condition number is \(\kappa = \Theta \parens{ 1/h^2 }\). Directly inverting \(L+A\) would have a prohibitive cost as a result. We overcome this as follows:
\begin{enumerate}[label=(\roman*)]
\item We rewrite \((L+A)^{-1}\) as \((I+L^{-1}A)^{-1} L^{-1}\) as in \cite{Tong_2021_fast_inversion}. The first term now has an \(O(1)\) condition number, making this rearrangement fruitful if we can perform \(L^{-1}\) fast. 
\item \(L\) is an elliptic operator in a heterogeneous setting for which \cite{deiml2025quantumrealizationfiniteelement} provides a quantum preconditioning technique. Specifically, they give an operator \(F\) such that \(L^{-1} = F \parens{F^T L F}^+ F^T\) where \(F^T L F\) has effective condition number \(O(1)\), and a construction for \(F^T L F\). We give a construction for \(F\). 
\end{enumerate} 

{\item Finally we show that a uniform superposition can be used as an initial state, giving constant overlap.}
\end{enumerate}

Using this approach, we establish the following.

\begin{theorem} %
    \label{thm:intro_final_complexity}
\Cref{prob:neutron_diffusion_problem} can be solved with accuracy \(\epsilon\) and constant probability of success using 
$
    O \parens{ z \cdot \frac 1 \epsilon \poly(\log \parens{ \frac 1 \epsilon} } ) 
$
one- and two-qubit gates and classical operations, where \(z\) is the number of different material regions. The big $O$ hides constant factors depending on coefficients \(D\), \(\Sigma_a\), \(\nu\Sigma_f\), and consequently various norms of the solution.
\end{theorem}

We  emphasize that \Cref{thm:intro_final_complexity} does not immediately imply a significant quantum speedup over all possible classical algorithms for \Cref{prob:neutron_diffusion_problem}. Establishing this would require further investigation into other classical methods, such as nonuniform or adaptive finite element methods, which can mesh finely in regions of low regularity and coarsely in regions of high regularity \cite{Nochetto_2010}, as well as Monte Carlo methods. In particular, \Cref{prob:neutron_diffusion_problem} is an elliptic problem, and there exist Monte Carlo methods such as walk-on-spheres to compute functionals of the boundary value solution of this problem in \(\tilde O(1/\epsilon^2)\) steps classically \cite{sawhney2022gridfreemontecarlopdes}. While we are not aware of literature on Monte Carlo methods for the eigenvalue version in piecewise constant media as in \Cref{prob:neutron_diffusion_problem}, and dealing with interfaces seems challenging, we cannot definitively rule out such an algorithm. We further discuss Monte Carlo algorithms for related problems in \Cref{sec:related_work}.

Rather, our work shows that it is possible for quantum algorithms to significantly speed up the uniform FEM even for low-dimensional PDEs, to the extent that it becomes a reasonable algorithm even for extremely large mesh sizes. This suggests a direction of research to determine if a significant overall quantum advantage for PDEs is possible: consider existing quantum algorithms for broader classes of PDEs and investigate whether they are applicable in heterogeneous media or irregular domains, applying preconditioning techniques to reduce the complexity if necessary. If the brute-force strategy of uniform meshing is close to the best one can do for the given class of PDEs, then quantum algorithms may offer a significant advantage. For instance, it is known to be difficult to design Monte Carlo methods for second-order hyperbolic PDEs \cite{Yu2023Monte}, for which a fast quantum algorithm with constant coefficients is available \cite{Costa_2019_Wave_Equation}. In the field of computational fluid dynamics, mesh generation and adaptivity are considered severe bottlenecks, and adaptive mesh techniques have not seen widespread use due to inadequate error estimation capabilities and complex geometries \cite{cfd_meshing_2030}. Quantum algorithms for PDEs in computational fluid dynamics have previously been considered, and one of the bottlenecks to quantum speedup here is the condition number \cite{jennings2025endtoendquantumalgorithmnonlinear}. By rigorously
treating the question of regularity in a quantum algorithms context, and showing the capability
of preconditioners to overcome the large condition numbers of discretized algorithms, we lay the
foundation for investigating quantum algorithms for heterogeneous media.

The rest of our paper is structured as follows. In \Cref{sec:related_work}, we discuss the relevant literature in more detail. In \Cref{sec:preliminaries}, we discuss notation, assumptions, and other background. In \Cref{sec:fem}, we introduce the finite element method, state the discrete version of our problem, and present properties of the matrices involved. In \Cref{sec:pde_convergence_analysis}, we give bounds on the mesh size \(h\) required to achieve accuracy \(\epsilon\) for our problem. In \Cref{sec:preconditioner_construction}, we give a construction for the preconditioner \(F\). In \Cref{sec:block_encoding_hamiltonian}, we construct the block encoding of the Hamiltonian \(H\). In \Cref{sec:complexity_analysis}, we prove the main theorem. In \Cref{sec:numerical_experiments}, we give numerical evidence for classical hardness. We conclude and further discuss open questions in \Cref{sec:conclusion_and_open_questions}.

\section{Related Work}
\label{sec:related_work}

In this section we briefly discuss some additional connections to related work.

\paragraph{Quantum finite element methods.} Montanaro and Pallister \cite{Montanaro_2016_quantum_algorithms_finite_element} first studied quantum algorithms for finite element methods. They considered the Poisson equation \(\nabla^2 \psi = f\) with given boundary conditions, and the problem of obtaining some linear functional of \(\psi\). They show that for optimal preconditioning, the classical and quantum scaling in \(\epsilon\) to solve this problem for \(d\) dimensions are \(\tilde O \parens{(1/\epsilon)^{d/2}}\) and \(\tilde O\parens{(1/\epsilon})\) and respectively. This differs from our \Cref{prob:neutron_diffusion_problem} in two main ways. First, we consider an eigenvalue problem while  \cite{Montanaro_2016_quantum_algorithms_finite_element} considers a boundary value problem. Second, the piecewise-constant diffusion coefficient \(D(\bx)\) that appears in our case forces solutions to have low regularity. Nevertheless, for both eigenvalue problems and boundary value problems, the analysis for the size of the mesh required to obtain a given error are the same, and we also use the bound given in Equation 3 of \cite{Montanaro_2016_quantum_algorithms_finite_element}: \(\abs {\phi -  \phi_h}_{H^1} \leq Ch^r \abs{\phi}_{H^{r+1}}\) where \(\abs{\,\cdot\,}_{H^l}\) is the Sobolev \(l\)-seminorm (see \Cref{sec:preliminaries}). That is, the distance between the true and approximate eigenfunction is bounded by \(h^r\) times the \((r+1)\)th seminorm of \(\phi\). The more regular \(\phi\) is, the larger the value of \(r\) for which the seminorm exists, and we can get better convergence and consequently smaller mesh sizes. In our case, the best lower bound we find for \(r\) is quite small: \(\gamma/2\pi\), resulting in large mesh sizes. We exploit the fact that quantum algorithms are not affected by the mesh size, but classical algorithms are. 

We also use the techniques of \cite{deiml2025quantumrealizationfiniteelement} for preconditioning in the case of piecewise constant \(D(\bx)\). Other studies of quantum finite element/difference methods include an investigation of the drift-diffusion equation \cite{Devereux_2025_drift_diffusion}, 
an approach to constructing flexible meshes \cite{alkadri2025quantumalgorithmfiniteelement}, and the related multi-scale elliptic equation in the non-end-to-end setting \cite{hu2023quantumalgorithmsmultiscalepartial}.

\paragraph{Eigenvalue problems.} {Reference} \cite{Abrams_1999} discuss{es} the use of quantum phase estimation to find eigenvalues of a Hamiltonian and \cite{Jaksch2003_Eigenvector_approximation_coarse_grid} discuss{es} preparing the required initial state classically using a coarse grid approximation of the Hamiltonian. 
Both \cite{Parker_2020_QPE_generalized} and \cite{shao_2021_generalized_eigenvalue_ode} discuss quantum algorithms for \textit{generalized} eigenvalue problems of the form \(T \phi = \lambda C \phi\). \cite{Parker_2020_QPE_generalized} consider the case where \(T\) is symmetric and \(C\) is positive-definite, and convert this into a standard eigenvalue problem by writing \(C^{-1/2} T C^{-1/2} \psi = \lambda \psi\). They then use phase estimation to solve the problem using a number of gates that scales with the square of the condition number \(\kappa(C)\), and inversely with the error \(\epsilon\). In our case, \(C\) is singular, so this rearrangement cannot be used. Even if we considered \(T^{-1/2} C T^{-1/2} \psi = \lambda' \psi\), this involves taking the square root of \(T^{-1}\) which is also known to scale polynomially with condition number. Thus, we use a different, but related rearrangement that avoids these issues. {Reference} \cite{shao_2021_generalized_eigenvalue_ode} consider{s} a more general class of matrices and solve{s} the eigenvalue problem by converting the equation into a system of ordinary differential equations. This reduces the condition number dependence of \cite{Parker_2020_QPE_generalized} to \(\kappa(C)^{1.5}\).

\paragraph{Quantum preconditioning.}{Reference} \cite{Clader_2013_preconditioned} discuss{es} preconditioning for quantum linear systems. However, they do not account for subnormalization due to the preconditioner as discussed in \cite{deiml2025quantumrealizationfiniteelement}, and also do not consider the dependence of the matrix size on the error \(\epsilon\), which appears to invalidate the claimed exponential speedup
\cite{Montanaro_2016_quantum_algorithms_finite_element}. {Reference} \cite{Lapworth_2025} use{s} preconditioning to give constant-factor improvements in the condition number. {Reference} \cite{Shao_2018_circulant} discuss{es} applying circulant preconditioners and give numerical evidence for reduction in condition number. {Reference} \cite{Golden_2022_hydrological_flow} stud{ies} a system quite similar to ours in a different context: modeling hydrological flow in heterogeneous media. They use the inverse Laplacian as a preconditioner. However, the final condition number of the system is still \(\poly(1/h)\), unlike the preconditioner of \cite{deiml2025quantumrealizationfiniteelement}, which gives an \(O(1)\) final condition number. Recently, Li \cite{li2025new} developed a linear systems algorithm that trades condition number dependence for an input-specific parameter. 

\paragraph{Monte Carlo methods for neutron transport.} 

The \(k\)-eigenvalue neutron transport equation is a first-order hyperbolic equation in six variables (of which \Cref{prob:neutron_diffusion_problem} is an approximation). The standard way to solve this problem is using Monte Carlo particle transport with fission source iteration \cite{mncp_code}. The simplest form involves sampling \(N_0\) initial neutrons and simulating their paths through the reactor: whether they fission and release new neutrons or are lost due to leakage or capture, as well as the paths of the subsequent new neutrons. This simulation is done for several generations, keeping track of the number of neutrons produced in each generation. The value \(k\) can be computed from this simulation data to \(\epsilon\) error by setting \(N_0 = O(1/\epsilon^2)\). This computation involves a renormalization process between each generation that depends on all the samples.

One can ask if this procedure can be sped up on a quantum computer using the technique in \cite{Montanaro_2015}. Given an algorithm \(\mathcal A\) that can sample a random variable \(x\) from the required probability distribution and an algorithm that can compute a function \(g(x)\), %
one can obtain \(E[g(x)]\) to \(\epsilon\) error using only \(O(1/\epsilon)\) calls to \(\mathcal A\), a quadratic improvement over classical sampling. However, due to the renormalization process that occurs between generations classically that depends on all the samples, it does not seem straightforward to write \(k\) as an expectation of a function \(g(x)\), depending on just one sample. While we cannot rule out the possibility of a quantum speedup using similar techniques, this would seem to require a novel approach.

\paragraph{Monte Carlo methods for the neutron diffusion eigenvalue problem.} 

Classical Monte Carlo methods to solve the neutron diffusion eigenvalue problem (\Cref{prob:neutron_diffusion_problem}) that we consider in this work are not as mature as those for neutron transport. Such methods have not been the focus of the neutron transport community as \Cref{prob:neutron_diffusion_problem} is only an approximation of the neutron transport equation, and a Monte Carlo method for the full equation is already available as discussed above. Nevertheless, for the fixed-source version of \Cref{prob:neutron_diffusion_problem} [\((\nabla \cdot (D(\bx) \nabla) - \Sigma_a(\bx)) \phi(\bx) = f(\bx)\)], there have been efforts to develop Monte Carlo methods to compute \(\phi(\bx)\) at a given \(\bx\). For example, \cite{sawhney2022gridfreemontecarlopdes} adapt the walk-on-spheres method to diffusion coefficients varying continuously in space, and \cite{ding2025walkoninterfacesmontecarloestimator} and \cite{LEJAY2010} give an algorithm for piecewise constant diffusion coefficients. 
While these methods use \(O(1/\epsilon^2)\) samples, the cost per sample is unclear, especially for varying geometries. Moreover, we are not aware of any work that adapts these methods to the eigenvalue version. The naive strategy for such an adaptation would be to discretize the spatial domain and perform a random walk for each discretized point/region to find the fission source to use for power iteration. However, this would have the same unfavourable complexity as uniform discretization.  Finally, the Diffusion Monte Carlo method has been used to compute the ground state energy of operators similar to the left-hand side of \Cref{eq:diffusion_equation_first} \cite{Reynolds1990}. Unfortunately, this work does not contain an error analysis, and has not been extended to spatially varying coefficients to the best of our knowledge. Overall, it appears that a more thorough analysis of the complexity of classical Monte Carlo methods for eigenvalue problems in heterogeneous media is required to determine the exact speedup of our algorithm over such methods.

\section{Preliminaries}
\label{sec:preliminaries}

In this section, we introduce notation, assumptions, and background material on block-encodings and Sobolev spaces that are used throughout the paper.

\subsection{Notation}
\label{sec:notation}

Typically, we denote the number of cubes in a mesh by \(M = \parens{\frac{1}{h}}^3\) and the number of internal nodes in a mesh by \(N = \parens{\frac{1}{h} - 1}^3\). We sometimes use \(N\) in other contexts and clarify this as needed.

We use the following notations:
\begin{itemize}
\item \(\phi, \psi, f\) refer to functions either in continuous or finite element space, depending on the context;

\item the notation \(\phi_h, \psi_h\) makes explicit that the functions \(\phi,\psi\) are in finite element space with mesh size \(h\);

\item \(\varphi_{ijk}\) refer to basis functions;

\item \(u, v\) refer to the discrete coefficient vectors corresponding to finite element functions; and

\item \(\hat{u}, \hat{v}\) refer to normalized coefficient vectors in discrete space.
\end{itemize}

\subsection{Assumptions}
\label{sec:assumptions}

We assume that one- and two-qubit gates can be implemented perfectly, and we give all complexities in terms of the number of one- and two- qubit gates, which we sometimes simply refer to as gates. To implement such circuits with any finite, universal, inverse-closed gate set, the gate complexity overhead for an implementation with error at most $\epsilon$ is \(\poly(\log(1/\epsilon))\) \cite{Nielsen_and_chuang}. 

The finite element method involves meshing the domain into cells. We assume we can use a small enough cell size that each cell contains only one material region, as in \cite{deiml2025quantumrealizationfiniteelement}.

\subsection{Background}
\label{sec:background}
\subsubsection{Block encodings}
We extensively use the block-encoding framework for quantum algorithms \cite{Gilyen_19_QSVT_and_beyond}. This framework provides a convenient way of encoding a general matrix into a unitary operator. We give the basic definition here, but refer the reader to \cite{Gilyen_19_QSVT_and_beyond} for further details about arithmetic with block encodings. 
\begin{definition} (Block encoding)
\label{def:block_encoding}
Suppose that \(A\) is an \(s\)-qubit operator, \(\alpha, \epsilon \in \mathbb R_+\), and \(q \in \mathbb N\). Then we say that the \((s+q)\)-qubit unitary \(U\) is an \((\alpha, q, \epsilon)\)-block encoding of \(A\) if
\begin{equation}
\norm{ A - \alpha (\bra{0}^{\otimes q} \otimes I) U (\ket{0}^{\otimes q} \otimes I) } \le \epsilon.
\end{equation}
\end{definition}

\subsubsection{Sobolev spaces}
The $L^p$ norm of a function $f \colon \Omega \to \RR$ is
\begin{equation}
\norm{f}_{L^p(\Omega)} = \parens{\int_\Omega \abs{f(\bx)}^p\, \dd{\bx} }^{1/p},
\end{equation}
which we abbreviate as $\norm{f}_{L^p}$ when the domain $\Omega$ is clear.
Normalization is mentioned explicitly when needed; otherwise we assume functions are \(L^2\)-normalized and vectors are \(\ell^2\)-normalized.

Throughout this work, all derivatives are to be understood in a weak sense. In particular, the physical solution may not be twice differentiable, yet the PDE in \Cref{prob:neutron_diffusion_problem} asks us to take two derivatives of the function. Therefore, in \Cref{sec:fem}, we convert the PDE into a weak form that permits the use of weak derivatives. We defer to any standard PDE text (e.g., \cite{evans10_partial_differential_equations}) for the formalization of weak derivatives and Sobolev spaces. We briefly describe the relevant aspects of Sobolev spaces below.

\begin{definition}
\label{def:sobolev}
Given an open subset $\Omega \subset \RR^n$ and a nonnegative integer $k$, the Sobolev space $H^k(\Omega)$ consists of all locally integrable (absolutely integrable over every compact subset) functions $f \colon \Omega \to \RR$ whose (weak) derivatives up to $k$th order all have finite $L^2$ norm. The Sobolev norm is 
\begin{equation}
\norm{f}_{H^k(\Omega)} = \parens[\Bigg]{\sum_{\alpha \in S_{\leq k}} \int_\Omega \abs{ D^\alpha f}^2 \, \dd{\bx}}^{1/2},
\end{equation}
where $S_{\leq k} := \{ \alpha = (\alpha_1, \alpha_2, \dots \alpha_n) \in \mathbb Z_{\geq 0}^n : \abs{\alpha} \leq k\}$, $D^\alpha f = \frac{\partial^{|\alpha|} f}{\partial x_1^{\alpha_1} \cdots \partial x_n^{\alpha_n}}$, $\abs{\alpha} = \sum_{i = 1}^n \alpha_i$, and $\dd{\bx}$ is the differential volume element in $\RR^n$.
\end{definition}

We also use the following common notation to denote a Sobolev semi-norm, where we only sum over derivatives of order exactly $k$: 
\begin{equation}
\abs{f}_{H^k(\Omega)} = \parens[\Bigg]{\sum_{\alpha \in S_k} \int_\Omega \abs{ D^\alpha f}^2 \, \dd{\bx}}^{1/2}, 
\end{equation}
where \({S_k := \{\alpha \in \mathbb Z_{\geq 0}^n : \abs{\alpha} = k\}}\). 

When the domain is clear, we abbreviate the Sobolev norm and seminorm as $\norm{f}_{H^k}$ and $\abs{f}_{H^k}$, respectively. For example, for $n=3$ and $k=1$, we have
\begin{equation}
\norm{f}_{H^1} = \parens{\int_\Omega |f|^2 \dd{\bx} + \int_\Omega \abs{\nabla f}^2 \dd{\bx}}^{1/2},
\end{equation}
where we use the vector shorthand $\abs{\nabla f}^2 = \nabla f \cdot \nabla f = \parens{\pdv{f}{x}}^2 + \parens{\pdv{f}{y}}^2 + \parens{\pdv{f}{z}}^2$.

While \Cref{def:sobolev} only applies for integer $k$, a more general construction gives Sobolev spaces for all positive real $k$, including non-integers. { We use the Slobodeckij construction, as detailed in~\cite[Definition B.30]{ern2004theory}.

\begin{definition}
\label{def:fractional-sobolev}
Given an open subset $\Omega \subset \RR^n$, a nonnegative integer $k$, and fractional part $0 < s < 1$, the fractional Sobolev space $H^{k+s}(\Omega)$ consists of all functions $f \in H^k(\Omega)$ such that $\frac{D^\alpha f(x) - D^\alpha f(y)}{\abs{x-y}^{\frac{n}{2} + s}} \in L^2(\Omega \times \Omega)$ for every $\alpha \in S_{k}$, or equivalently, the fractional Sobolev seminorm
\begin{equation}
\abs{f}_{H^{k+s}(\Omega)} \coloneq \parens[\Bigg] {\sum_{\alpha \in S_{k}} \int_\Omega \int_\Omega \frac{\abs{D^\alpha f(x) - D^\alpha f(y)}^2}{\abs{x-y}^{n+2s}} \dd{x} \dd{y} }^{1/2}
\label{eq:fractional-seminorm}
\end{equation}
is finite. The fractional Sobolev norm is given by
\(
\norm{f}_{H^{k+s}(\Omega)} = \parens{\norm{f}_{H^k(\Omega)}^2 + \abs{f}_{H^{k+s}(\Omega)}^2}^{1/2}.
\)
\end{definition}
}

Since the PDEs of interest have Dirichlet boundary conditions, we would like the approximate solutions to also satisfy Dirichlet boundary conditions. Therefore, we restrict the space of approximate solutions to $H^1_0(\Omega)$, which intuitively is the subspace of functions in $H^1(\Omega)$ that are $0$ on the boundary. For continuous functions, this is precise, but for non-continuous functions, care must be taken in defining what the value on the boundary means (see, e.g., \cite[Section 5.5]{evans10_partial_differential_equations} for a precise definition). To indicate this boundary condition precisely, we write \(\norm{\;\cdot\;}_{H^1_0} \coloneqq \norm{\;\cdot\;}_{H^1}\).

\section{Finite Element Method}
\label{sec:fem}
An issue with directly solving \Cref{prob:neutron_diffusion_problem} is that we are demanding the solution be twice differentiable at the outset, which at times might not include the physically relevant solutions. An extreme example comes from PDEs for fluid dynamics modeling shock waves which have discontinuities \cite{evans10_partial_differential_equations}. Thus, the standard approach to numerically solving PDEs is to consider a weaker formulation of the problem which has less stringent constraints, and then separately considering the issue of regularity of the solution. See \cite{evans10_partial_differential_equations} for more details, and \cite{Montanaro_2016_quantum_algorithms_finite_element} for the weak formulation in a quantum algorithms setting. We follow the same approach here. 

In order to approximately solve the problem, we consider solutions in a finite-dimensional subspace of the function space (in our case, piecewise multilinear functions on a mesh). This is the essence of the finite element method (FEM) \cite{ern2004theory}.

In this section, we first present the weak formulation of our problem, then define the finite element scheme (the more tractable subspace in which we look for solutions). This leads to a discrete eigenvalue problem. Finally, we describe properties of the matrices in this problem that are relevant for our quantum algorithm. 

\subsection{Weak Formulation}
\label{sec:weak_formulation}

We define the bilinear forms \(a, b \colon H^1_0(\Omega) \times H^1_0(\Omega) \to \RR\) %
as follows:
\begin{align}
a(\phi, \psi) &= \int_\Omega \parens[\big]{D(\bx) \nabla \phi \cdot \nabla \psi + \Sigma_a(\bx) \phi \psi} \dd{\mathbf x} \label{eq:def_a} \\
b(\phi, \psi) &= \int_\Omega \nu \Sigma_f(\bx) \phi \psi \, \dd{\bx}. \label{eq:def_b}
\end{align}
The weak form of \Cref{prob:neutron_diffusion_problem} is as follows.

\begin{problem} (Weak formulation)
\label{prob:orig_prob_weak_form}
We say that \((\lambda, \phi)\) is an eigenvalue-eigenfunction pair for the bilinear forms \(a\) and \(b\) (\Cref{eq:def_a,eq:def_b}) if \(\phi \in H^1_0(\Omega)\) and 
\begin{equation}
    \label{eq:weak_form}
a(\phi, \psi) = \lambda b(\phi, \psi) 
\end{equation}
for all \(\psi \in H^1_0(\Omega)\). Given \(a\) and \(b\), find \(\lambda_{\min}\), the smallest eigenvalue \(\lambda\) of \Cref{eq:weak_form}. 
\end{problem}

\subsection{Finite Element Scheme}
\label{sec:fem_scheme}

First, we describe the finite element function space \(V^h_0\). Let \(M = \parens{\frac{1}{h}}^3 \). Divide the domain \(\Omega = [0, 1]^3\) into
\(M\) uniform cubes \(c_{pqr}\) of side length \(h\) where \(p, q, r \in \{1,2,\ldots,\frac{1}{h}\}\) index the cubes in each dimension.
Explicitly, $c_{pqr}$ is the cube bounded by $(p-1)h \le x \le ph$, $(q-1)h \le y \le qh$, and $(r-1)h \le z \le rh$. Consider a reference cube \(C = [0, h]^3\) and let \(F_{pqr} \colon C \to c_{pqr}\) be the mapping of coordinates from \(C\) to \(c_{pqr}\) via translation:
\(F_{pqr}(x, y, z) = (x+(p-1)h, y+(q-1)h, z+(r-1)h).\) 

The standard \(Q_1\) finite element space (space of trilinear functions) on the reference cube is defined as \cite{john_num_pde_fem}
\begin{equation}
Q_1(C) = \spanop\{1, x, y, z, xy, xz, yz, xyz\}.
\end{equation}
Using the reference cube, we can define for each cube \(c_{pqr}\)
\begin{equation}
Q_1(c_{pqr}) = \{ \hat v \circ F^{-1}_{pqr} : \hat v \in Q_1(C) \}
\end{equation}
where \(\circ\) denotes function composition. 
We use this to define \(V^h_0(\Omega)\), the space of continuous piecewise trilinear functions on \(\Omega\) that vanish on the boundary:
\begin{equation}
V^h_0 = \{f \in H^1_0(\Omega) : f|_{c_{pqr}} \in Q_1(c_{pqr}) \text{ for all } p, q, r\},  
\end{equation}
where $f|_{c_{pqr}}$ denotes the restriction of $f$ to $c_{pqr}$.

A convenient basis for the space \(V^h_0\) is the nodal basis (also known as the Lagrange basis or the hat function basis) \cite{john_num_pde_fem, Montanaro_2016_quantum_algorithms_finite_element}. We define this first for \(1\) dimension, and then generalize to arbitrary dimensions. 

\begin{definition}[Nodal basis]
\label{def:nodal_basis}
Let \(h\) be the mesh cell parameter such that \(\frac 1 h \in \mathbb N\). Let \(r \in [0, 1]\). Then we define \(\frac{1}{h} - 1\) equidistant nodes \(\{r_m\}_m\) such that \(r_m \coloneqq mh\). Associated with each node \(r_m\), we define a nodal basis function \(\varphi_m(r) \colon [0, 1] \to [0, 1]\) as follows:
\begin{equation} 
\varphi_m(r) = 
\begin{cases} 
\frac{r - r_{m-1}}{h} & \text{if } r \in [r_{m-1}, r_m] \\
\frac{r_{m+1} - r}{h} & \text{if } r \in [r_m, r_{m+1}] \\
0 & \text{otherwise.} 
\end{cases} 
\
\end{equation}
We generalize this to arbitrary dimensions \(d\) as follows. Let \(\bm r \coloneqq (r_1, r_2 \dots r_d) \in [0, 1]^d\). Define the multi-index \(\bm m \coloneqq (m_1, m_2 \dots m_d)\) where \(m_i \in \left [1, \frac 1 h \right ]\) and grid nodes \(\{\bm r_{\bm m}\}_{\bm m}\) such that \(\bm r_{\bm m} \coloneqq (m_1 h, m_2 h \dots m_d h) \). Then the nodal basis function \(\varphi_{\bm m}(\bm r) \colon [0, 1]^d \to [0, 1]\) associated with node \(\bm r_{\bm m}\) is 
\begin{equation}
\label{eq:nodal_basis}
    \varphi_{\bm {m}}(\bm r) = \prod_{i \in [1, d]} \varphi_{m_i}(r_i).
\end{equation}
When we want to make the \(h\) dependence explicit, we denote this as \(\varphi^h_{\bm m}(\bm r)\). It will be convenient to use the notation \(\varphi_i(x)\), \(\varphi_{ij}(x, y)\), and \(\varphi_{ijk}(x, y, z)\) in the 1d, 2d, and 3d cases, respectively. 
\end{definition}

We are now ready to define an approximate version of \Cref{prob:orig_prob_weak_form}.
\begin{problem} (Finite element formulation)
\label{prob:finite-element-version} 
Given a positive constant \(h < 1\), find \(\lambda_{h, \min}\) which is the smallest value of $\lambda_h$ such that there exists \(\phi_h \in V^h_0(\Omega)\) satisfying
\begin{equation}
    \label{eq:weak_form_fem}
a(\phi_h, \psi_h) = \lambda_h b(\phi_h, \psi_h) 
\end{equation}
for all \(\psi_h \in V^h_0(\Omega)\).
\end{problem}

We discretize \Cref{prob:finite-element-version} as follows. Let \(\phi_h = \sum_{ijk} u_{ijk} \varphi_{ijk}\) and \(\psi_h = \sum_{i'j'k'} v_{i'j'k'} \varphi_{i'j'k'}\). Then, by linearity, it is sufficient to find coefficients \(u_{ijk}\) and the minimum \(\lambda_h\) such that for all \(\varphi_{i'j'k'}\), 
\begin{equation}
\label{eq:problem_discretized}
\sum_{ijk} u_{ijk} \; a(\varphi_{ijk}, \varphi_{i'j'k'}) = \lambda_h \sum_{ijk} u_{ijk} \; b(\varphi_{ijk}, \varphi_{i'j'k'}).
\end{equation}

Thus, we can restate \Cref{prob:finite-element-version} in matrix form.
\begin{problem} (Discrete generalized eigenvalue problem)
\label{prob:discrete_problem}
Find \(\lambda_{h, \min}\), the minimum \(\lambda_h\) such that there exists \(u_h \in \mathbb R^N\), where \(N = \left (\frac{1}{h} - 1 \right )^3\), such that
\begin{equation}
(L+A)u_h = \lambda_h C u_h
\end{equation}
where \(L\), \(A\), and \(C\) are matrices in \(\mathbb R^{N \times N}\) such that
\begin{equation}
\begin{aligned}
L_{ijk, i'j'k'} &= \int_\Omega D(\bx) \; \nabla \varphi_{ijk} \cdot \nabla \varphi_{i'j'k'} \, \dd \bx \\
A_{ijk, i'j'k'} &= \int_\Omega \; \Sigma_a(\bx) \; \varphi_{ijk} \varphi_{i'j'k'} \, \dd \bx\\
C_{ijk, i'j'k'} &= \int_\Omega \nu \Sigma_f(\bx) \; \varphi_{ijk} \varphi_{i'j'k'} \, \dd \bx.
\end{aligned}
\end{equation}
Moreover, from the constraints given in the original \Cref{prob:neutron_diffusion_problem}, we can define minimum and maximum values for \(D(\bx)\) and \(\Sigma_a(\bx)\) and a maximum value for \(\nu \Sigma_f(\bx)\) such that
\begin{equation}
    \begin{aligned}
    & 0 < D_{\min} \leq D(\bx) \leq D_{\max}, \\
    & 0 < \Sigma_{a, \min} \leq \Sigma_a(\bx) \leq \Sigma_{a, \max}, \\
    & 0 \leq \nu \Sigma_{f}(\bx) \leq \nu \Sigma_{f, \max}.
    \end{aligned}
\end{equation}
\end{problem}

We now rearrange \Cref{prob:discrete_problem} into a standard Hermitian eigenvalue problem with \(v_h = C^{1/2} u_h\) and \(k_h = 1/\lambda_h\), which is what we finally solve. 
\begin{problem} (Discrete standard eigenvalue problem)
\label{prob:discrete_problem_standard}
Find \(k_{h, \max}\), the maximum \(k_h\) such that there exists \(v_h \in \mathbb R^N\), where \(N = \left (\frac{1}{h} - 1 \right )^3\), such that
\begin{equation}
    H v_h = k_h v_h
\end{equation}
where \(H = C^{1/2} (L + A)^{-1} C^{1/2}\). The matrices \(L\), \(A\), and \(C\) and constraints on them are as defined in \Cref{prob:discrete_problem}.
\end{problem}

We select this rearrangement (as opposed to alternatives discussed in \Cref{sec:related_work}) as it is suited to the constraints of our problem, and eventually allows us to construct an efficient block encoding of \(H\).

\subsection{Matrix Properties}
\label{sec:matrix_properties}

{ We now discuss further properties of the matrices \(L\), \(A\), and \(C\) from \Cref{prob:discrete_problem}.} 

In this section, \(i, j, k, i', j',k' \in \left [1, \frac 1 h - 1 \right ]\). Observe that if \(\vert i - i' \vert > 1\) or \(\vert j - j' \vert > 1\) or \(\vert k - k' \vert > 1 \), then \(\varphi_{ijk} \varphi_{i'j'k'} = 0\) and \(\nabla \varphi_{ijk} \cdot \nabla \varphi_{i'j'k'} = 0\). Thus, these matrices are sparse with at most \(27\) elements in each row. 

To bound eigenvalues and condition numbers of \(L\), \(A\), and \(C\), it is useful to define the 1D, 2D, and 3D \textit{mass matrices} and Laplacians (discretized using the finite element method), and prove bounds on their eigenvalues.

\begin{definition}
\label{def:mass_matrix}
The 1D mass matrix \(M^{(1)}\) has entries 
\begin{equation}
M^{(1)}_{i,i'} = \int_{[0, 1]} \varphi_i \varphi_{i'} \; \dd x.
\end{equation}
Likewise,
\begin{equation}
M^{(2)}_{ij, i'j'} = \int_{[0, 1]^2} \varphi_{ij} \varphi_{i'j'} \; \dd x \, \dd y
\end{equation}
\begin{equation}
M^{(3)}_{ijk, i'j'k'} = \int_{[0, 1]^3} \varphi_{ijk} \varphi_{i'j'k'} \; \dd x \, \dd y \, \dd z.
\end{equation}
\end{definition}

\begin{definition}
\label{def:laplacian}
The 1D Laplacian \(P^{(1)}\) has entries
\begin{equation}
P^{(1)}_{i, i'} = \int_{[0, 1]} \nabla \varphi_i \cdot \nabla \varphi_{i'} \; \dd x.
\end{equation}
Likewise,
\begin{equation}
P^{(2)}_{ij, i'j'} = \int_{[0, 1]^2} \nabla  \varphi_{ij} \cdot \nabla \varphi_{i'j'} \; \dd x \, \dd y
\end{equation}
\begin{equation}
P^{(3)}_{ijk, i'j'k'} = \int_{[0, 1]^3} \nabla \varphi_{ijk} \cdot \nabla \varphi_{i'j'k'} \; \dd x \, \dd y \,\dd z.
\end{equation}
\end{definition}

\begin{lemma}
\label{lem:mass_matrix_1d_eigs}
The minimum and maximum eigenvalues of \(M^{(1)}\) are both \(\Theta(h)\), which implies that the condition number is \(\Theta(1)\).
\end{lemma}
\begin{proof}
    By a simple calculation, 
    \begin{equation}
M^{(1)}=\frac{h}{6}
\begin{bmatrix}
4 & 1 & 0 & \cdots & 0\\
1 & 4 & 1 & \ddots & \vdots\\
0 & 1 & 4 & \ddots & 0\\
\vdots & \ddots & \ddots & \ddots & 1\\
0 & \cdots & 0 & 1 & 4
\end{bmatrix}.
    \end{equation}

    The lemma follows from a well-known formula for eigenvalues of tridiagonal Toeplitz matrices \cite{Toeplitz_Noschese2006}.
\end{proof}

\begin{lemma}
\label{lem:laplacian_matrix_1d_eigs}
The minimum and maximum eigenvalues of \(P^{(1)}\) are \(\Theta \left (  h \right )\) and \(\Theta \left ( \frac{1}{h} \right )\), respectively. Thus, the condition number is \(\Theta \left ( \left (\frac{1}{h} \right )^2 \right)\).
\end{lemma}
\begin{proof}
By a simple calculation, 
\begin{equation} 
P^{(1)} = \frac 1 h 
\begin{bmatrix}
2 & -1 & 0 & \cdots & 0\\
-1 & 2 & -1 & \ddots & \vdots\\
0 & -1 & 2 & \ddots & 0\\
\vdots & \ddots & \ddots & \ddots & -1\\
0 & \cdots & 0 & -1 & 2
\end{bmatrix}.
\end{equation}
Once again, we can obtain the eigenvalues from the formula for a Toeplitz matrix \cite{Toeplitz_Noschese2006} and the lemma statement follows.
\end{proof}

\begin{lemma}
\label{lem:mass_matrix_3d_eigs}
The minimum and maximum eigenvalues of \(M^{(3)}\) are both \(\Theta(h^3)\), which implies that the condition number is \(\Theta(1)\).
\end{lemma}
\begin{proof}
    From \Cref{def:mass_matrix}, we have
    \begin{equation}
    \begin{aligned}
    M^{(3)}_{ijk, i'j'k'} &= \int_{[0, 1]^3} \varphi_{ijk}(x, y, z) \varphi_{i'j'k'} (x, y, z) \, \dd x \, \dd y \, \dd z \\
        &= \left (\int_{[0, 1]}\varphi_i(x) \varphi_{i'}(x) \, \dd{x}\right ) \left ( \int_{[0, 1]}\varphi_j(y) \varphi_{j'}(y) \, \dd{y}\right ) \left ( \int_{[0, 1]}\varphi_k(z) \varphi_{k'}(z) \, \dd{z}\right ).\\
    \end{aligned} 
    \end{equation}
    This implies 
    \begin{equation}
    M^{(3)} = M^{(1)} \otimes M^{(1)} \otimes M^{(1)}.
    \end{equation}
    Thus, using \Cref{lem:mass_matrix_1d_eigs}, we obtain the lemma statement.
\end{proof}

\begin{lemma}
\label{lem:laplacian_matrix_3d_eigs}
The minimum and maximum eigenvalues of \(P^{(3)}\) are \(\Theta \left (  h^3 \right )\) and \(\Theta \left ( h \right )\), respectively. Thus, the condition number is \(\Theta \left ( \left (\frac{1}{h} \right )^2 \right)\).
\end{lemma}
\begin{proof}
From \Cref{def:laplacian}, we have
\begin{equation}
\begin{aligned}
P^{(3)}_{ijk, i'j'k'} &= \int_{[0, 1]^3} \nabla \varphi_{ijk}(x, y, z) \cdot \nabla \varphi_{i'j'k'} (x, y, z) \, \dd x \, \dd y \, \dd z \\
&= \int_{[0, 1]^3} \left (\frac \partial {\partial x} \varphi_{ijk}(x, y, z) \frac \partial {\partial x} \varphi_{i'j'k'}(x, y, z) + \frac \partial {\partial y} \varphi_{ijk}(x, y, z) \frac \partial {\partial y} \varphi_{i'j'k'}(x, y, z) \right.\\
&\;\;\;\;\;\;\;\;\;\;\;+ \left.\frac \partial {\partial z} \varphi_{ijk}(x, y, z) \frac \partial {\partial z} \varphi_{i'j'k'}(x, y, z) \right ) \, \dd x \, \dd y \, \dd z\\
& = \int_{[0, 1]^3} \left (\varphi_j(y) \varphi_{j'}(y) \varphi_k(z) \varphi_{k'}(z) \frac \dd {\dd x} \varphi_i(x) \frac \dd {\dd{x}} \varphi_{i'}(x) \right .  \\
& \;\;\;\;\;\;\;\;\;\;\; + \varphi_i(x) \varphi_k(z) \varphi_{i'}(x) \varphi_{k'}(z) \frac d {d y} \varphi_j(y) \frac \dd {\dd y} \varphi_{j'}(y) \\
& \;\;\;\;\;\;\;\;\;\;\; + \left . \varphi_i(x) \varphi_j(y) \varphi_{i'}(x) \varphi_{j'}(y) \frac \dd {\dd z} \varphi_k(z) \frac \dd {\dd z} \varphi_{k'}(z) \right ) \, \dd x \, \dd y \, \dd z \\
& = \left (\int_{[0, 1]} \varphi_j(y) \varphi_{j'}(y) \dd{y} \right ) \left (\int_{[0, 1]} \varphi_k(z) \varphi_{k'}(z) \dd{z} \right ) \left (\int_{[0, 1]} \frac \dd {\dd x} \varphi_i(x) \frac \dd {\dd x} \varphi_{i'}(x) \dd{x} \right ) \\
 & \;\;\;\;\;\; +\left (\int_{[0, 1]} \varphi_i(x) \varphi_{i'}(x) \dd{x} \right ) \left (\int_{[0, 1]} \varphi_k(z) \varphi_{k'}(z) \dd{z} \right ) \left (\int_{[0, 1]} \frac \dd {\dd y} \varphi_j(y) \frac \dd {\dd y} \varphi_{j'}(y) \dd{y} \right ) \\
& \;\;\;\;\;\; +\left (\int_{[0, 1]} \varphi_i(x) \varphi_{i'}(x) \dd{x} \right ) \left (\int_{[0, 1]} \varphi_j(y) \varphi_{j'}(y) \dd{y} \right ) \left (\int_{[0, 1]} \frac \dd {\dd z} \varphi_k(z) \frac \dd {\dd z} \varphi_{k'}(z) \dd{z} \right ).
\end{aligned}
\end{equation}
This implies
\begin{equation}
P^{(3)} = P^{(1)} \otimes M^{(1)} \otimes M^{(1)} + M^{(1)} \otimes P^{(1)} \otimes M^{(1)} + M^{(1)} \otimes M^{(1)} \otimes P^{(1)}.
\end{equation}
Since \(P^{(1)}\) and \(M^{(1)}\) only differ by a multiplicative factor and an additive shift by the identity matrix, they have the same eigenvectors. Then, from the eigenvalues of \(P^{(1)}\) and \(M^{(1)}\) given in (\Cref{lem:mass_matrix_1d_eigs,lem:laplacian_matrix_1d_eigs}), we obtain the lemma statement.
\end{proof}

Now we are ready to prove properties of \(L\), \(A\), and \(C\). 

\begin{theorem}
The minimum and maximum eigenvalues of \(L\) (defined in \Cref{prob:discrete_problem}) are \(\Omega \left (  h^3 \right )\) and \(O \left ( h \right )\), respectively. Thus, the condition number is \(O\left ( \left (\frac{1}{h} \right )^2 \right )\).
\end{theorem}

\begin{proof}
    Let \(v\) be a unit vector in \(\mathbb R^{N}\), and \(f_v = \sum_{ijk} v_{ijk} \varphi_{ijk}\) be the corresponding function in \(V^h_0\). Then, we have   
    \begin{equation}
        \begin{aligned}
        \lambda_{\min}(L) &= \min_{v, \|v\|=1} \langle v | L | v \rangle \\ 
        &= \min_{v, \|v\|=1} \int_\Omega D(\bx) \; \nabla f_v \cdot \nabla f_v \, \dd \bx \ \\
        &\geq D_{\min} \min_{v, \|v\|=1} \int_\Omega \nabla f_v \cdot \nabla f_v \, \dd \bx \\
        &= D_{\min} \lambda_{\min}(P^{(3)}) \\
        &= \Omega(h^3),
        \end{aligned}
    \end{equation}
    where the next-to-last line follows from \Cref{lem:laplacian_matrix_3d_eigs}. In the same manner, we also have \(\lambda_{\max}(L) = O(h)\). Thus, the condition number is \(O\left ( \left (\frac{1}{h} \right )^2 \right )\).
\end{proof}

\begin{theorem}
    \label{thm:properties_of_A}
The minimum and maximum eigenvalues of \(A\) (defined in \Cref{prob:discrete_problem}) are both \(\Theta(h^3)\), which implies that the condition number is \(\Theta(1)\).
\end{theorem}
\begin{proof}
    Let \(v\) be a unit vector in \(\mathbb R^{N}\), and \(f_v = \sum_{ijk} v_{ijk} \varphi_{ijk}\) be the corresponding function in \(V^h_0\). Then, we have   
    \begin{equation}
        \begin{aligned}
        \lambda_{\min}(A) &= \min_{v, \|v\|=1} \langle v | A | v \rangle \\ 
        &= \min_{v, \|v\|=1} \int_\Omega \Sigma_a(\bx) \; f_v^2 \, \dd \bx \\
        &\geq \Sigma^{\min}_a \min_{v, \|v\|=1} \int_\Omega f_v^2 \, \dd \bx \\
        &= \Sigma^{\min}_a \lambda_{\min}(M^{(3)}) \\
        &= \Omega(h^3),
        \end{aligned}
    \end{equation}
    where the next-to-last line follows from \Cref{lem:mass_matrix_3d_eigs}. In the same manner, we also have \(\lambda_{\max}(A) = O(h^3)\). This in turn implies both \(\lambda_{\min}\) and \(\lambda_{\max}\) are \(\Theta(h^3)\) and the condition number is \(\Theta(1)\).
\end{proof}

To prove bounds on the eigenvalues and condition number of \(C\), we use the following lemma.

\begin{lemma}
\label{lem:mass_matrix_eigs}
Let \(M^t\) be the \(8 \times 8\) mass matrix on cube \(t\) of length \(h\) defined as
\begin{equation}
M^t_{ijk, i'j'k'} = \int_t \varphi^t_{ijk} \varphi^t_{i'j'k'} \dd{\mathbf x}\!,
\end{equation}
where \(i, j, k\) take values \(1\) and \(2\).
Then
\begin{enumerate}
\item the minimum eigenvalue of \(M^t\) is \(\frac {h^3} {216}\) and
\item the maximum eigenvalue of \(M^t\) is \(27 \cdot  \left (\frac {h^3}{216} \right )\).
\end{enumerate}
Thus, the condition number of \(M^t\) is \(27\).
\end{lemma}
\begin{proof}
We can directly compute
\begin{equation}
M = \frac{h^3}{216}
\begin{bmatrix}
 8 & 4 & 4 & 2 & 4 & 2 & 2 & 1 \\
 4 & 8 & 2 & 4 & 2 & 4 & 1 & 2 \\
 4 & 2 & 8 & 4 & 2 & 1 & 4 & 2 \\
 2 & 4 & 4 & 8 & 1 & 2 & 2 & 4 \\
 4 & 2 & 2 & 1 & 8 & 4 & 4 & 2 \\
 2 & 4 & 1 & 2 & 4 & 8 & 2 & 4 \\
 2 & 1 & 4 & 2 & 4 & 2 & 8 & 4 \\
 1 & 2 & 2 & 4 & 2 & 4 & 4 & 8
\end{bmatrix}.
\end{equation}
The eigenvalues can be directly computed to give the result.
\end{proof}

\begin{theorem}
    \label{thm:properties_of_C}
Let \(C\) be defined as in \Cref{prob:discrete_problem}. Then, \(C\) is block-diagonal with blocks that are either identically \(0\) or have an \(O(1)\) condition number.  
\end{theorem}

\begin{proof}
Consider the basis \(\{\ket{ijk}\}_{ijk}\) where \(i, j, k\) range from \(1\) to \(1/h - 1\). We divide this basis into two sets. Let \(B_0\) contain \( \ket{ijk}\) such that for every cube \(c_{pqr}\) incident to the node \((i, j, k)\), we have \(\nu\Sigma_f(c_{pqr}) = 0\). (Recall that \(\nu\Sigma_f\) is constant on each cube \(c_{pqr}\).) Let \(B_1\) be the set of all other basis elements.

We prove the theorem in three steps. First, we show that \(C\) is block diagonal with respect to the decomposition \(\spanop(B_0) \oplus \spanop(B_1)\). Next, we show that the block corresponding to \(B_0\) is identically \(0\), and finally, that the block corresponding to \(B_1\) has a condition number \(\kappa = O(1)\).

\emph{Step 1 (Block-diagonal structure).} Let \(u \in B_0\) and \(v \in B_1\). Let \(\phi_u = \sum_{ijk} u_{ijk} \varphi_{ijk}\) and \(\phi_v = \sum_{ijk} v_{ijk} \varphi_{ijk}\). By construction of \(B_0\) and \(B_1\), any cube \(c_{pqr}\) on which both \(\phi_u\) and \(\phi_v\) are supported must have \(\nu\Sigma_f(c_{pqr}) = 0\). Therefore, 
\begin{equation}
\langle u, C v \rangle = \langle v, C u \rangle = \int_\Omega \nu \Sigma_f \phi_u \phi_v \dd \bx = 0,
\end{equation}
implying that \(C\) is block diagonal with respect to the decomposition \(\spanop(B_0) \oplus \spanop(B_1)\).

\emph{Step 2 (The \(B_0\) block).} 
If \(u \in B_0\), then \(\phi_u\) is only supported on cubes with \(\nu\Sigma_f = 0\). Hence 
\begin{equation}
\langle u | C | u \rangle = \int_\Omega \nu \Sigma_f \phi_u^2 \dd \bx = 0.
\end{equation}

\emph{Step 3 (Bounding condition number of \(B_1\) block).}
Let \(v \in B_1\). Then we have 
\begin{equation}
\label{eq:vCv_expansion}
\begin{aligned}
\langle v \vert C \vert v \rangle = \int_\Omega \nu \Sigma_f \phi_v^2 \dd \bx  = \sum_t \int_{\Omega_t} \nu \Sigma_f(\bx) \phi_{v|t}^2 \dd \bx  = \sum_{t \;| \; \nu\Sigma_f(t) \neq 0} \int_{\Omega_t} \nu \Sigma_f(\bx) \phi_{v|t}^2 \dd \bx,
\end{aligned}
\end{equation}
where \(t\) indexes all the cubes in \(\Omega\) and \(\phi_{v|t}\) is the function \(\phi_v\) restricted to the cube region \(\Omega_t\).
Thus, we have 
\begin{equation}
\label{eq:vCv_in_terms_of_cubes}
\nu \Sigma^{\min}_f \sum_{t| \nu\Sigma_f(t) \neq 0} \int_{\Omega_t} \phi_{v|t}^2 \dd \bx \leq \langle v \vert C \vert v \rangle \leq \nu \Sigma_f^{\max}\int_\Omega \phi_v^2 \dd \bx.
\end{equation}
Now, we prove the upper and lower bounds on \(\langle v | C | v \rangle\) separately.
\begin{enumerate}
\item \emph{Upper bound:} Using \Cref{eq:vCv_in_terms_of_cubes} and \Cref{lem:mass_matrix_3d_eigs}, we have
\begin{equation}
    \label{eq:upper_bound_main_eq}
\langle v \vert C \vert v \rangle \leq \nu \Sigma^{\max}_f \int_\Omega \phi_v^2 \dd \bx \leq \nu \Sigma^{\max}_f \lambda_{\max}(M^{(3)}) \|v\|^2 = \nu \Sigma_f^{\max}\Theta(h^3) \|v\|^2.
\end{equation}
\item \emph{Lower bound:} From \Cref{eq:vCv_in_terms_of_cubes}, we have
\begin{equation}
    \label{eq:lower_bound_main_eq}
\langle v \vert C \vert v \rangle \geq \nu \Sigma^{\min}_f \sum_{t| \nu\Sigma_f(t) \neq 0} \int_{\Omega_t} \phi_{v|t}^2 \dd \bx.
\end{equation} 
From \Cref{lem:mass_matrix_eigs}, we obtain
\begin{equation}
    \int_{\Omega_t} \phi_{v|t}^2 \dd \bx \geq \lambda_{\min}(M^t) \|v_t\|^2 = \Theta(h^3) \|v_t\|^2 = \Theta(h^3) \sum_{\substack{ (ijk) \text{ where } (i, j, k) \\ \text{ incident on } t}} v_{ijk}^2,
\end{equation}
where \(v_t\) is the vector of coefficients \(v_{ijk}\) for \((i, j, k)\) incident on cube \(t\).
This gives
\begin{equation}
    \sum_{t \,|\,  \nu\Sigma_f(t) \neq 0} \int_{\Omega_t} \phi_{v|t}^2 \dd \bx \geq \Theta(h^3) \sum_{t \,|\,  \nu\Sigma_f(t) \neq 0} \sum_{\substack{ (ijk) \text{ where } (i, j, k) \\ \text{ incident on } t}} v_{ijk}^2.
\end{equation}
Rearranging so that the outer sum is over \(i, j, k\), we have
\begin{equation}
    \sum_{t| \nu\Sigma_f(t) \neq 0} \int_{\Omega_t} \phi_{v|t}^2 \dd \bx \geq \Theta(h^3) \sum_{ijk} v_{ijk}^2 \cdot \left (\text{\# of cubes } t \text{ with } \Sigma_f(t) \neq 0 \text{ incident on } (i, j, k) \right ).
\end{equation}
Because \(v \in B_1\), for every \((i, j, k)\) with \(v_{ijk} \neq 0\), there is at least one cube \(t\) incident on \((i, j, k)\) such that \(\Sigma_f(t) \neq 0\). Thus, we have
\begin{equation}
    \label{eq:C_con_number_lower_bound_single_cube}
    \sum_{t| \nu\Sigma_f(t) \neq 0} \int_{\Omega_t} \phi_{v|t}^2 \dd \bx \geq \Theta(h^3) \sum_{ijk} v_{ijk}^2 = \Theta(h^3) \|v\|^2.
\end{equation}
Using \Cref{eq:lower_bound_main_eq} and \Cref{eq:C_con_number_lower_bound_single_cube}, we have
\begin{equation}
    \label{eq:C_spectrum_lower_bound}
\langle v \vert C \vert v \rangle \geq \nu \Sigma^{\min}_f \Theta(h^3) \|v\|^2. 
\end{equation}
\end{enumerate}

Combining the upper and lower bounds, and using the fact that \(C|_{B_1}\) is Hermitian, we have
\begin{equation}
\kappa(C|_{B_1}) = \frac{\lambda_{\max}(C|_{B_1})}{\lambda_{\min}(C|_{B_1})} = O(1).
\end{equation}
This completes the proof.
\end{proof}

\section{PDE Convergence Analysis}
\label{sec:pde_convergence_analysis}

In this section, we consider how well the discretized problem approximates the underlying continuous one. {Specifically,}
we show that an eigenvalue \(\lambda_h\) of the discretized problem converges to the eigenvalue \(\lambda\) of the original problem polynomially in the mesh cell size \(h\). This helps determine how to choose \(h\) to achieve a desired accuracy \(\epsilon\) in the eigenvalue. 

We work in the context of~\cite[Chapter 8]{babuska_1991_eigenvalue_problems}, which describes the convergence of discretized eigenvalue problems. We first summarize the conditions to apply the convergence theorems of \cite{babuska_1991_eigenvalue_problems}, given in the first few pages of Chapter 8 of that reference, into a definition adapted to the Hilbert spaces of our problem setup.

Throughout this section, all big-$O$ notation is with respect to the parameter $h$.

\begin{definition}
    \label{def:babuska_osborn_conditions}
    We say that an eigenvalue problem of the form 
    \begin{equation}
    A(\phi, \psi) = \lambda B(\phi, \psi),
    \end{equation} where \(A\) and \(B\) are bilinear forms \(H^1_0(\Omega) \times H^1_0(\Omega) \rightarrow \RR\), satisfies the Babuska-Osborn conditions if the following hold:
    \begin{enumerate}
        \item \(A(\cdot, \cdot)\) is a bilinear form on \(H^1_0(\Omega) \times H^1_0(\Omega)\) and there exists some constant \(C_A  > 0\) such that
        \begin{equation}
        \abs{A(\phi, \psi)} \leq C_A \norm{\phi}_{H^1_0} \norm{\psi}_{H^1_0} \quad \forall \; \phi, \psi \in H^1_0(\Omega).
        \end{equation}
        \item Inf-sup condition 1: there exists a constant \(\alpha > 0\) such that
        \begin{equation}
        \inf_{\substack{\phi \in H^1_0(\Omega) \\ \norm{\phi}_{H^1_0} = 1}} \sup_{\substack{\psi \in H^1_0(\Omega) \\ \norm{\psi}_{H^1_0} = 1}} \abs{A(\phi, \psi)} \geq \alpha.
        \end{equation}
        \item Inf-sup condition 2: for all nonzero $\psi \in H_0^1(\Omega)$,
            \begin{equation}
            \sup_{\phi \in H^1_0(\Omega)} \abs{A(\phi, \psi)} > 0.
            \end{equation}
        \item \(B(\cdot, \cdot)\) is a bilinear form on \(H^1_0(\Omega) \times H^1_0(\Omega)\) and there exists some constant \(C_B > 0\) such that
        \begin{equation}
            \abs{B(\phi, \psi)} \leq C_B \norm{\phi}_{L^2(\Omega)} \norm{\psi}_{H^1_0(\Omega)} \quad \forall \phi, \psi \in H^1_0(\Omega).
        \end{equation}
        \item Inf-sup condition 1 in finite element space: there exists a function \(\beta \colon \RR_{>0} \to \RR_{>0}\) such that
        \begin{equation}
        \inf_{\substack{\phi \in V_0^h \\ \norm{\phi}_{H^1_0} = 1}} \sup_{\substack{\psi \in V_0^h \\ \norm{\psi}_{H^1_0} = 1}} \abs{A(\phi, \psi)} \geq \beta(h)
        \end{equation}
        for every $h > 0$.
        \item Inf-sup condition 2 in finite element space: for all nonzero $\psi \in V_0^h$,
            \begin{equation}
            \sup_{\phi \in V_0^h} \abs{A(\phi, \psi)} > 0.
            \end{equation}
        \item The finite element space can approximate all functions in \(H^1_0(\Omega)\):
            \begin{equation}
            \forall \phi \in H^1_0(\Omega), \quad \lim_{h \to 0} \beta(h)^{-1} \inf_{\chi \in V_0^h} \norm{\phi - \chi}_{H^1_0} = 0.
            \end{equation}
    \end{enumerate}
\end{definition}

We now show that a problem of the form we consider indeed satisfies these conditions.

\begin{theorem}
    \label{thm:babuska_osborn_satisfied}
    The eigenvalue problem given in \Cref{prob:orig_prob_weak_form} and its discretization in \Cref{prob:discrete_problem} satisfy the Babuska-Osborn conditions.
\end{theorem}
\begin{proof}
    We verify each of the conditions in turn.
    \begin{enumerate}
        \item We have \(a(\phi, \psi) = \int_\Omega \bigl( D(\bx) \nabla \phi \cdot \nabla \psi + \Sigma_a(\bx) \phi \psi \bigr) \dd{\bx}\). Since \(D(\bx)\) and \(\Sigma_a(\bx)\) are bounded above by constants \(D_{\max}\) and \(\Sigma_{a, \max}\), respectively, we have
        \begin{equation}
        \begin{aligned}
        \abs{a(\phi, \psi)} &\leq D_{\max} \int_\Omega \abs{\nabla \phi} \abs{\nabla \psi} \dd{\bx} + \Sigma_{a, \max} \int_\Omega \abs{\phi} \abs{\psi} \dd{\bx} \\
        &\leq D_{\max} \norm{\phi}_{H^1_0} \norm{\psi}_{H^1_0} + \Sigma_{a, \max} \norm{\phi}_{L^2} \norm{\psi}_{L^2} \\
        &\leq (D_{\max} + \Sigma_{a, \max}) \norm{\phi}_{H^1_0} \norm{\psi}_{H^1_0},
        \end{aligned}
        \end{equation}
        where we used Cauchy-Schwarz and the definition of the \(H^1_0\) norm. Thus the first condition is satisfied with \(C_a = D_{\max} + \Sigma_{a, \max}\).
        \item Since \(D(\bx)\) and \(\Sigma_a(\bx)\) are bounded below by constants \(D_{\min} > 0\) and \(\Sigma_{a, min} > 0\), respectively, we have
        \begin{equation}
        \begin{aligned}
        \inf_{\substack{\phi \in H^1_0(\Omega) \\ \norm{\phi}_{H^1_0} = 1}} \sup_{\substack{\psi \in H^1_0(\Omega) \\ \norm{\psi}_{H^1_0} = 1}} \abs{a(\phi, \psi)} &\geq \inf_{\substack{\phi \in H^1_0(\Omega) \\ \norm{\phi}_{H^1_0} = 1}} \abs{a(\phi, \phi)} \\
        &= \inf_{\substack{\phi \in H^1_0(\Omega) \\ \norm{\phi}_{H^1_0} = 1}} \int_\Omega D(\bx) \abs{\nabla \phi}^2 + \Sigma_a(\bx) \abs{\phi}^2 \dd{\bx} \\
        &\geq \inf_{\substack{\phi \in H^1_0(\Omega) \\ \norm{\phi}_{H^1_0} = 1}} \min(D_{\min}, \Sigma_{a, \min}) \norm{\phi}_{H^1_0}^2 \\
        &= \min(D_{\min}, \Sigma_{a, \min}).
        \end{aligned}
        \end{equation}
        Thus the second condition is satisfied with \(\alpha = \min(D_{\min}, \Sigma_{a, \min})\).
        \item The third condition follows by the same argument as the second condition.
        \item We have \(b(\phi, \psi) = \int_\Omega \nu \Sigma_f(\bx) \phi \psi \dd{x}\). Since \(\nu\Sigma_f(\bx)\) is bounded above by a constant \(\nu\Sigma_{f, \max}\), we have
        \begin{equation}
        \begin{aligned}
        \abs{b(\phi, \psi)} &\leq \nu \Sigma_{f, \max} \int_\Omega \abs{\phi} \abs{\psi} \dd{\bx} \\
        &\leq \nu \Sigma_{f, \max} \norm{\phi}_{L^2} \norm{\psi}_{L^2} \\
        &\leq \nu \Sigma_{f, \max} \norm{\phi}_{L^2} \norm{\psi}_{H^1_0}
        \end{aligned}
        \end{equation}
        where we used Cauchy-Schwarz and the definition of the \(H^1_0\) norm. Thus the fourth condition is satisfied with \(C_b = \nu \Sigma_{f, \max}\).
        \item The proof of the fifth condition is the same as that of the second, replacing \(H^1_0(\Omega)\) with \(V_0^h\). In particular, we may take $\beta(h) = \min(D_{\min}, \Sigma_{a, \min})$ to be a constant function.
        \item The proof of the sixth condition is the same as that of the third, replacing \(H^1_0(\Omega)\) with \(V_0^h\).
        \item The seventh condition is proved in \cite[Theorem 3.16]{ern2004theory}. \qedhere
    \end{enumerate}
\end{proof}

If these conditions are satisfied, then if \(\lambda\) is an eigenvalue of \Cref{prob:orig_prob_weak_form} with multiplicity \(m\), there exist \(m\) eigenvalues \(\lambda_h^{(1)}, \ldots, \lambda_h^{(m)}\) of \Cref{prob:finite-element-version} (counting multiplicities) such that \(\lim_{h \to 0} \lambda_h^{(i)} = \lambda\) for \(i = 1, \ldots, m\) \cite[Theorem 8.1]{babuska_1991_eigenvalue_problems}. Following \cite{babuska_1991_eigenvalue_problems}, we define the relevant eigenspaces as follows.

\begin{definition} [Eigenspaces]~ 
    \label{def:eigenspaces}
    \begin{enumerate}
    \item \(\widebar M(\lambda)\) denotes the eigenspace corresponding to the eigenvalue \(\lambda\) in \Cref{eq:weak_form}. 
    \item \(M(\lambda) \coloneqq \{ \phi \in \widebar M(\lambda) : \norm{\phi}_{H^1} = 1 \}\).
    \item \(\widebar M_h(\lambda)\) denotes the direct sum of eigenspaces corresponding to the eigenvalues \(\lambda_{h, j}\) in \Cref{eq:weak_form_fem} that converge to \(\lambda\) as \(h \to 0\).
    \item \(M_h(\lambda) \coloneqq \{ \phi \in \widebar M_h(\lambda) : \norm{\phi}_{H^1} = 1 \}\).
    \end{enumerate}
\end{definition}

The eigenvalue convergence bound relates to the quantity
\begin{equation}
\eps_h(\lambda) \coloneq \sup_{\phi \in M(\lambda)} \inf_{\psi \in V_0^h} \norm{\phi - \psi}_{H^1},
\end{equation} %
which describes how closely an eigenvector in the \(\lambda\)-eigenspace $M(\lambda)$ can be approximated by one in the finite element space $V_0^h$. We first give an asymptotic bound for $\eps_h(\lambda_{\min})$ and then apply it to give a bound on the eigenvalue convergence. (Note the difference between \(\eps\) above, and \(\epsilon\) used elsewhere in the paper.)

\begin{lemma}
    Let \(D(\bx)\) be as defined in \Cref{prob:neutron_diffusion_problem} and let \(D_{\max}\) and \(D_{\min}\) be the maximum and minimum values of \(D(\bx)\), respectively. Then the largest value \(\gamma\) such that \(\gamma \leq kD(\bx) \leq \gamma^{-1}\) for some constant \(k\) is \(\sqrt{D_{\min}/D_{\max}}\).
\end{lemma}
\begin{proof}
    We have \(\gamma \leq kD_{\min} \) and \(kD_{\max} \leq \gamma^{-1}\). Thus, \(k \geq \gamma/D_{\min}\) and \(k \leq \gamma^{-1}/D_{\max}\). Combining these two inequalities, we have
    \begin{equation}
        \frac{\gamma}{D_{\min}} \leq \frac{1}{\gamma D_{\max}} \implies \gamma^2 \leq \frac{D_{\min}}{D_{\max}}.
    \end{equation}
    Thus, the largest value of \(\gamma\) is \(\sqrt{D_{\min}/D_{\max}}\).
\end{proof}

Henceforth we take $\gamma = \sqrt{D_{\min}/D_{\max}}$.

\begin{lemma}
\label{lem:eps-bound}
For every nonzero eigenvalue $\lambda$, we have $\eps_h(\lambda) = O(h^{\gamma/(2\pi)})$.
\end{lemma}

\begin{proof}
Let $\phi_0 \in L^2(\Omega)$ be an {eigenfunction} of \Cref{prob:orig_prob_weak_form} with nonzero eigenvalue $\lambda = 1/k_0$.
Now consider the Laplacian interface problem
\begin{equation}
\int_\Omega D(\bx) \nabla \phi \cdot \nabla \psi \dd{\bx} = \int_\Omega \bracks{\parens{ \frac{1}{k_0}\nu\Sigma_f(\bx) - \Sigma_a(\bx) }\phi_0} \psi \dd{\bx}.
\end{equation}
Since $\frac{1}{k_0} \nu\Sigma_f(\bx) - \Sigma_a(\bx)$ is piecewise constant,  
\begin{equation}
\norm{\parens{\frac{1}{k_0} \nu\Sigma_f - \Sigma_a} \phi_0}_{L^2} \le \max_{\bx \in \Omega} \abs{\frac{1}{k_0} \nu\Sigma_f(\bx) - \Sigma_a(\bx)} \norm{\phi_0}_{L^2} < \infty,
\end{equation}
so $\parens{\frac{1}{k_0} \nu\Sigma_f - \Sigma_a}\phi_0 \in L^2(\Omega)$.
By \cite[Theorem 2.20]{Petzoldt2001}, the solution $\phi(\bx)$ to this interface problem
has regularity $\phi \in H^{1+\gamma/(2\pi)}(\Omega)$. But we know the solution is $\phi(\bx) = \phi_0(\bx)$, and the solution is unique by the Lax-Milgram theorem \cite{evans10_partial_differential_equations}. Thus $\phi_0 \in H^{1+\gamma/(2\pi)}(\Omega)$.

Then by \cite[Theorem 6.4]{ern-guermond}, there exists a constant $C$ and a function $\psi \in V_0^h$ such that
\begin{align}
\label{eq:ern-guermond-1}
\abs{\phi_0 - \psi}_{H^1(c_{pqr})} &\le C h^{\gamma/(2\pi)} \abs{\phi_0}_{H^{1+\gamma/(2\pi)}(B(c_{pqr}))} \\
\label{eq:ern-guermond-0}
\norm{\phi_0 -v}_{L^2(c_{pqr})} &\le C h^{1+\gamma/(2\pi)} \abs{\phi_0}_{H^{1+\gamma/(2\pi)}(B(c_{pqr}))}
\end{align}
for every cell $c_{pqr}$, where $B(c_{pqr})$ denotes the union of $c_{pqr}$ and all neighboring cells of $c_{pqr}$. 
Since $\psi$ is piecewise trilinear, each coordinate of $\nabla(\phi_0 - \psi)$ has bounded $L^2$ norm over all of $\Omega$ and
\( \abs{\phi_0 - \psi}^2_{H^1(\Omega)} = \sum_{p,q,r} \abs{\phi_0-\psi}^2_{H^1(c_{pqr})}. \)
Hence \eqref{eq:ern-guermond-1} is equivalent to
\begin{equation}
\abs{\phi_0-\psi}^2_{H^1(\Omega)} \le C^2 h^{\gamma/\pi} \sum_{p,q,r} \abs{\phi_0}^2_{H^{1+\gamma/(2\pi)}(B(c_{pqr}))}.
\end{equation}
{
Note that the fractional seminorm integrates over pairs of points $x, y \in B(c_{pqr})$, as in~\Cref{eq:fractional-seminorm}, and any such pair occurs in at most $27$ instances of $B(c_{pqr})$, as any given $x$ occurs in at most $27$ such sets. Thus
\begin{equation}
\sum_{p,q,r} \abs{\phi_0}^2_{H^{1+\gamma/(2\pi)}(B(c_{pqr}))} 
\le 27 \abs{\phi_0}^2_{H^{1+\gamma/(2\pi)}(\Omega)},
\end{equation}
}
and therefore
\begin{equation}
\abs{\phi_0 - \psi}^2_{H^1(\Omega)} \le 27 \parens{C h^{\gamma/(2\pi)} \abs{\phi_0}_{H^{1+\gamma/(2\pi)}(\Omega)}}^2.
\end{equation}
Similarly, from \eqref{eq:ern-guermond-0} we obtain
\begin{equation}
\norm{\phi_0-\psi}^2_{L^2(\Omega)} \le 27 \parens{C h^{1+\gamma/(2\pi)} \abs{\phi_0}_{H^{1+\gamma/(2\pi)}(\Omega)}}^2.
\end{equation}

Finally, we bound the $H^1$ norm of the error:
\begin{equation}
\begin{aligned}
\norm{\phi_0 - \psi}_{H^1(\Omega)} &= \parens{\norm{\phi_0 - \psi}_{L^2(\Omega)}^2 + \abs{\phi_0 - \psi}_{H^1(\Omega)}^2}^{1/2} \\
&\le \parens{\frac{1}{h} \norm{\phi_0 - \psi}_{L^2(\Omega)}^2 + \abs{\phi_0 - \psi}_{H^1(\Omega)}^2}^{1/2} \\
&\le \sqrt{54} C h^{\gamma/(2\pi)} \abs{\phi_0}_{H^{1+\gamma/(2\pi)}(\Omega)}.
\end{aligned}
\end{equation}
We know that $\abs{\phi_0}_{H^{1+\gamma/(2\pi)}(\Omega)}$ is finite as $\phi_0 \in H^{1+\gamma/(2\pi)}(\Omega)$, and since $\phi_0$ is defined independently of $h$, this norm is a constant independent of $h$. Hence $\norm{\phi_0 - \psi}_{H^1(\Omega)} = O(h^{\gamma/(2\pi)})$. 

This shows that $\inf\limits_{\psi \in V^h_0} \norm{\phi_0 - \psi}_{H^1(\Omega)} = O(h^{\gamma/(2\pi)})$. Since \(M(\lambda)\) belongs to a  finite-dimensional subspace, and \(\norm{\phi_0}_{H^1} = 1\) for all \(\phi_0 \in M(\lambda)\), from norm equivalence, the theorem follows.
\end{proof}

\begin{theorem}[Convergence of eigenvalue]
\label{thm:eigenvalue_convergence} 
    Let \(\lambda_{\min}\) be the solution of the weak form \Cref{prob:orig_prob_weak_form} and let \(\lambda_{h,\min}^{(j)}\) be one of the eigenvalues of \Cref{prob:discrete_problem} converging to \(\lambda_{\min}\). Then
    \begin{equation}
        \abs{\lambda_{\min} - \lambda_{h,\min}^{(j)}} = O \left ( h^{\gamma/\pi} \right ),
    \end{equation}
    where \(\gamma = \sqrt{D_{\min}/D_{\max}}\) and \(D_{\min}\) and \(D_{\max}\) are the minimum and maximum values of \(D(\bx)\) given in \Cref{prob:orig_prob_weak_form}. Thus, we also have \(
        \abs{k_{\max} - k_{h,\max}^{(j)}} = O \left ( h^{\gamma/\pi} \right )\) where \(k_{\max} = 1/\lambda_{\min}\) and \(k_{h,\max}^{(j)} = 1/\lambda_{h,\min}^{(j)}\).
\end{theorem}

\begin{proof}
From \cite[Theorem 8.3]{babuska_1991_eigenvalue_problems}, we have \(\abs{\lambda_{\min} - \lambda_{h,\min}^{(j)}} \le C (\beta(h)^{-1} \eps \eps^*)^{1/\alpha}\), where \(\eps \coloneq \eps_h(\lambda)\) as in \Cref{lem:eps-bound}. Here the quantity \(\alpha\) is called the \emph{ascen}t of eigenvalue \(\lambda\), which is both defined and shown to be \(1\) in our setting in \cite{babuska_1991_eigenvalue_problems}. The quantity \(\eps^*\) is defined analogously to \(\eps\) for the adjoint problem. In our case, since both \(a\) and \(b\) are symmetric, this implies \(\eps^* = \eps\). We showed in \Cref{thm:babuska_osborn_satisfied} that we may take $\beta(h)$ independent of $h$. Thus, 
\(
\abs{\lambda_{\min} - \lambda_{h,\min}^{(j)}} \le C \eps^2 = O(h^{\gamma/\pi}).
\)
\end{proof}

\section{Construction of the Preconditioner}
\label{sec:preconditioner_construction}

In this section, we construct a block encoding of the preconditioner that we use to prepare the Hamiltonian in \Cref{sec:block_encoding_hamiltonian}.

We use the preconditioner considered by \cite{deiml2025quantumrealizationfiniteelement}, which is a modification of the classical BPX preconditioner \cite{BPX_1990}. In \cite{deiml2025quantumrealizationfiniteelement}, there is no need for an explicit block encoding of the preconditioner \(F\), as that work prepares states of the form \(\ket {Fx}\) directly without separately constructing a block encoding of \(F\). However, we want to find an eigenvalue of an operator whose decomposition includes \(F\), rather than apply \(F\) to a state, so we construct an explicit block encoding. 

We use the following notation in this section.

\begin{definition}
    \label{def:num_grid_points}
    The number of grid points at level \(l\) for \(d\) dimensions is
    \begin{equation}
        n_l^{d} \coloneqq (2^l - 1)^d.
    \end{equation}
    The total number of grid points up to level \(L\) is
    \begin{equation}
        N_L^{(d)} \coloneqq \sum_{l=1}^L n_l^{d}.
    \end{equation}
\end{definition}

Unless specified otherwise, vectors are \(1\)-indexed. For a vector \(u\) of length \(b - 1\), it is convenient to define ``ghost'' nodes \(u_0 = u_{b} = 0\) (representing Dirichlet boundary conditions) in the proofs that follow. 

\subsection{Interpolation Operator Definition and Properties}
\label{sec:interpolation_operator_properties}

In order to define the preconditioner, we first describe interpolation operators and their properties, which give rise to an equivalent definition of the modified BPX preconditioner \cite{deiml2025quantumrealizationfiniteelement} and aid its construction.

An interpolation operator takes a function defined on a coarse mesh and represents it on a finer mesh. Formally, we have the following. 

\begin{definition}[Interpolation operator in \(d\) dimensions]    \label{def:Interpolation_D_dim_all_levels} 
    Let \(f \in V^h_0\) where \(h = 2^{-l}\). Let \(u \in \mathbb R^{n^{d}_l}\) be the coefficient vector of \(f\) in the nodal basis \(\{\varphi^h_{\bm{m}}\}_{\bm{m}}\) as defined in \Cref{eq:nodal_basis}, such that \(f = \sum_{\bm{m}} u_{\bm{m}} \varphi^h_{\bm{m}}\). Also consider \(h' = 2^{-l'}\) where \(l' > l\). Let \(u' \in \mathbb R^{n^{d}_{l'}}\) be the coefficient vector of the same function \(f\) in the nodal basis \(\{\varphi^{h'}_{\bm{m}}\}_{\bm{m}}\), such that \(f = \sum_{\bm{m}} u'_{\bm{m}} \varphi^{h'}_{\bm{m}}\). Then the interpolation operator in \(d\) dimensions, \(I^d_{l \rightarrow l'}\colon \mathbb R^{n^{d}_l} \rightarrow \mathbb R^{n^{d}_{l'}}\), is the linear operator such that \(I^d_{l \rightarrow l'} u = u'\).
\end{definition}

These can easily be visualized in one dimension when only moving up one level of refinement, i.e., for $l' = l+1$.

\begin{observation}[1D interpolation operator for one level] 
\label{obs:1d_interpolation_operator_one_level}
For any \(u \in \mathbb R^{n_l}\), \(I^{1}_{l \rightarrow l+1} u 
\in \mathbb R^{n_{l+1}}\) satisfies
    \begin{equation}
        \begin{aligned}
        \left (I^{1}_{l \rightarrow l+1} u \right )_{2x} &= u_x & x &\in \{1,2,\ldots,n_l\} \\
        \left ( I^{1}_{l \rightarrow l+1} u \right )_{2x+1} &= \frac {u_x + u_{x+1}} {2} & x &\in \{ 0, 1, \ldots,n_l-1\}.
        \end{aligned}
    \end{equation} 
In matrix form, 
\begin{equation}
I^{1}_{l\to l+1}
=
\begin{bmatrix}
1/2 & 0   & 0   & \cdots & 0 \\
1   & 0   & 0   & \cdots & 0 \\
1/2 & 1/2 & 0   & \cdots & 0 \\
0   & 1   & 0   & \cdots & 0 \\
0   & 1/2 & 1/2 & \ddots & \vdots \\
\vdots & \ddots & \ddots & \ddots & 0 \\
0 & \cdots & 0 & 1   & 0 \\
0 & \cdots & 0 & 1/2 & 1/2 \\
0 & \cdots & 0 & 0   & 1 \\
0 & \cdots & 0 & 0   & 1/2
\end{bmatrix}
\in \mathbb{R}^{n_{l+1}\times n_l}.
\end{equation}
\end{observation}

\begin{observation}[Product of interpolation operators] 
\label{obs:product_of_interpolation_operators}
For any \(L > l\),
    \begin{equation}
        I^d_{l \rightarrow L} = \left (I^d_{L-1 \rightarrow L} \right ) \left (I^d_{L-2 \rightarrow L-1} \right )\cdots \left (I^d_{l \rightarrow l+1} \right ).       
    \end{equation} 
\end{observation}

\begin{lemma}
    \label{lem:Spectral_norm_I1}
    The spectral norm of the interpolation operator satisfies \(\Vert I^1_{l \rightarrow l+1} \Vert \leq \sqrt{2}\).
\end{lemma}
\begin{proof}
    We have 
    \begin{equation}
        \Vert I^1_{l \rightarrow l+1} \Vert = \max_{f \text{ s.t. } \Vert f \Vert = 1} \Vert I^1_{l \rightarrow l+1} f \Vert.
    \end{equation}
    
    From \Cref{obs:1d_interpolation_operator_one_level}, we have
    \begin{equation}
        \begin{aligned}
        \Vert I^1_{l \rightarrow l+1} f \Vert^2 &= \sum_{i=0}^{2^l - 1} f_i^2 + \sum_{i=0}^{2^l - 1} \left (\frac {f_i + f_{i+1}} {2} \right )^2 \\
        &= \sum_{i=0}^{2^l - 1} f_i^2 + \frac{1}{4} \sum_{i=0}^{2^l - 1} \left (f_i + f_{i+1} \right )^2 \\ 
        & \leq  \sum_{i=0}^{2^l - 1} f_i^2 + \frac{1}{2} \sum_{i=0}^{2^l - 1} \left (f_i^2 + f_{i+1}^2 \right ) \\
        &\leq \sum_{i=1}^{2^l - 1} f_i^2 + \sum_{i=1}^{2^l - 1} f_i^2 \\
        &= 2 \Vert f \Vert^2 \\
        &= 2
        \end{aligned}
    \end{equation}
    where the third line uses \((a+b)^2 \leq 2(a^2 + b^2)\).
    Thus, the lemma statement follows.
\end{proof}

\begin{lemma}
    \label{lem:Interpolation_operator_tensor_product} 
    For any dimension $d$,
    \begin{equation}
        I^d_{l \rightarrow l'} = \bigotimes_{i=1}^d I^1_{l \rightarrow l'}.
    \end{equation}
\end{lemma}
\begin{proof}
    We prove the statement for \(d=2\); the proof for arbitrary \(d\) is a straightforward generalization. Let the hat functions \(\varphi^h_i\) and \(\varphi^h_{ij}\) when written in terms of the nodal basis with mesh size \(h'\), where \(h = 2^{-l}\) and \(h' = 2^{-l'}\). be given by
    \begin{equation}
        \label{eq:phi_h_in_terms_of_phi_h_prime}
        \begin{aligned}
            \varphi^h_{i} &= \sum_{p} c^{i}_{p} \varphi^{h'}_{p} \\
            \varphi^h_{ij} &= \sum_{pq} c^{ij}_{pq} \varphi^{h'}_{pq}.
        \end{aligned}
    \end{equation}

    From \Cref{def:Interpolation_D_dim_all_levels}, we have
    \begin{equation}
        \label{eq:I1_I2_coefficients}
        \begin{aligned}
        I^1_{l \rightarrow l'}[p, i] &= c^{i}_p \\
        I^2_{l \rightarrow l'}[pq, ij] &= c^{ij}_{pq}.
        \end{aligned}
    \end{equation}

    Starting with \Cref{eq:phi_h_in_terms_of_phi_h_prime} and the definition of nodal basis \Cref{eq:nodal_basis}, we have
    \begin{equation}
        \begin{aligned}
            \varphi^h_{ij}(x, y) &= \varphi^h_i(x) \varphi^h_j(y) \\
            &= \left (\sum_p c^i_p \varphi^{h'}_p(x) \right ) \left (\sum_q c^j_q \varphi^{h'}_q(y) \right ) \\
            &= \sum_{pq} c^i_p c^j_q \varphi^{h'}_p(x) \varphi^{h'}_q(y) \\
            &= \sum_{pq} c^i_p c^j_q \varphi^{h'}_{pq}(x, y).
        \end{aligned}
    \end{equation}

    Thus from \Cref{eq:I1_I2_coefficients}, we have
    \begin{equation}
        I^2_{l \rightarrow l'}[pq, ij] = I^1_{l \rightarrow l'}[p, i] I^1_{l \rightarrow l'}[q, j],
    \end{equation}
    which proves the lemma statement for \(d=2\). 
\end{proof}

We now introduce some operator embeddings (various ways of zero-padding an operator) that aid in the construction of the preconditioner. 

\begin{definition}[Operator embeddings] 
\label{def:operator_embeddings}
In the following, we say that \(M\) is a \(P \times Q\) block matrix if it has \(P\) blocks in every column and \(Q\) blocks in every row. We write \(M[p, q]\) to denote the \(q\)th block in the \(p\)th row, using \(1\)-based indexing where \(p \in \{1, \ldots, P\}\) and \(q \in \{1, \ldots, Q\}\). The sizes of the blocks within the matrix can be variable, and will be explicitly mentioned when relevant. We define the following \emph{zero embeddings} of operators: 
\begin{enumerate}
\item For a matrix \(A \in \mathbb R^{m \times n}\), define the \(2 \times 2\) matrix \(A' \in \mathbb R^{m' \times m'}\) where \(A'[1, 1] = A\) and all other blocks are \(0\). The value \(m'\) will be clear from context. In other words,
\begin{equation} 
A' = \begin{bmatrix} 
A & 0 \\ 
0 & 0 
\end{bmatrix}.
\end{equation}

\item For an interpolation operator \(I^d_{l \rightarrow L} \in \mathbb R^{n^d_L \times n^d_l}\) (\Cref{def:Interpolation_D_dim_all_levels}), we define the \(1 \times L\) matrix \(I^{d''}_{l \rightarrow L} \in \mathbb R^{n^d_L \times N^{(d)}_L}\) such that \(I^{d''}_{l \rightarrow L}[1, s] \in \mathbb R^{n^d_L \times n^d_s}\) where \(s \in \{1, \ldots, L\}\). Moreover, \(I^{d''}_{l \rightarrow L}[1, l] = I^d_{l \rightarrow L}\) and all other blocks are \(0\). In other words,
\begin{equation}
I^{d''}_{l \rightarrow L} = 
\begin{bmatrix}
0 & 0 & \cdots & I^{d}_{l \rightarrow L} & \cdots & 0
\end{bmatrix}.
\end{equation}

\item For an interpolation operator \(I^d_{l \rightarrow l'} \in \mathbb R^{n^d_{l'} \times n^d_{l}}\) (\Cref{def:Interpolation_D_dim_all_levels}), we define the \(L \times L\) matrix \(\hat I^{d}_{l \rightarrow {l'}} \in \mathbb R^{N^{(d)}_L \times N^{(d)}_L}\) such that \(\hat I^{d}_{l \rightarrow L}[s, s'] \in \mathbb R^{n^d_{s} \times n^d_{s'}}\). Moreover, \(\hat I^{d}_{l \rightarrow l'}[l', l] = I^d_{l \rightarrow l'}\) and all other blocks are \(0\). In other words,
\begin{equation}
\widehat I^{d}_{l \rightarrow l'} \;=\;
\begin{bmatrix}
0 & \cdots & 0 & \cdots & 0 \\
\vdots & \ddots & \vdots &        & \vdots \\
0 & \cdots & I^{d}_{l \rightarrow l'} & \cdots & 0 \\
\vdots &        & \vdots & \ddots & \vdots \\
0 & \cdots & 0 & \cdots & 0
\end{bmatrix}.
\end{equation}
\end{enumerate}
\end{definition}

\begin{observation}[Product of zero embeddings]
\label{obs:product_hat_embeddings}
\begin{equation}
\hat I^d_{l' \rightarrow l''} \cdot \hat I^d_{l \rightarrow l'} = \hat I^d_{l \rightarrow l''}.
\end{equation}
\end{observation}

\begin{lemma}
    \label{lem:interpolation_embedding_spectral_norm}
    The spectral norm of \(I^{d''}_{l \rightarrow L}\) (\Cref{def:operator_embeddings}) is at most \(2^{d(L-l)/2}\).
\end{lemma}
\begin{proof} 
    The spectral norm of the interpolation operator embedding \(I^{d''}_{l \rightarrow L}\) is equal to the spectral norm of \(I^d_{l \rightarrow L}\) (\Cref{def:Interpolation_D_dim_all_levels}).
    Thus, from \Cref{lem:Interpolation_operator_tensor_product}, \Cref{obs:product_of_interpolation_operators}, and \Cref{lem:Spectral_norm_I1}, we have
    \begin{equation}
        \begin{aligned}
        \Vert I^{d''}_{l \rightarrow L} \Vert &= \Vert I^d_{l \rightarrow L} \Vert \\
        &= \left \Vert \bigotimes_{i=1}^d I^1_{l \rightarrow L} \right \Vert \\
        &= \prod_{i=1}^d \Vert I^1_{l \rightarrow L} \Vert \\
        &= \prod_{i=1}^d \norm{\prod_{j=l}^{L-1}  I^1_{j \rightarrow j+1}} \\
        &\leq \prod_{i=1}^d \left ( \prod_{j=l}^{L-1} \sqrt{2} \right ) \\
        &= 2^{d(L-l)/2}.
        \end{aligned}
    \end{equation}
    
\end{proof}

\subsection{Preconditioner Definition and  Properties}
\label{sec:preconditioner_properties}

We are now ready to define the modified BPX preconditioner in terms of interpolation operators (\Cref{def:operator_embeddings}).
\begin{definition}
    \label{def:BPX_operator_F}
    The modified BPX preconditioner \(F^d_L\) for \(d\) dimensions and \(L\) levels is \cite[p.~12]{deiml2025quantumrealizationfiniteelement}
    \begin{equation}
        F^d_L  = \sum_{l=1}^L 2^{-l(2-d)/2} I^{d''}_{l \rightarrow L}.
    \end{equation}
\end{definition}

Using the properties of interpolation operators, we can bound the spectral norm of the preconditioner \(F\) as follows.
\begin{theorem}
    \label{thm:spectral_norm_of_F}
    The spectral norm of the modified BPX preconditioner satisfies \(\norm{F^d_L} = O \left ( \left (\frac{1}{h}\right )^{d/2} \right )\), where \(h = 2^{-L}\) is the mesh spacing at level \(L\).
\end{theorem}
\begin{proof}
    From \Cref{def:BPX_operator_F}, we have
    \begin{equation}
        F^d_L  = \sum_{l=1}^L 2^{-l(2-d)/2} I^{d''}_{l \rightarrow L}.
    \end{equation}
    
    Using \Cref{lem:interpolation_embedding_spectral_norm}, we have
    \begin{equation}
        \norm{I^{d''}_{l \rightarrow L}} \leq 2^{d(L-l)/2}.
    \end{equation}
    Thus, we have
    \begin{equation}
    \begin{aligned}
        \norm{F^d_L} &\leq \sum_{l=0}^L \left (2^{-l(2-d)/2} \right ) \left ( 2^{d(L-l)/2} \right )\\
        &= O(2^{dL/2})\sum_{l=1}^L 2^{-l} \\
        &= O\left ( \left (\frac{1}{h}\right )^{d/2} \right )
    \end{aligned}
    \end{equation}
as claimed.
\end{proof}

It is also useful to define an embedding for \(F^d_L\) following those of the interpolation operators in \Cref{def:operator_embeddings}. 

\begin{definition}
    \label{def:BPX_operator_F_embedding}
Define \(\widehat F^d_L \in \mathbb{R}^{\,N^{(d)}_L \times N^{(d)}_L}\) by
\[
\widehat F^d_L =
\begin{bmatrix}
\mathbf{0}_{\left (N_L^{(d)}-n^{d}_L \right )\times N^{(d)}_L}\\[2mm]
F^d_L
\end{bmatrix},
\]
i.e., \(\widehat F^d_L\) is zero in its first \(N_L^{(d)}-n^{(d)}_L\) rows and coincides with \(F^d_L\) (\Cref{def:BPX_operator_F}) in its last \(n^{(d)}_L\) rows.
\end{definition}

\subsection{Block Encodings of Interpolation Operators}
\label{sec:block_encoding_interpolation_operators}

We now construct block encodings of \(\parens{\hat I ^{d}_{l \rightarrow l+1}}'\), where the hat and the prime mean we are using two embeddings from \Cref{def:operator_embeddings}]. These are then used to construct a block encoding of the preconditioner \(\hat F^d_L\) in \Cref{sec:block_encoding_preconditioner}.

Our starting point is a block encoding of the one-dimensional interpolation operator for one level, \(\parens{I^1_{l \rightarrow l+1}}'\) (\Cref{obs:1d_interpolation_operator_one_level}).

Since this is a sparse matrix, the first idea might be to use the standard sparse-matrix block-encoding construction from classical oracles that provide matrix entries as in \cite[Lemma 48]{Gilyen_19_QSVT_and_beyond}. However, this gives us a block-encoding factor \(\alpha = \sqrt{s_r s_c}\), where \(s_r\) and \(s_c\) are the maximum number of non-zero entries in any row or column, respectively. In our case, this gives \(\alpha = \sqrt 6\). 

Instead, we directly construct quantum row and column oracles as considered in \cite[Lemma 47]{Gilyen_19_QSVT_and_beyond}. With this approach, we obtain a block-encoding factor \(\alpha = \sqrt 2\), equal to the upper bound we showed on the spectral norm of the operator in \Cref{lem:Spectral_norm_I1}.

This is only a constant-factor improvement in block-encoding factor. However, in \Cref{sec:block_encoding_preconditioner}, we use this block encoding as a subroutine to construct the block encoding of \(\parens{\hat I^d_{l \rightarrow L}}'\), which results in raising the block-encoding factor to the power \(dL\). This gives an overall block-encoding factor of \(2^{dL/2} = \parens{\frac{1}{h}}^{d/2}\), as compared with \(2^{1.3 dL} = \parens{\frac{1}{h}}^{1.3 d}\) using the standard construction, significantly reducing the overall complexity. 

First, we restate Lemma 47 of \cite{Gilyen_19_QSVT_and_beyond}.

\begin{lemma}
   
    \label{lem:block_encoding_row_col_oracle}
    Consider an arbitrary matrix \(A \in \mathbb{C}^{N_r \times N_c}\) with entries \(A_{jk} = e^{i\phi_{jk}}\vert A_{jk} \vert\). Let \(\max(N_r, N_c) \leq m \coloneqq 2^b\) for some integer \(b\). Let \(A'\) be an \(m \times m\) embedding of \(A\) (\Cref{def:operator_embeddings}). Assume \(0\)-indexing for the matrix entries. 
    
    Let \(r\) and \(c\) be \(b\)-qubit registers, \(r_{\slack}\) and \(c_{\slack}\) be \(1\)-qubit registers, and \(\anc\) be a \(q\)-qubit register for some \(q\).

    Suppose we are given a ``column oracle'' \(P\) that acts as
    \begin{equation}
        P\colon \ket{0}_r \ket{0}_{r_{\slack}} \ket{k}_c \ket{0}_{\anc}  \mapsto 
        \left ( \sum_{j = 0}^{N_r-1} e^{i\phi_{jk}}\sqrt{\frac{\vert A_{jk}\vert}{c_{\max}}} \ket{j}_r \ket 0_{r_{\slack}} + 
        \sqrt{1 - \frac{c_k}{c_{\max}}}\ket{0}_r \ket{1}_{r_{\slack}} \right ) \ket{k}_c \ket{0}_{\anc}
    \end{equation} when \(k < N_c\) and 
      \begin{equation}
        P\colon \ket{0}_r \ket{0}_{r_{\slack}} \ket{k}_c \ket{0}_{\anc}  \mapsto 
        \ket{0}_r \ket{0}_{r_{\slack}} \ket{k}_c \ket{1}_{\anc}
    \end{equation} otherwise,
    and a ``row oracle'' \(Q\) that acts as
    \begin{equation}
        Q\colon \ket{0}_c \ket{0}_{c_{\slack}} \ket{j}_r \ket{0}_{\anc}  \mapsto 
        \left ( \sum_{k = 0}^{N_c-1} \sqrt{\frac{\vert A_{jk}\vert}{r_{\max}}} \ket{k}_c \ket 0_{c_{\slack}} + 
        \sqrt{1 -\frac{r_j}{r_{\max}}}\ket{0}_c \ket{1}_{c_{\slack}} \right ) \ket j_r \ket{0}_{\anc}
    \end{equation}
    when \(j < N_r\) and
      \begin{equation}
        Q\colon \ket{0}_c \ket{0}_{c_{\slack}} \ket{j}_r \ket{0}_{\anc}  \mapsto 
        \ket{0}_c \ket{0}_{c_{\slack}} \ket{j}_r \ket{1}_{\anc}
    \end{equation} otherwise,
    where \(c_k = \sum_{j=0}^{N_r-1} \vert A_{jk} \vert \), \(r_j = \sum_{k=0}^{N_c-1} \vert A_{jk} \vert \), \(c_{\max} = \max_k c_k\), and \(r_{\max} = \max_j r_j\). Then \(\swapop_{r,c}Q^\dagger P\) is a \((\sqrt{r_{\max} c_{\max}}, q+b+2, 0) \)-block encoding of \(A'\).
\end{lemma}
\begin{proof}
    To show that \(U = \swapop_{r,c}Q^\dagger P\) is an \((\alpha, q+b+2, 0) \)-block encoding of \(A'\), it suffices to show that \(U \ket{0} \otimes \ket{\psi} = \ket{0} \otimes A'/\alpha \ket{\psi} + \ket{\perp} \) where \(\ket{0}\) is a \((q+b+2)\)-qubit zero state and  \(\ket{\perp}\) is supported on states of the form \(\ket w \otimes \ket {g_w}\) where \(\ket w\) is orthogonal to \(\ket 0\) and \(\ket {g_w}\) is any vector subnormalized such that the right-hand side is a normalized state \cite{Gilyen_19_QSVT_and_beyond}.

    We begin with the state
    \begin{equation}
        \label{eq:initial_state_block_encoding}
        \ket{0}_r \ket{0}_{r_{\slack}} \ket{\psi}_c \ket{0}_{c_{\slack}} \ket{0}_{\anc}
    \end{equation}
    where \(\ket{\psi}_c = \sum_{k=0}^{m-1} \psi_k \ket{k}_c\).

    Applying \(P\) to \Cref{eq:initial_state_block_encoding}, we have 
    \begin{equation}
        \label{eq:after_P_block_encoding}
        \begin{aligned}
        &\sum_{k=0}^{N_c - 1} \psi_k \left ( \sum_{j = 0}^{N_r-1} e^{i\phi_{jk}}\sqrt{\frac{\vert A_{jk} \vert}{c_{\max}}} \ket{j}_r \ket 0_{r_{\slack}} + \sqrt{1 - \frac{c_k}{c_{\max}}}\ket{0}_r \ket{1}_{r_{\slack}} \right ) \ket k_c \ket{0}_{c_{\slack}} \ket{0}_{\anc} \\
        &+ \left ( \cdots \right )_{r,  r_{\slack}, c, c_{\slack}} \ket{1}_{\anc},
        \end{aligned}
    \end{equation}
    where we write \((\cdots)\) to denote irrelevant terms on the remaining registers, which are subnormalized such that the overall state is normalized. We can split this as
    \begin{equation}
        \begin{aligned}
        &\sum_{k=0}^{N_c-1} \psi_k \left ( \sum_{j = 0}^{N_r-1} e^{i\phi_{jk}}\sqrt{\frac{\vert A_{jk}\vert}{c_{\max}}} \ket{j}_r  \ket k_c \ket{0}_{c_{\slack}} \ket{0}_{\anc} \right ) \ket 0_{r_{\slack}}  \\
        &+ \sum_{k=0}^{N_c-1} \psi_k \left ( \sqrt{1 - \frac{ c_k}{c_{\max}}}\ket{0}_r  \ket k_c \ket{0}_{c_{\slack}} \ket{0}_{\anc} \right ) \ket{1}_{r_{\slack}} \\
        &+ \left ( \cdots \right )_{r, r_{\slack}, c, c_{\slack}} \ket{1}_{\anc}.
        \end{aligned}
    \end{equation}

    Now, applying \(Q^\dagger\) to the above state, we obtain 
    \begin{equation}
        \begin{aligned}
        &\left(\sum_{k=0}^{N_c-1} \psi_k \sum_{j=0}^{N_r-1} e^{i\phi_{jk}} \sqrt{\frac{\vert A_{jk}\vert}{c_{\max}}} \sqrt{\frac{\vert A_{jk}\vert}{r_{\max}}} \ket{j}_r \ket{0}_{c_{\slack}} \ket{0}_{\anc} \right) \ket{0}_c \ket{0}_{r_{\slack}}  \\
        &+ \sum_{k = 1}^{N_c-1} \left(\cdots \right)_{r, c_{\slack}, \anc} \ket{k}_c \ket{0}_{r_{\slack}} \\
        &+ \left( \cdots \right)_{r, c, c_{\slack}, \anc} \ket{1}_{r_{\slack}} \\
        &+ \left ( \cdots \right )_{r, c, r_{\slack}, c_{\slack}} \ket{1}_{\anc}.
        \end{aligned}
    \end{equation}

    Swapping the \(r\) and \(c\) registers, we have the state
    \begin{equation}
        \begin{aligned}
        &\ket{0}_r \ket{0}_{r_{\slack}} \ket{0}_{c_{\slack}} \ket{0}_{\anc} \otimes \left (
            \frac{A'}{\sqrt{r_{\max} c_{\max}}} \ket{\psi}_c \right ) \; + \\
            &\parens{ \sum_{k = 1}^{N_c-1} \left(\cdots \right)_{r, c_{\slack}, \anc} \ket{k}_r \ket{0}_{r_{\slack}} 
        + \left( \cdots \right)_{r, c, c_{\slack}, \anc} \ket{1}_{r_{\slack}} 
        + \left ( \cdots \right )_{r, c, r_{\slack}, c_{\slack}} \ket{1}_{\anc}}
        \end{aligned}
    \end{equation} where the second term is the desired \(\ket{\perp}\) vector. This proves the result.
\end{proof}

Next, we construct the oracles \(P\) and \(Q\) used in \Cref{lem:block_encoding_row_col_oracle} for \(\parens{I^1_{l \rightarrow l+1}}'\).

\begin{lemma}
    \label{lem:preparing_P_for_I}
    Let \(P, m, b, r, c, r_{\slack}, c_{\slack}, c_{\max}, q\) be as defined in \Cref{lem:block_encoding_row_col_oracle}. The oracle \(P\) for the \(m \times m\) matrix \(\parens{I^1_{l \rightarrow l+1}}'\) (\Cref{def:operator_embeddings}), where \(m = 2^{L+1}\), can be implemented using \(O(L)\) gates and \(O(L)\) ancillas.
\end{lemma}
\begin{proof}
    We give a step-by-step procedure to implement \(P\).
    We begin with the state
    \begin{equation}
       \ket{0}_r  \ket{k}_c \ket{00}_{\anc_1}\ket{0 \ldots 0}_{\anc_2}\ket{0}_{\anc_3} \ket{0}_{r_{\slack}}.
    \end{equation}
    Controlling on whether index \(k\) is greater than equal to \(2^{l} - 1\), which can be done with comparators \cite{yuan2023improvedqftbasedquantumcomparator} in \(O(L)\) gates and ancillas, we flip \(\anc_3\), giving
    \begin{equation}
       \ket{0}_r  \ket{k}_c \ket{00}_{\anc_1}\ket{0 \ldots 0}_{\anc_2} \ket{1}_{\anc_3}\ket{0}_{r_{\slack}}.
    \end{equation}
    Otherwise, we proceed as follows.
    First, we implement \(2k\) in the \(r\) register, giving 
    \begin{equation}
       \ket{2k}_r  \ket{k}_c \ket{00}_{\anc_1}\ket{0 \ldots 0}_{\anc_2} \ket{0}_{\anc_3}
       \ket{0}_{r_{\slack}}.
    \end{equation}
    Then we construct the following ancilla state:
    \begin{equation}
         \ket{2k}_r  \ket{k}_c \left ( \sqrt{\frac{1}{4}} \ket{00}_{\anc_1} + \sqrt{\frac{1}{2}} \ket{01}_{\anc_1} + \sqrt{\frac{1}{4}} \ket{11}_{\anc_1} \right ) \ket{0 \ldots 0}_{\anc_2}
         \ket{0}_{\anc_3} \ket{0}_{r_{\slack}}.
    \end{equation}
    Controlled on \(\anc_1\), we apply increment operations to the \(r\) register using \(O(L)\) gates and ancilla qubits \cite{cuccaro2004newquantumripplecarryaddition}  to obtain
    \begin{equation}
        \begin{aligned}
            &\left ( \sqrt{\frac 1 4} \ket{2k}_r  \ket{00}_{\anc_1} + \sqrt{\frac 1 2} \ket{2k+1}_r\ket{01}_{\anc_1}  + \sqrt{\frac 1 4} \ket{2k+2}_r \ket{11}_{\anc_1}  \right ) \\
            & \otimes \ket k_c \ket{0 \ldots 0}_{\anc_2} \ket{0}_{\anc_3} \ket{0}_{r_{\slack}}.
        \end{aligned}
    \end{equation}
    Finally, we uncompute the \(\anc_1\) register. This is done in two steps. First, we use \(\ket k_c\) to prepare \(\ket{2k+1}\) in the \(\anc_2\) register. Controlled on whether the \(r\) register matches the \(\anc_2\) register, we apply a gate converting \(\ket{01}_{\anc_1} \mapsto \ket{00}_{\anc_1}\). Then we uncompute the \(\anc_2\) register. The same procedure is applied for \(\ket{2k+2}_r\).
    This gives us
    \begin{equation}
        \begin{aligned}
            &\left ( \sqrt{\frac 1 4} \ket{2k}_r + \sqrt{\frac 1 2} \ket{2k + 1}_r  + \sqrt{\frac 1 4} \ket{2k+2}_r  \right ) \\
            &\otimes \ket k_c \ket{00}_{\anc_1} \ket{0 \ldots 0}_{\anc_2} \ket{0}_{\anc_3} \ket{0}_{r_{\slack}}
        \end{aligned}
    \end{equation}
    which is the desired state, since \(c_{\max} = 2\) for the \(\parens{I^1_{l \rightarrow l+1}}'\) matrix. 
\end{proof}

\begin{lemma}
    \label{lem:preparing_Q_for_I}
    Let \(Q, m, b, r, c, r_{\slack}, c_{\slack}, r_{\max}, c_{\max} , q, r_j, c_k\) be as defined in \Cref{lem:block_encoding_row_col_oracle}. The oracle \(Q\) for the \(m \times m\) matrix \(\parens{I^1_{l \rightarrow l+1}}'\) (\Cref{def:operator_embeddings}) where \(m = 2^{L+1}\) can be implemented using \(O(L)\) gates and \(O(L)\) ancillas.
\end{lemma}
\begin{proof}
    We give a step-by-step procedure to implement the oracle \(Q\). We start with the state
    \begin{equation}
        \label{eq:initial_state_Q_for_I}
         \ket{0}_c  \ket{j}_r \ket{00}_{\anc_1}\ket{0 \ldots 0}_{\anc_2} \ket{0}_{c_{\slack}}
    \end{equation}

    The first observation is that rows \(0\) and \(2^{l+1}-2\) have \(r_j = \sum_{k = 0}^{N_c-1} \vert {I^1_{l \rightarrow l+1}}_{jk} \vert = 1/2\), while all other rows have \(r_j = 1\). Thus, the amplitude in the \(\ket{1}_{c_{\slack}}\) state is only non-zero for these two rows. 
    
    We can handle these two cases separately. Controlling on whether \(\ket{j}_r = \ket{0}_r\), we prepare the following from \Cref{eq:initial_state_Q_for_I}:
    \begin{equation}
         \left ( \sqrt{\frac 1 2} \ket{0}_c \ket{0}_{c_{\slack}} + \sqrt{\frac 1 2} \ket{0}_c \ket{1}_{c_{\slack}} \right ) \ket{j}_r \ket{00}_{\anc_1}\ket{0 \ldots 0}_{\anc_2} .
    \end{equation}

    Similarly, controlling on whether \(\ket{j}_r = \ket{2^{l+1}-2}_r\), we prepare
    \begin{equation}
         \left ( \sqrt{\frac 1 2} \ket{2^l - 2}_c \ket{0}_{c_{\slack}} + \sqrt{\frac 1 2} \ket{0}_c \ket{1}_{c_{\slack}} \right ) \ket{j}_r \ket{00}_{\anc_1}\ket{0 \ldots 0}_{\anc_2} .
    \end{equation}

    For all other rows, controlled on \(j\) being odd, we prepare
    \begin{equation}
            \ket{(j-1)/2}_c \ket{j}_r \ket{00}_{\anc_1}\ket{0 \ldots 0}_{\anc_2} \ket{0}_{c_{\slack}}.
    \end{equation} 

    Controlled on \(j\) being even, we prepare
    \begin{equation}
            \left ( \sqrt{\frac 1 2} \ket{(j/2 - 1}_c + \sqrt{\frac 1 2} \ket{j/2}_c \right ) \ket{j}_r \ket{00}_{\anc_1}\ket{0 \ldots 0}_{\anc_2} \ket{0}_{c_{\slack}},
    \end{equation}
    which can be accomplished using \(O(L)\) gates and ancilla qubits \cite{cuccaro2004newquantumripplecarryaddition}.

    This completes the construction of the oracle \(Q\) using \(O(L)\) gates and ancilla qubits.
\end{proof}

We are now ready to construct the block encoding of the \(d\)-dimensional interpolation operator for one level, \(\parens{I^d_{l \rightarrow l+1}}'\).

\begin{lemma}
    \label{lem:block_encoding_I_d_l_to_l+1}
    A \((2^{d/2}, O(dL), 0)\)-block encoding of the \(m^d \times m^d\) matrix \(\parens{I^d_{l \rightarrow l+1}}'\) (\Cref{def:operator_embeddings}) can be constructed using \(O(dL)\) gates, where \(m = 2^{(L+1)}\).
\end{lemma}
\begin{proof}
    Using \Cref{lem:block_encoding_row_col_oracle}, we can construct a \((\sqrt{2}, O(L), 0)\)-block encoding of \(\parens{I^{1}_{l \rightarrow l+1}}'\) using \(\poly(L)\) gates where the oracles for \(P\) and \(Q\) are constructed using \Cref{lem:preparing_P_for_I} and \Cref{lem:preparing_Q_for_I}, respectively. Calling this block encoding \(U_1\), we have
    \begin{equation}
        \parens{I^{1}_{l \rightarrow l+1}}' = \sqrt 2 (\bra{0}^{\otimes O(L)} \otimes I) U_1 (\ket{0}^{\otimes O(L)} \otimes I) .
    \end{equation}

    Using \Cref{lem:block_encoding_row_col_oracle}, we have
    \begin{equation}
        \begin{aligned}
            {\left ( \parens{I^{1}_{l \rightarrow l+1}}' \right )}^{\otimes d}&= (\sqrt 2)^d (\bra{0}^{\otimes O(dL)} \otimes I) U_1^{\otimes d} (\ket{0}^{\otimes O(dL)} \otimes I) \\
            &= 2^{d/2} (\bra{0}^{\otimes O(dL)} \otimes I) \swapop_{r,c}^{\otimes d} (Q^\dagger)^{\otimes d} P^{\otimes d} (\ket{0}^{\otimes O(dL)} \otimes I), \\
        \end{aligned}
    \end{equation}
    where the ancillas have been grouped together on the right-hand side. Thus, \(\swapop_{r, c}^{\otimes d} (Q^\dagger)^{\otimes d} P^{\otimes d}\) is a \((2^{d/2}, O(dL), 0)\)-block encoding of \({\left ( \parens{I^{1}_{l \rightarrow l+1}}' \right )}^{\otimes d}\).
    
    It remains to obtain a block encoding of \(\parens{I^{d}_{l \rightarrow l+1}}'\) from the block encoding of 
  \({\left ( \parens{I^{1}_{l \rightarrow l+1}}' \right )}^{\otimes d}\). Since \(I^d_{l \rightarrow l+1} = {I^{1}_{l \rightarrow l+1}}^{\otimes d} \) (\Cref{lem:Interpolation_operator_tensor_product}), we can apply permutations to obtain this using \Cref{lem:tensor_product_lemma_zpbe} and \Cref{cor:tensor_product_arbitrary_d} below.
\end{proof}

While a zero embedding of a tensor product is not equal to the corresponding tensor product of zero embeddings, they are related by permutations. The following lemma formalizes this, and the corollary extends the result to arbitrary dimensions.

\begin{lemma}
    \label{lem:tensor_product_lemma_zpbe}
    Let \(A'\) be an \(M_a \times N_a\) matrix having an \(m_a \times n_a\) submatrix \(A\) in its top-left corner. Let \(B'\) be an \(M_b \times N_b\) matrix having an \(m_b \times n_b\) submatrix \(B\) in its top-left corner. Let \( \parens{A \otimes B}'\) be an \(M_a M_b \times N_a N_b\) matrix having an \(m_a m_b \times n_a n_b\) submatrix \(A \otimes B\) in its top-left corner. Let \(N_{\max} = \max(M_a, N_a, M_b, N_b)\). Given an \(\parens{\alpha, q, \epsilon}\)-block encoding \(U_{A' \otimes B'}\) of \(A' \otimes B'\), we can construct an \(\parens{\alpha, q + \poly(\log(N_{\max})), \epsilon}\)-block encoding \(U_{\parens{A \otimes B}'}\) of \(\parens{A \otimes B}'\) using a single call to \(U_{A \otimes B}\) and \( \poly(\log(N_{\max}))\) additional gates. 
\end{lemma}

\begin{proof}
The structure of the proof is as follows. We first observe that
\begin{equation}
    \parens{A \otimes B}' = P_r (A' \otimes B') P_c
\end{equation}
where \(P_r\) and \(P_c\) are permutation matrices. We then construct \(\parens{1, \poly(\log(N_{\max})), 0}\)-block encodings of \(P_r\) and \(P_c\) (called \(U_{P_r}\) and \(U_{P_c}\), respectively) using \(\poly(\log(N_{\max}))\) gates. We then use the multiplication lemma of \cite{Gilyen_19_QSVT_and_beyond} to obtain the desired block encoding of \(\parens{A \otimes B}'\).

The permutation matrix \(P_c\) can be described as follows.

\begin{algorithm}[H]
\caption{Description of \(P_c\)}
\label{alg:Pc_desc}
\If {\(k < n_a n_b\)} {
    \(t = \floor{k / n_b}\)\\
    \(t' = k \bmod n_b\)\\
    \(P_c \ket{k} = \ket{t N_b + t'}\)
}
\ElseIf { k < \(n_a N_b\)}{
    \(t = \floor{(k - n_a n_b) / (N_b - n_b)}\)\\
    \(t' = (k - n_a n_b) \bmod (N_b - n_b)\)\\
    \(P_c \ket{k} = \ket{t N_b + n_b + t'}\)
}
\Else {
    \(P_c \ket{k} = \ket{k}\)
}
\end{algorithm}

As the above description of $P_c$ consists of standard arithmetic operations, it can be implemented classically using \(\poly(\log(N_{\max}))\) gates and ancillas. It also follows that we can implement the following quantum oracle with only a constant-factor overhead in the number of gates and \(\poly(\log(N_{\max}))\) ancilla qubits that are uncomputed at the end of the operation \cite[Section 3.2.5]{Nielsen_and_chuang}:
\begin{equation}
O_{P_c} \colon \ket k \ket 0 \ket 0_{\anc} \mapsto \ket k \parens{P_c \ket k} \ket 0_{\anc}.
\end{equation}

Furthermore, \(P_c\) is simple enough that its inverse also has an efficient classical computation.

\begin{algorithm}[H]
\caption{Description of \(P_c^{-1}\)}
\label{alg:Pc_inverse_desc}
\If {\(k < n_a N_b\)} {
    \(t = \floor{k / N_b}\)\\
    \(t' = k \bmod N_b\)\\
    \If{\(t' < n_b\)} {
        \(P_c^{-1} \ket{k} = \ket{t n_b + t'}\)
    }
    \Else {
        \(P_c^{-1} \ket{k} = \ket{n_a n_b + t (N_b - n_b) + (t' - n_b)}\)       
    }
}
\Else {
    \(P_c^{-1} \ket{k} = \ket{k}\)
}
\end{algorithm}

This can be used to uncompute the input \(k\). Hence, we can construct a block-encoding \(U_{P_c}\) of \(P_c\) as follows:
\begin{equation}
    U_{P_c} \colon \parens{\swapop} \circ \parens{O_{P_c^{-1}}} \circ \parens{O_{P_c}}.
\end{equation}

\(U_{P_r}\) can be constructed in an analogous manner, and can be combined with \(U_{A \otimes B}\) using the multiplication lemma of \cite{Gilyen_19_QSVT_and_beyond} to obtain the desired block encoding of \(\parens{A \otimes B}'\).
\end{proof}

\begin{corollary}
\label{cor:tensor_product_arbitrary_d} \Cref{lem:tensor_product_lemma_zpbe} holds for a tensor product of \(d\) matrices with an additional factor \(d\) in the number of gates and ancillas.
\end{corollary}
\begin{proof}
The permutations in \Cref{lem:tensor_product_lemma_zpbe} can be applied two matrices at a time. 
\end{proof}

Finally, we can permute the block encoding of \(\parens{I^{d}_{l \rightarrow l+1}}'\) to obtain a block encoding of \(\parens{\hat I^{d}_{l \rightarrow l+1}}'\) as described in the lemma below, achieving the main goal of this subsection. 

\begin{lemma}
    \label{lem:final_interpolation_ops}
    We can construct a \((2^{d/2}, O(dL), 0)\)-block encoding for \(\parens{\hat I^{d}_{l \rightarrow l+1}}'\) using \(\poly(dL)\) gates.
\end{lemma}
\begin{proof}
    We can use permutations to obtain a block encoding of \(\parens{\hat I^{d}_{l \rightarrow l+1}}'\) from a block encoding of \(\parens{I^{d}_{l \rightarrow l+1}}'\). We begin with the state
    \begin{equation}
        \ket{0}_{\anc} \otimes \sum_{k = 0}^{m^d - 1} \ket k.
    \end{equation}
    We subtract \(\sum_{i = 1}^{l-1} n^{d}_i\) mod \(m^d\), giving
    \begin{equation}
        \ket{0}_{\anc} \otimes \sum_{k = 0}^{m^d - 1} \ket{k - \sum_{i = 1}^{l-1} n^{d}_i \bmod m^d}.
    \end{equation}
    Next, we apply the block encoding of \(\parens{I^{d}_{l \rightarrow l+1}}'\) (\Cref{lem:block_encoding_I_d_l_to_l+1}), using \(\poly(dL)\) gates and ancillas, to obtain
    \begin{equation}
        \ket{0}_{\anc} \otimes \sum_{k = 0}^{m^d - 1} \parens{I^{d}_{l \rightarrow l+1}}' \ket{k - \sum_{i = 1}^{l-1} n^{d}_i \bmod m^d}.
    \end{equation}
    Finally, we add back \(\sum_{i = 1}^{l} n^{d}_i\) mod \(m^d\) to obtain
    \begin{equation}
        \ket{0}_{\anc} \otimes \sum_{k = 0}^{m^d - 1} \parens{\hat I^{d}_{l \rightarrow l+1}}' \ket{k}.
    \end{equation}
    The additions and subtractions can be implemented using \(O(dL)\) gates and ancillas \cite{cuccaro2004newquantumripplecarryaddition}. This gives us the required block encoding of \(\parens{\hat I^{d}_{l \rightarrow l+1}}'\). 
\end{proof}

\subsection{Block Encoding of Preconditioner}
\label{sec:block_encoding_preconditioner}
We are now ready to give a block encoding of \(\hat F^d_L\) as defined in \Cref{def:BPX_operator_F_embedding}. Essentially, we multiply the block encodings of \(\parens{\hat I^{d}_{l \rightarrow l+1}}'\) to obtain \(\parens{\hat I^d_{1 \rightarrow L}}'\) and then use linear combination of block encodings \cite{Gilyen_19_QSVT_and_beyond} to finally obtain a block encoding of \(\hat F^d_L\).

To describe the construction in more detail, we first recall relevant results from \cite{Gilyen_19_QSVT_and_beyond}. 

\begin{definition}[State preparation pair \cite{Gilyen_19_QSVT_and_beyond}]\label{def:state-preparation-pair}
Let $y \in \mathbb{C}^m$ and $\|y\|_1 \le \beta$.  
The pair of unitaries $(P_L,P_R)$ is called a $(\beta,b,\varepsilon)$-state-preparation pair if
\[
P_L \ket{0}^{\otimes b} = \sum_{j=0}^{2^b-1} c_j \ket{j}
\qquad \text{and} \qquad
P_R \ket{0}^{\otimes b} = \sum_{j=1}^{2^b-1} d_j \ket{j}
\]
such that
\[
\sum_{j=0}^{m-1} \bigl|\beta c_j^{*} d_j - y_j\bigr| \le \varepsilon
\]
and for all $j \in \{m,\ldots,2^b-1\}$ we have $c_j^{*} d_j = 0$.
\end{definition}

\begin{lemma}[Linear combination of block-encoded matrices \cite{Gilyen_19_QSVT_and_beyond}]\label{lem:lcu-block-encoding}
Let $A = \sum_{j=1}^m y_j A_j$ be an $s$-qubit operator and
$\varepsilon \in \mathbb{R}_{+}$.  
Suppose that $(P_L,P_R)$ is a $(\beta,b,\varepsilon_1)$-state-preparation pair for $y$ and
\[
W = \sum_{j=0}^{m-1} \ket{j}\!\bra{j} \otimes U_j
      + \Bigl( I - \sum_{j=0}^{m-1} \ket{j}\!\bra{j} \Bigr) \otimes I_a \otimes I_s
\]
is an $(s+a+b)$-qubit unitary such that for all $j \in \{0,\ldots,m\}$,
$U_j$ is an $(\alpha,a,\varepsilon_2)$-block encoding of $A_j$.  
Then we can implement an $(\alpha\beta,\, a+b,\, \alpha\varepsilon_1 + \alpha\beta\varepsilon_2)$-block encoding of $A$ using each of $W$, $P_R$, and $P_L^{\dagger}$ once.
\end{lemma}

Using this, we construct the preconditioner block encoding as follows.

\begin{theorem}
    \label{thm:final_preconditioner}
    We can implement an \(\left ( L2^{dL/2}, poly({d}L), 0 \right )\)-block encoding of \(\widehat F^d_L\) using \(\poly(dL)\) gates and ancillas.
\end{theorem}
\begin{proof}
From \Cref{lem:final_interpolation_ops}, we can obtain \(\left ( 2^{d/2}, \poly({d}L), 0 \right )\)-block encodings of \(\hat I^d_{l \rightarrow l+1}\) using \(\poly(dL)\) gates. Then we can obtain \(\left (2^{d(L-l)/2}, \poly({d}L), 0 \right )\)-block encodings of \(\hat I^d_{l \rightarrow L}\) for \(l = 1\) to \(L-1\) using \Cref{obs:product_hat_embeddings} and the multiplication lemma \cite{Gilyen_19_QSVT_and_beyond}. 

    Observe that these block encodings also serve as \(\left (2^{dL/2 - l}, \poly({d}L), 0 \right )\)-block encodings of \(\hat G^d_l \coloneqq 2^{-l(2 - d)/2} \hat I^d_{l \rightarrow L}\). Subnormalizing these block encodings by a factor \(2^{-l}\), we obtain \(\left (2^{dL/2}, \poly({d}L), 0 \right )\)-block encodings of \(\hat G^d_l\) for \(l = 1\) to \(L-1\).

    It remains to implement \(\widehat F^d_L = \sum_{l = 1}^L \hat G^d_l\). To do this using linear combination of block encodings (\Cref{lem:lcu-block-encoding}), we require the state preparation pair \((P_L, P_R)\) and the unitary \(W\). 
    \begin{enumerate}
        \item For \(P_L\) and \(P_R\), we can use the Grover-Rudolph state preparation method \cite{grover2002creatingsuperpositionscorrespondefficiently} to perform \(\ket 0 \mapsto \sum_{j = 1}^L \ket{j}\) using \(O(L)\) gates and ancillas. This gives us a \(\left ( L, \lceil \log L \rceil, 0 \right )\)-state-preparation pair for the all-$1$s vector.
        \item For \(W\), controlled on \(j\), we apply the block encoding of \(\hat F^d_j\). This can be done using one call to each block encoding of \(\hat G^d_j\) and \(O(L)\) additional gates and ancillas.
    \end{enumerate}

    Thus, applying \Cref{lem:lcu-block-encoding}, we obtain an \(\left ( L2^{dL/2}, \poly({d}L), 0 \right )\)-block encoding of \(\widehat F^d_L\) using \(\poly(dL)\) gates and ancillas. Note that all the matrices in the lemma are top-left embeddings (\Cref{def:operator_embeddings}), but we ignore the primes for ease of notation.
\end{proof}

\section{Block Encoding of the Hamiltonian}
\label{sec:block_encoding_hamiltonian}

{ We now have a block of \(F\) which we use in the final block encoding. In \Cref{sec:diagonal_material_matrix_BE} and \Cref{sec:block-encoding-of-absorption-and-fission-matrices}, we provide implementations for loading the data (diffusion, fission, and absorption coefficients) and constructing the matrices that require this data. In this section, we develop other auxiliary results used to implement the Hamiltonian, and then present the final construction. 
}

\begingroup

\ifdefined\hypersetup
  \hypersetup{allcolors=blue}
\fi
\endgroup

{ First, }we recall lemmas used to construct \(L^{-1}\) and \(C^{1/2}\). We refer the reader to \cite[Section 5]{deiml2025quantumrealizationfiniteelement} for an excellent review of challenges with quantum preconditioning, and how their scheme (that we also use) overcomes them. 

\begin{theorem}
    \label{thm:FLFT}
    Let \(L\) be as defined in \Cref{prob:discrete_problem}. Let \(F\) be the BPX preconditioner matrix (\Cref{def:BPX_operator_F}). Then
    \begin{equation}
        F(F^T L F)^+F^T = L^{-1}
    \end{equation}
    where \(M^+\) denotes the Moore-Penrose pseudoinverse of a matrix \(M\).
\end{theorem}
\begin{proof}
    Since \(L\) is 
    positive definite, we can write it as \(L = S^TS\) where \(S\) is invertible. Then, we rewrite \(B: = SF\). Because \(S\) is invertible and \(F\) has full row rank (the columns of \(F\) corresponding to level \(l = L\) form an identity matrix), it follows that \(B\) has full row rank and \(B^T\) has full column rank. This implies 
    \begin{equation}
    \label{eq:BB_plus}
\begin{aligned}
    BB^+ &= I \\
    (B^T)^+ B^T &= I.
\end{aligned}
\end{equation} Thus, using \Cref{eq:BB_plus}, we have
    \begin{equation}
    \begin{aligned}
        F(F^T L F)^+F^T &= F(F^T S^T S F)^+F^T \\
        &= S^{-1} B (B^T B)^+ B^T S^{-T} \\
        &= S^{-1} B B^+ (B^T)^+ B^T S^{-T} \\
        &= S^{-1} B B^+ (B B^+)^T S^{-T} \\
        &= S^{-1}S^{-T} \\
        &= L^{-1}.
    \end{aligned}
\end{equation}
\end{proof}

For the Moore-Penrose pseudoinverse, we use the following theorem paraphrased from \cite[Theorem 41]{Gilyen_19_QSVT_and_beyond} (refer there for formal definitions).

\begin{theorem}[Moore-Penrose Pseudoinverse via QSVT]
\label{thm:moore_penrose_pseudoinverse_qsvt}
Let \(U\) be a unitary and let there exist projectors \(\Pi\) and \(\widetilde \Pi\) such that \(A = \tilde \Pi U \Pi\). Let \(\Pi_{0, \geq \eta}\) and \(\widetilde \Pi_{0, \geq \eta}\) be the projectors onto the subspaces spanned by the { right and left} singular vectors of \(A\), respectively, with singular values in \([ \eta, 1]\). Then there is an \(m = O \parens { \frac 1 \eta \log \parens{ \frac 1 \delta}} \) and an efficiently computable \(\Phi \in \mathbb R^m\) such that
\begin{equation}
    \norm{
        \parens{ \bra{+} \otimes \Pi_{0, \geq \eta} } U_{\Phi} \parens{ \ket{+} \otimes \widetilde \Pi_{0, \geq \eta} } 
        - \Pi_{0, \geq \eta} \parens{\frac \eta 2 \cdot A^+ }\widetilde \Pi_{0, \geq \eta}
    } \leq \delta .
\end{equation}
Moreover, \(U_{\Phi}\) can be implemented with \(m\) uses of \(U\) and \(U^{\dagger}\), \(m\) uses of of \(C_{\Pi}NOT\) and \(m\) uses of \(C_{\widetilde \Pi} NOT\), and \(m\) single-qubit gates. 
\end{theorem}
We use the following lemma of \cite{Chakraborty_2019_Block_Encoded_Matrix_Powers} to implement the square root of the fission operator.

\begin{lemma}
\label{lem:positive_powers_block_encoding}
    Let \(c \in (0, 1]\) and \(\kappa \geq 2\). Let \(H\) be a Hermitian matrix such that \(I/\kappa \preceq H \preceq I\). Suppose we are given a unitary \(U\) that is an \((\alpha, a, \delta)\)-block encoding of \(H\), where \(\delta = o \parens{\epsilon/\parens{\kappa \log^3(\kappa/\epsilon)}}\), that can be implemented using \(T_U\) elementary gates. Then for any \(\epsilon\), we can implement a unitary \(\tilde U\) that is a \((2, a + O(\log \log (1/\epsilon)), \epsilon)\)-block encoding of \(H^c\) in  
    \begin{equation}
        O \parens{\alpha \kappa \parens{a + T_U} \log^2(\kappa/\epsilon)}.
    \end{equation}
    one- and two-qubit gates. 
\end{lemma}

{ We now present the final block encoding of the Hamiltonian. }

\begin{theorem}
    \label{thm:block_encoding_Hamiltonian}
    Let \(\delta\) be an error parameter and \(h\) the mesh size parameter. Then we can prepare an \begin{equation}
\parens{O \parens{ \log^2 \frac 1 h}, O \parens{ \poly \parens{ \log \frac 1 h, \log \frac 1 {\delta}}}, O \parens { \delta^{1/{2}} \cdot \poly \parens { \log \frac 1 \delta, \log \frac 1 h} }}
\end{equation}
block encoding of the Hamiltonian \(H \coloneqq C^{1/2} \parens{L+A}^{-1} C^{1/2}\) (\Cref{prob:discrete_problem}) using
\(
 O \parens{ z \cdot \poly \parens{\log \frac 1 h, \log \frac 1 {\delta}}}
 \)
 one- and two-qubit gates and
 \(
 \poly \parens{\log \frac 1 h, \log \frac 1 {\delta}}
 \)
ancillas.
\end{theorem}

We outline the proof construction here and defer a detailed proof to the appendix (\Cref{thm:appendix_block_encoding_Hamiltonian}). 
We have the following decomposition:
\begin{equation}
    \begin{aligned}
    H &= C^{1/2} \parens{L+A}^{-1} C^{1/2} \\
    &= C^{1/2} \parens{I + L^{-1} A}^{-1} L^{-1} C^{1/2} \\
    &= C^{1/2} \parens{I + \parens{F (F^T L F)^+ F^T} A}^{-1} \parens{F (F^T L F)^+ F^T} C^{1/2} \\
    &= \underbracket{\parens{ \frac C {h^3} }^{1/2}}_{1} \underbracket{\parens{I + \parens{h^{3/2} F} (F^T L F)^+ \parens{h^{3/2} F }^T \frac A {h^3}}^{-1}}_{2} \underbracket{\parens{h^{3/2}F} (F^T L F)^+ \parens{ h^{3/2} F}^T}_{3} \underbracket{\parens{ \frac C {h^3} }^{1/2}}_{4}, \\
    \end{aligned}
\end{equation}
where in the second step we use fast-inversion preconditioning \cite{Tong_2021_fast_inversion} and in the third step we use \Cref{thm:FLFT}.

We consider each of the above four components separately and then combine them using the multiplication lemma of \cite{Gilyen_19_QSVT_and_beyond}.

\begin{enumerate}
\item For the first term and fourth term: \(C\) is a sparse matrix, but it has one zero block and one block with a constant condition number (by \Cref{thm:properties_of_C}). The complexity of applying the square root using QSVT depends on the condition number, which is unbounded for the full matrix \(C\). To get around this, we first replace the zero block by an identity block { scaled by a \(h^3\) factor} { and call the resulting matrix \(C^{(p)}\)} as in \Cref{lem:loading_C_p}. Then we apply the square root using \Cref{lem:positive_powers_block_encoding}. We block-encode projectors \(\Pi_C\) into the non-zero subspace of \(C\) (\Cref{lem:projectors_for_C}) and construct the block encoding of \(\Pi_C \parens{C^{(p)}}^{1/2} \Pi_C\) to obtain a block encoding of \(C^{1/2}\). Finally, \(\norm{C} = O(h^3)\), so we can divide all the blocks  by \(h^3\). 

\item For the third term: We can obtain a block encoding of \(F^TLF\) from \cite[Theorem 6.3]{deiml2025quantumrealizationfiniteelement}. Being able to produce this with block-encoding factor and gate complexity polylogarithmic in \(1/h\) is the crux of this quantum preconditioner and the main result of \cite{deiml2025quantumrealizationfiniteelement}. This term has a constant effective condition number \cite{deiml2025quantumrealizationfiniteelement}, and we can apply the Moore-Penrose pseudoinverse of \cite[Theorem 41]{Gilyen_19_QSVT_and_beyond} to obtain \(\parens{F^T L F}^+\). For \(F\) and \(F^T\), we use the preconditioner block encoding constructed in \Cref{thm:final_preconditioner} of  \Cref{sec:preconditioner_construction} and finally apply the multiplication lemma to obtain the final block encoding. 
\item The second term is just the third term again multiplied by a \(A/h^3\), which we can obtain using \Cref{sec:block-encoding-of-absorption-and-fission-matrices} and adding the identity. Once again, we can apply the Moore-Penrose pseudoinverse construction of \cite[Theorem 41]{Gilyen_19_QSVT_and_beyond} along with the robustness lemma of \cite[Lemma 22]{Gilyen_19_QSVT_and_beyond} to control the error.
\end{enumerate}

\section{Complexity Analysis}
\label{sec:complexity_analysis}
Armed with a block encoding of the Hamiltonian from \Cref{thm:block_encoding_Hamiltonian}, we can perform standard phase estimation to obtain the largest eigenvalue \cite{shao_2021_generalized_eigenvalue_ode,Chakraborty_2019_Block_Encoded_Matrix_Powers} of \Cref{prob:discrete_problem_standard}. 

First, we recap the formal definition of the QPE problem and its complexity for a Hermitian matrix \cite{shao_2021_generalized_eigenvalue_ode}.

\begin{definition}
    \label{def:QPE} Let \(A\) be an \(n \times n\) Hermitian matrix with spectral decomposition \(A = \sum_{k=1}^n \lambda_k \ketbra{u_k}{u_k}\). Let \(\epsilon \in (0,1)\). The quantum phase estimation (QPE) problem with accuracy \(\epsilon\) is defined as follows. Given access to \(\sum_{k=1}^n \beta_k \ket{u_k}\), perform the mapping
    \begin{equation}
        \label{eq:qpe_mapping}
        \sum_{k=1}^{n} \beta_k \ket{0} \ket{u_k} \mapsto \sum_{k=1}^{n} \beta_k \ket{\tilde \lambda_k}\ket{u_k}
    \end{equation}
    such that \(|\tilde \lambda_k - \lambda_k| \leq \epsilon\) for all \(k \in \{1, 2, \ldots, n\}\).
\end{definition}

\begin{lemma}[\cite{shao_2021_generalized_eigenvalue_ode,Chakraborty_2019_Block_Encoded_Matrix_Powers}]
\label{lemma:qpe_hermitian}
Let \(\epsilon, \tilde \epsilon \in (0,1)\) and let \(\epsilon' = \tilde \epsilon \epsilon/4 \log^2(1/\epsilon)\). Given an \((\alpha, q, \epsilon')\)-block encoding of Hermitian matrix \(A\) that is implemented in \(O(T)\) gates,  there is a quantum algorithm that solves the QPE problem of \(A\) with accuracy \(\epsilon\), with success probability at least \(1 - \tilde \epsilon\), in \(O \left(T_{\mathrm{in}} + \alpha \epsilon^{-1}(q+T)\poly(\log(1/\tilde \epsilon))\right)\) gates where \(T_{\mathrm{in}}\) is the number of gates required to prepare the initial state. 
\end{lemma}

Finally, we analyze the overall complexity to solve our problem. 

\begin{theorem} 
    \label{thm:final_complexity}
\Cref{prob:neutron_diffusion_problem} can be solved with accuracy \(\epsilon\) and constant success probability using 
$
    O \parens{ z \cdot \frac 1 \epsilon \poly(\log \parens{ \frac 1 \epsilon}) } 
$
one- and two-qubit gates and classical operations, where \(z\) is the number of different material { regions}. The big $O$ hides constant factors depending on coefficients \(D\), \(\Sigma_a\), \(\nu\Sigma_f\), and consequently various norms of the solution. 
\end{theorem}

\begin{proof}
We divide our error budget into two equal parts. 
We use the first to control the convergence error of the finite element method. By \Cref{thm:eigenvalue_convergence}, mesh size \(h = O(\epsilon^{\pi/\gamma})\) suffices to obtain the correct eigenvalue to error \(\epsilon/2\). 

We use the remaining \(\epsilon/2\) error for the accuracy of phase estimation. By \Cref{lemma:qpe_hermitian}, block-encoding error at most \(\epsilon' = O \parens{\epsilon/ \log^2(1/\epsilon)}\) suffices for constant probability of success. For this, it suffices to set \(\delta = O \parens{ \epsilon^3}\) in \Cref{thm:block_encoding_Hamiltonian}. Substituting these values of \(h\) and \(\delta\) into \Cref{thm:block_encoding_Hamiltonian}, we have a block encoding of the Hamiltonian \(H\) with parameters \(\alpha = O \parens{ \poly(\log \parens{1/\epsilon} ) }\) and  \(q = O \parens{ \poly(\log \parens{1/\epsilon} ) }\) that can be implemented with \(T = O \parens{ z \cdot \poly(\log \parens{1/\epsilon} ) }\) one- and two-qubit gates. 

In \Cref{thm:uniform_initial_state_preparation} in the appendix, we show that the uniform superposition state has a constant overlap with the final state and can be prepared using \(T_{\mathrm{in}} = \poly(\log(1/h)) = \poly(\log(1/\epsilon))\) gates.

Substituting these values into \Cref{lemma:qpe_hermitian}, we obtain the claimed complexity. 
\end{proof}

We conclude this section by considering the dependence on \(z\), the number of material regions, in the complexity of \Cref{thm:final_complexity}. The multiplicative dependence on \(z\) comes from the cost of  {implementing the matrices encoding the \(D(\mathbf x)\), \(\Sigma_a(\mathbf x)\), and \(\nu \Sigma_f(\mathbf x)\) data}.

The multiplicative dependence on \(z\) can be avoided if the classical complexity of determining the coefficients at a given point does not scale with number of regions (for example, consider patterned instances as in \Cref{fig:checkerboard_pattern}). 

\section{Numerical Experiments: Convergence of Classical Algorithm with Uniform Mesh}
\label{sec:numerical_experiments}

In this section, we numerically explore hard instances for the classical uniform finite element method as applied to \Cref{prob:neutron_diffusion_problem}. In particular, we give an example of a combination of material geometries and diffusion coefficients for which the observed order of convergence indicates that the classical uniform finite element method requires \(\Omega\parens{1/\epsilon^p}\) steps for some large \(p\) to produce the eigenvalue \(\lambda\) to \(\epsilon\) error, whereas \Cref{thm:final_complexity} shows that the quantum algorithm running the same scheme uses only \(\tilde O(1/\epsilon)\) steps. 

We consider a checkerboard pattern of alternating diffusion coefficients as shown in \Cref{fig:checkerboard_pattern}. This configuration has been previously considered in the literature \cite{Bruce_kellogg_1974} and is notorious for slow convergence rates of uniform FEM methods. For example, Nochetto \cite{Nochetto_2010} shows that uniform methods struggle on this pattern and uses it as a motivating example for an adaptive method, which can overcome the slow convergence rates in several cases. 

\begin{figure}[ht]
    \centering
    \includegraphics[width=0.5\textwidth]{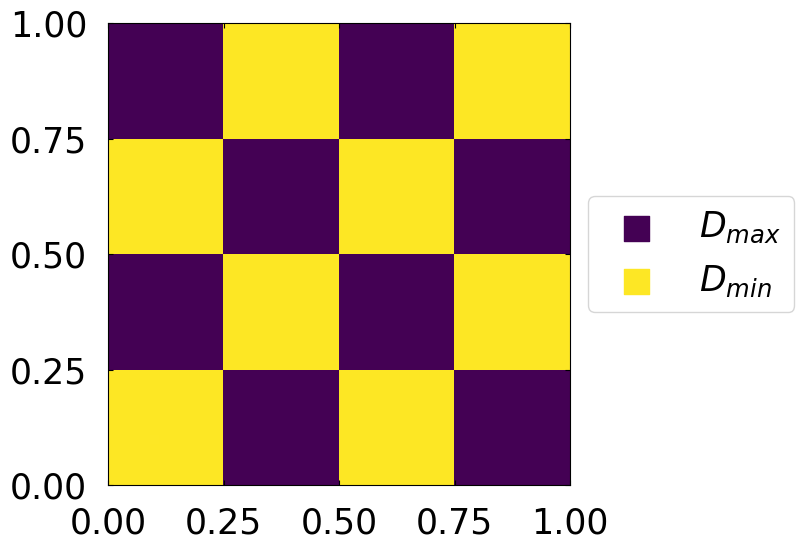}
    \caption{A \(2D\) checkerboard pattern with diffusion coefficients \(D_{\max}\) and \(D_{\min}\).}
    \label{fig:checkerboard_pattern}    
\end{figure}

While the checkerboard pattern of alternating materials is contrived, and we artificially set high values of diffusion coefficients in the experiment that follows, it is common in neutron transport applications for material properties to be presented as piecewise constant on a (sometimes uniform) Cartesian grid. Light water nuclear reactors are often separated into uniformly-sized assemblies, each of which contains a different makeup of isotopes. This results from varying initial enrichments of fissile fuel being used, lengths of time the assembly was present in the reactor, and locations of the assembly within the reactor (which affects neutron flux and therefore isotopic composition). Additionally, the assemblies themselves are often made up of cells containing varying materials in a uniform grid. If the simulation calls for higher precision where the material compositions of individual fuel cells are to be differentiated, this would be possible with our quantum algorithm. As an example, the C5G7 MOX benchmark that is commonly used within the nuclear industry exhibits this uniform Cartesian mesh in 2 and 3 spatial dimensions (see \Cref{fig:C5G7-full-image} and \cite[Figures 1--3]{Lewis_2001_C5G7}). However, this benchmark assumes multiple energy groups and Robin boundary conditions while our algorithm assumes Dirichlet boundary conditions and a single energy group, so our algorithm is not directly applicable without modification of the original data.

\begin{figure}
     \centering
     \includegraphics[width=0.4\textwidth]{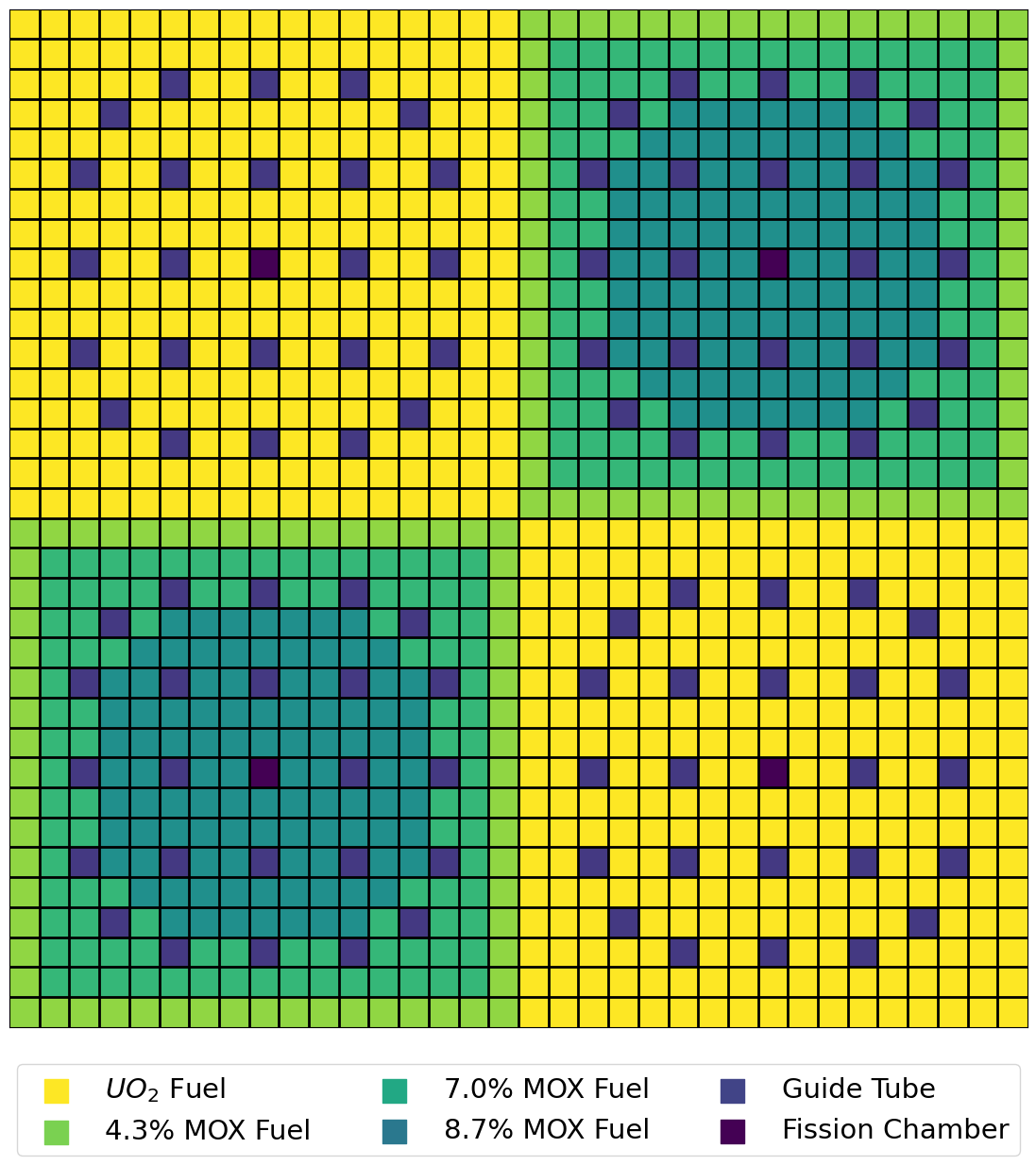}
     \caption{2D cross section of a portion of the C5G7 reactor geometry (recreation of Figure 3 of \cite{Lewis_2001_C5G7}).
     }
     \label{fig:C5G7-full-image}    
 \end{figure}

For our experiments, we take \(D_{\min} = 1\) and vary \(D_{\max}\). For convenience, we take \(\Sigma_a(\bx) = \nu \Sigma_f(\bx) = 1\) everywhere. Regularity bounds in this case are given by \cite{Petzoldt2001}. While the results of \cite{Petzoldt2001} are for a general boundary value problem,  in particular, they imply that the minimum eigenfunction of \Cref{prob:neutron_diffusion_problem} is in \(H^{1+\chi}(\Omega)\), where \(\Omega = [0,1]^2\) and
\begin{equation}
    \chi \geq \frac{4}{\pi} \arctan \parens{\sqrt{\frac{1}{D_{\max}}}} - \eta    
\end{equation}
for any \(\eta > 0\). From the proof in \Cref{thm:eigenvalue_convergence}, this implies that the theoretical worst-case convergence rate of eigenvalues for any uniform method is \( \epsilon = \abs{\lambda - \lambda_h} = O(h^{2 \chi^*})\) for mesh cell size \(h\),  where \(\chi^* = \frac{4}{\pi} \arctan \parens{\sqrt{\frac{1}{D_{\max}}}}\). In turn, the theoretical worst-case scaling of the number of mesh elements \(N\) with \(\epsilon\) is \(N = O(\epsilon^{-p^{*}})\) where \(p^* = 1/\chi^*\), since \(N = \Theta(1/h^2)\) in the 2D case for a uniform mesh.

Here, we aim to test whether this slow convergence rate is observed in practice for the minimum eigenvalue, for sufficiently large mesh sizes. If the true eigenvalue \(\lambda^*\) were available, we could compute the obtained eigenvalue \(\lambda_N\) for various mesh sizes, plot the error \(|\lambda_N - \lambda^*|\) as a function of \(N\), and extrapolate the convergence order \(\chi\). However, since the true eigenvalue is not available, we resort to the heuristic given in \cite{Journal_of_Fluids_Engineering_2008} to report the observed order of convergence.  Here, three different mesh sizes \(N_{1} < N_2 < N_3\) are selected, such that \(N_1 = N_2/r = N_3/r^2\). Then, the observed order of convergence \(\chi'\) is estimated as
\begin{equation}
    \label{eq:observed_order_convergence}
    { \chi' = \left|\frac{\log \left|\frac{\lambda_{N_2} - \lambda_{N_3}}{\lambda_{N_1} - \lambda_{N_2}}\right|}{\log(r)}\right|,}
\end{equation}
giving an estimated exponent \(p' = 1/\chi'\). 

First we consider the \(4 \times 4\) checkerboard pattern of \Cref{fig:checkerboard_pattern}, and consider \(11\) values of \(D_{\max}\): \(1
\), \(10\), \(20\), \(30\), \(40\), \(50\), \(60\), \(70\), \(80\), \(90\), and \(100\). For each value of \(D_{\max}\),  we consider \(10\) levels of mesh refinement, containing \(N_0\) to \(N_9\) mesh elements, where \(N_0 = c \cdot 4^2\) and \(N_i = 4 \cdot N_{i-1}\) (so \(N_9 = c \cdot 4^{11}\)), with the constant \(c\) depending on the mesh element we use. Thus, the largest meshes in our experiments have about \(4\) million elements. For our experiments, we consider two different meshing schemes: triangular elements with piecewise linear functions (P1), for which \(c = 2\), and square elements with piecewise bilinear functions (Q1), for which \(c = 1\). The meshes for level \(0\), \(1\), and \(5\) for the triangular case are shown in the top row of \Cref{fig:mesh_levels}. The results from both meshing schemes are similar, so we only show data for the triangular case. 

\begin{figure}
    \centering
    \includegraphics[width=0.9\textwidth]{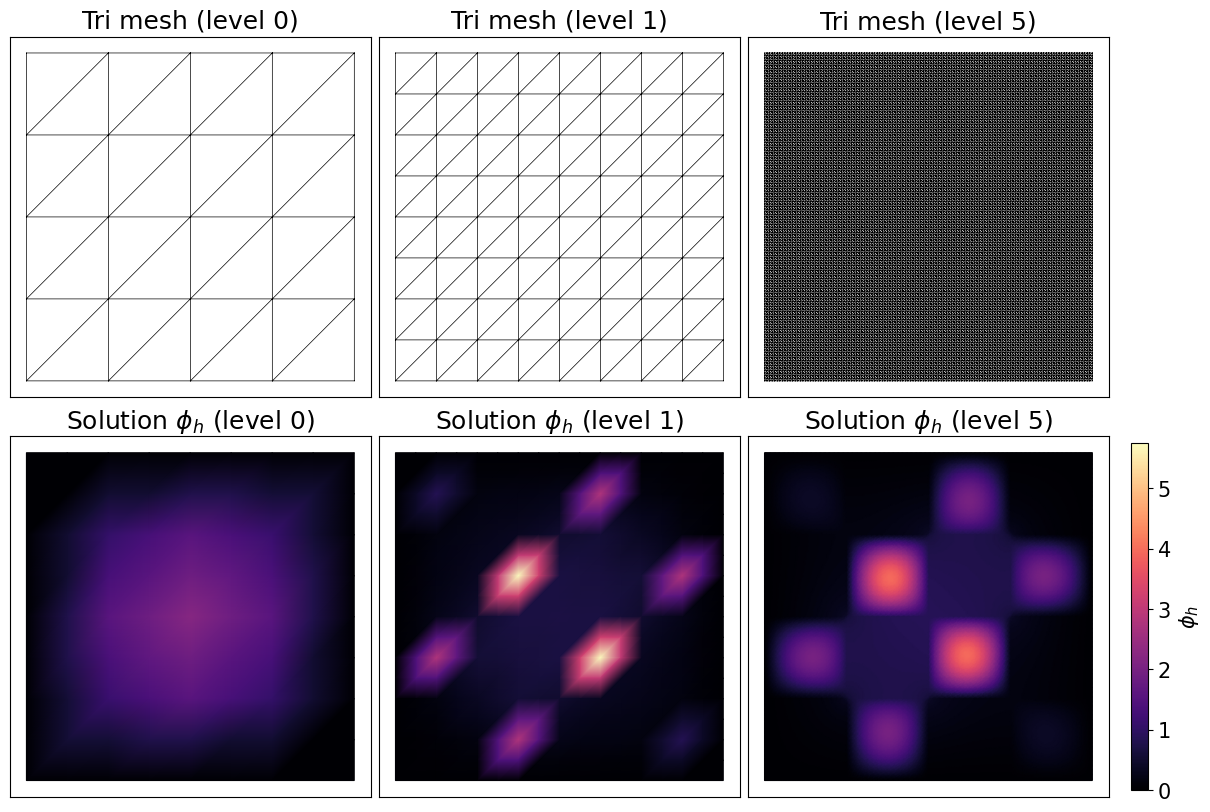}
    \caption{Triangular meshes (top row) and solutions (bottom row) for refinement levels \(0\) (left), \(1\) (middle), and \(5\) (right) for \(D_{\max} = 100.\)}
    \label{fig:mesh_levels}
\end{figure}

For each level \(i\) from \(0\) to \(9\), we run the FEM algorithm with \(N_i\) mesh elements, and compute the minimum eigenvalue \(\lambda_i\). The minimum eigenfunctions corresponding to levels \(0\), \(1\), and \(5\) are shown in the bottom row of \Cref{fig:mesh_levels}. Then, for each level \(i\) from \(0\) to \(7\), we compute the observed exponent \(p'_i\) as in \Cref{eq:observed_order_convergence} using the three grids \(N_i\), \(N_{i+1}\), and \(N_{i+2}\). The values of \(p_i'\) for \(D_{\max} = 1\) to \(D_{\max} = 40\) are plotted in \Cref{fig:observed_exponent_convergence_1} and from \(D_{\max} = 50\) to \(D_{\max} = 100\) in \Cref{fig:observed_exponent_convergence_2}. We also plot the theoretical worst-case exponent \(p^*\) for each \(D_{\max}\) in the figure with a dotted line. 

\begin{figure}
    \centering
    \includegraphics[width=0.9\textwidth]{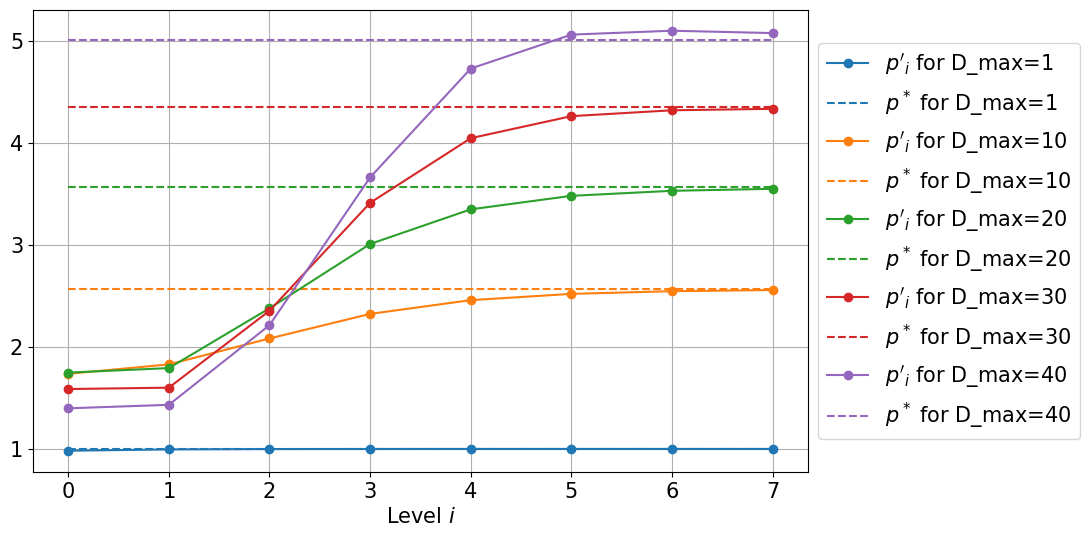}
    \caption{Observed exponent \(p'_i\) as a function of the level \(i\) for \(D_{\max} = 1\) to \(D_{\max} = 40\). The dotted lines represent the theoretical minimum exponent \(p^*\) for each value of \(D_{\max}\).}
    \label{fig:observed_exponent_convergence_1}
\end{figure}

\begin{figure}
    \centering
    \includegraphics[width=0.9\textwidth]{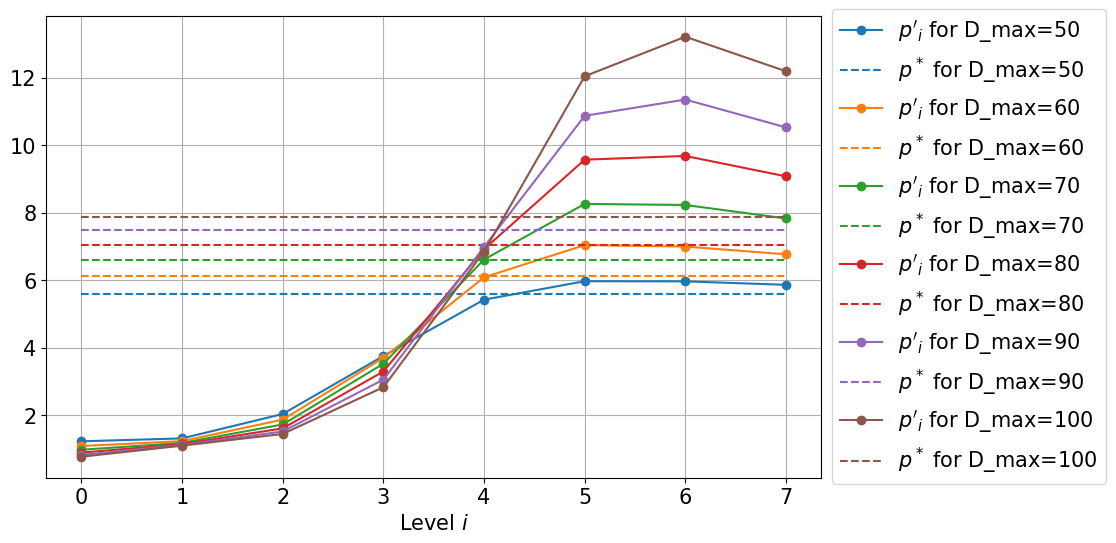}
    \caption{Observed exponent \(p'_i\) as a function of the level \(i\) for \(D_{\max} = 50\) to \(D_{\max} = 100\). The dotted lines represent the theoretical minimum exponent \(p^*\) for each value of \(D_{\max}\).}
    \label{fig:observed_exponent_convergence_2}
\end{figure}

The case \(D_{\max} = 1\) has constant coefficients, and we can analytically compute the true minimum eigenvalue \(\lambda^* = 2 \pi^2 + 1\). The fit shown in the log-log plot in \Cref{fig:loglog} shows that \(p^* \approx 1\) to within \(1 \%\). From \Cref{fig:observed_exponent_convergence_1}, the observed value of convergence \(p_i'\) is also \(1\) to within \(2 \%\) error in this case, so the heuristic matches the exponent in the case where we have the analytical solution. For \(D_{\max} = 10\) to \(D_{\max} = 40\), the observed exponent \(p'_i\) appears to converge to the theoretical maximum. For example, with \(D_{\max} = 40\) the observed exponent \(p'_i\) is around 5 and within \(2 \%\) of the theoretical maximum \(p^* \approx 5.008\), consistent with \(N = \Theta(\epsilon^{-5})\).
However, for \(D_{\max} = 50\) to \(D_{\max} = 100\), we notice that the observed exponent \(p'_i\) overshoots the theoretical maximum. After overshooting, we see an eventual downward trend. However, it still appears unlikely that the empirical exponent for large \(N\) will be better than the theoretical worst case. This gives a complexity as bad as around \(N = \Theta(\epsilon^{-8})\) at \(D_{\max} = 100\) (although more analysis would be required to confirm this).

We also run similar experiments for the \(2 \times 2\) checkerboard and find that the observed exponent \(p'_i\) is always equal to \(1\) independent of \(D_{\max}\). For the \(8 \times 8\) checkerboard, we do not observe any convergence and would require more data to draw conclusions. 
While we are unable to collect data in 3D, we expect the rates to be worse than in 2D. Code for the experiments presented here is available at \url{https://github.com/Tinkidinki/diffusion-fem-codes}. 

\begin{figure}
    \centering
    \includegraphics[width=0.6\textwidth]{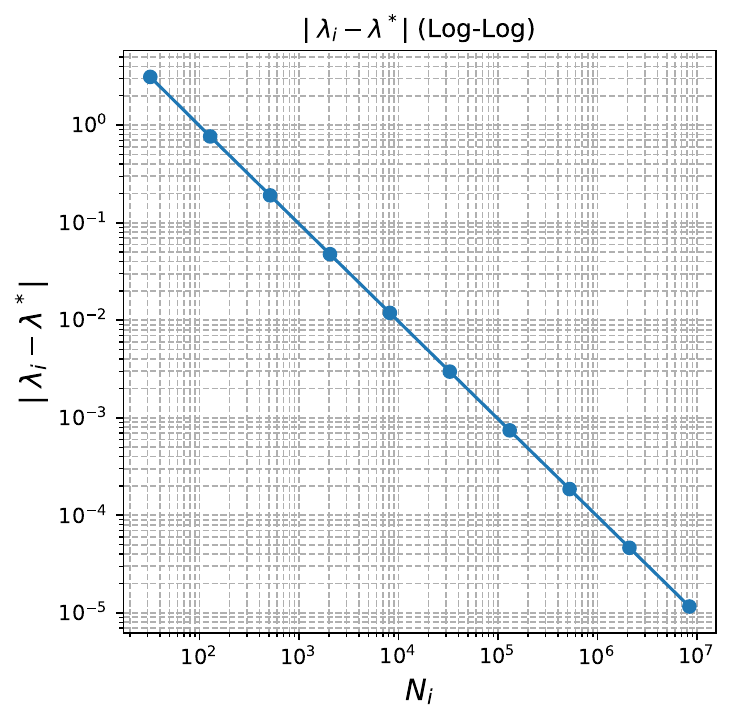}
    \caption{Log-log plot of the error \(|\lambda_i - \lambda^*|\) as a function of \(N_i\) for \(D_{\max} = 1\). The slope of the line is \(1\), which matches the observed exponent \(p'_i\).}
    \label{fig:loglog}
\end{figure}

Overall, these results indicate that the checkerboard pattern is a hard instance in practice for the classical uniform FEM, and a quantum algorithm running the same scheme can perform significantly better. We note that using classical nonuniform methods lower our speedup, and provide numerical experiments for this in \Cref{sec:nonuniform_classical_grid_scaling}.

\section{Conclusion and Open Questions}
\label{sec:conclusion_and_open_questions}

In this work, we have considered the neutron diffusion \(k\)-eigenvalue problem, a fixed-dimensional partial differential equation with discontinuous material coefficients, and compared the classical and quantum version of the uniform finite element method to solve this problem.  {  We provided evidence that the quantum version can give a signifiant end-to-end polynomial speedup over the classical version, and in particular, have shown that the quantum algorithm outperforms the empirically observed classical performance by a polynomial factor in \(1/\epsilon\) in the cases tested.} 

We have demonstrated the utility of quantum preconditioning methods for developing fast quantum algorithms for heterogeneous PDEs. 
To the best of our knowledge, this is also the first time the effects of solution regularity on FEM convergence rate have been analyzed in a quantum algorithms setting. 

Our work suggests several directions for future investigation: 
\begin{enumerate}
\item Irregularities in the shape of the domain such as sharp angles and re-entrant corners, as well as non-smooth problem coefficients, are known to lower solution regularity in PDEs and require fine meshing, at least close to the regions of irregularity. 
In such cases, quantum algorithms may provide a significant speedup even in the fixed dimensional setting. Further exploration of this possibility would have to consider PDEs where there is no ``shortcut'' available in place of meshing such as Monte Carlo methods. For instance, it is known to be difficult to design Monte Carlo methods for second-order hyperbolic equations \cite{Yu2023Monte}.

\item A quantum algorithm for the neutron diffusion \(k\)-eigenvalue problem could be a stepping stone towards one for the full neutron transport equation or, more generally, the linearized Boltzmann equation. If the methods presented here could be generalized to higher order approximations of the full neutron transport equation (such as the \(P_N\) or discrete ordinates approximations \cite{Bell_1970}), 
this might lead to a quantum speedup for computational neutron transport problems.
Additionally, as the number of dimensions increases with more accurate approximations (up to six dimensions for full transport), the asymptotic complexity of quantum algorithms compared to the best deterministic solver should also improve. This is because the complexity of the deterministic solver depends exponentially on the number of dimensions, a limitation that quantum approaches do not share. A major challenge here is that the angular discretization used to solve the linear Boltzmann equation changes the structure of the discretized operators and parts of the matrix $C$ may no longer be sparse.

\item Realistic neutronics calculations, either with neutron diffusion or transport, involve the energy dependence of the neutrons. For simplicity, our work made the one-speed approximation. With energy dependence, the overall system is non-sparse and non-Hermitian, increasing the complexity of a direct classical solution. Classical algorithms address this by solving each energy group independently, using results from previous calculations. In the case where particles only lose energy or slow down, the system is triangular and a single sweep through energy groups is adequate. If there is upscattering, a consistent solution is obtained iteratively. How to best handle energy dependence using a quantum approach is an open question. Specifically, whether it is best to emulate the standard classical approach in a quantum context, attempt to solve the full energy dependent equations in a single non-Hermitian system, or something else remains to be considered.

\item Our analysis assumed Dirichlet boundary conditions. This is not a fully faithful description of neutron diffusion, but is justifiable when the system is large relative to the mean distance neutrons diffuse during their lifetime. Practical neutron diffusion problems, however, require Neumann and Robin boundary conditions. Neumann-type boundary conditions are typically encountered in neutron diffusion for reflecting boundaries when exploiting symmetry, since many engineered systems are designed with some symmetry. Robin-type boundary conditions are needed to handle cases where the incident flow rate or inward partial current is specified at a boundary. The method will need to be extended to treat these boundary conditions for it to be a practical tool.

\item It might be possible to apply similar techniques to other PDEs that are related to the \(k\)-eigenvalue neutron diffusion equation, such as the eigenvalue linear reaction-diffusion, screened Poisson, or Helmholtz equations. Such equations appear in many fields such as electrostatics, fluid mechanics, particle physics, and image processing. 
These applications may require developing new techniques to accommodate features such as other coordinate systems, negative reaction terms, and continuously-varying coefficient functions. Additionally, our method estimates only the fundamental eigenvalue, which may not be sufficient for other applications. A future direction would be to extend our approach to estimate a set of eigenvalues.

\item Finally, we would like to understand the applicability and limitations of quantum Monte Carlo algorithms for neutron transport as an alternative to the finite element methods we have explored. Can we establish quantum lower bounds for the neutron diffusion \(k\)-eigenvalue problem, as well as other closer approximations of the full neutron transport equation, when using Monte Carlo?
\end{enumerate}

\begin{acknowledgments}
This research was supported by the US Department of Energy Office of Nuclear Energy, Nuclear Energy University Program award DE-NE0009417. A.M.C.\ and J.Y.\ were supported in part by the DoE ASCR Quantum Testbed Pathfinder program (awards No.\ DE-SC0019040 and No.\ DE-SC0024220) and NSF QLCI (award No.\ OMA-2120757), and J.Y.\ was also supported in part by DARPA SAVaNT ADVENT.

We thank Gonzalo Benavides, Meenakshi Krishnan, Giorgio Metafune, Connor Mooney, Ricardo Nochetto, Vasanth Pidaparthy, and Jeanie Qi for valuable discussions. In particular, we thank Ricardo for pointing us to key references on convergence analysis for finite element methods. We also thank anonymous reviewers for comments that helped to improve the paper.

 \textbf{AI use:} Claude Fable 5.1 produced a proof for \Cref{sec:uniform_superposition_initial_state_proof}, which was further simplified with the help of GPT-6 Astra and revised by the authors. Previous versions of these models were used for literature review and for technical discussions that aided in developing proofs in the rest of the paper. They were also used in generating code to run numerical experiments, as well as generating figure drafts that were edited.  The manuscript was largely prepared manually with assistance from the aforementioned models for editing and checking. 
The authors take full responsibility for the correctness of the manuscript.
\end{acknowledgments}

\bibliography{main}

\appendix
\section{Block Encoding of Hamiltonian}

In this appendix, we prove \Cref{thm:block_encoding_Hamiltonian} on the complexity of a block encoding of our final Hamiltonian. See \Cref{sec:block_encoding_hamiltonian} for an outline of the proof strategy.

\begin{theorem}[Restatement of \Cref{thm:block_encoding_Hamiltonian}]
    \label{thm:appendix_block_encoding_Hamiltonian}
    Let \(\delta\) be an error parameter and \(h\) the mesh size parameter. Then we can prepare an \begin{equation}
\parens{O \parens{ \log^2 \frac 1 h}, O \parens{ \poly \parens{ \log \frac 1 h, \log \frac 1 {\delta}}}, O \parens { \delta^{1/{2}} \cdot \poly \parens { \log \frac 1 \delta, \log \frac 1 h} }}
\end{equation}
block encoding of the Hamiltonian \(H \coloneqq C^{1/2} \parens{L+A}^{-1} C^{1/2}\) (\Cref{prob:discrete_problem}) using
\(
 O \parens{ z \cdot \poly \parens{\log \frac 1 h, \log \frac 1 {\delta}}}
 \)
 one- and two-qubit gates and
 \(
 \poly \parens{\log \frac 1 h, \log \frac 1 {\delta}}
 \)
ancillas.
\end{theorem}
\begin{proof}
Consider the following decomposition:
\begin{equation}
    \begin{aligned}
    H &= C^{1/2} \parens{L+A}^{-1} C^{1/2} \\
    &= C^{1/2} \parens{I + L^{-1} A}^{-1} L^{-1} C^{1/2} \\
    &= C^{1/2} \parens{I + \parens{F (F^T L F)^+ F^T} A}^{-1} \parens{F (F^T L F)^+ F^T} C^{1/2} \\
    &= \underbracket{\parens{ \frac C {h^3} }^{1/2}}_{1} \underbracket{\parens{I + \parens{h^{3/2} F} (F^T L F)^+ \parens{h^{3/2} F }^T \frac A {h^3}}^{-1}}_{2} \underbracket{\parens{h^{3/2}F} (F^T L F)^+ \parens{ h^{3/2} F}^T}_{3} \underbracket{\parens{ \frac C {h^3} }^{1/2}}_{4}.
    \end{aligned}
\end{equation}

We consider the terms in this decomposition separately.

\begin{enumerate}
\item For the first term and fourth term: 
\begin{enumerate}
\item \Cref{lem:loading_C_p} gives a block encoding of \(\frac{C^{(p)}}{h^3} / \norm{ \frac {C^{(p)}}{h^3} }\) with the following properties (recall that the normalization is \(O(1)\) from \Cref{thm:properties_of_C}):
\begin{equation}
\begin{aligned}
& \parens{O(1),O \parens{{1}}, {0}}\text{-block encoding} \\
 \text{ using } O& \parens{ z \cdot \poly \parens{\log \frac 1 h}} \text{ gates } \\
 \text { and } O& \parens{ \poly \parens{\log \frac 1 h}} \text{ ancillas}.
\end{aligned}
\end{equation}
\item Next, we apply the square root using \Cref{lem:positive_powers_block_encoding} (keeping track of normalization) to obtain a block encoding of \(\parens{\frac {C^{(p)}} {h^3}}^{1/2}\):
\begin{equation}
\begin{aligned}
& \parens{O(1),O \parens{ \log  
\log \frac 1 {\delta}} , \delta}\text{-block encoding} \\
 \text{ using } O& \parens{ z \cdot \poly \parens{\log \frac 1 h, \log \frac 1 {\delta}}} \text{ gates } \\
 \text { and } O& \parens{ \poly \parens{\log \frac 1 h, \log \frac 1 {\delta}}} \text{ ancillas}.
\end{aligned}
\end{equation}

\item We have a block encoding of the projector \(\Pi_C\) using \Cref{lem:projectors_for_C}:
\begin{equation}
\begin{aligned}
&(1, {O(1)}, 0)\text{-block encoding} \\
 \text{ using } O& \parens{z \cdot \poly \parens { \log \frac 1 h}} \text{ gates } \\
 \text { and } O& \parens{ \poly \parens { \log \frac 1 h}} \text{ ancillas}.
\end{aligned}
\end{equation}
\item Finally, we use the multiplication lemma \cite{Gilyen_19_QSVT_and_beyond} to obtain the final block encoding of \(\parens{\frac C {h^3} }^{1/2}{ = \Pi_{C} \parens{\frac {C^{(p)}} {h^3}}^{1/2} \Pi_{C}}\):
\begin{equation}
\begin{aligned}
& \parens{O(1),O \parens{  
\log \frac 1 {\delta}} , \delta}\text{-block encoding} \\
 \text{ using } O& \parens{ z \cdot \poly \parens{\log \frac 1 h, \log \frac 1 {\delta}}} \text{ gates } \\
 \text { and } O& \parens{ \poly \parens{\log \frac 1 h, \log \frac 1 {\delta}}} \text{ ancillas}. \\
\end{aligned}
\end{equation}
\end{enumerate}

\item For the third term:
\begin{enumerate}[ref=(\alph*)]
\item For the block encoding of \(F^TLF\), we first give a block encoding of \(\mathcal D \otimes I_3\), which we obtain from \Cref{sec:diagonal_material_matrix_BE}. It has the following properties: 
\begin{equation}
\begin{aligned}
&\parens{\Theta(1), {1}, {0}}  \text{-block encoding} \\
 \text{ using } O& \parens{ z \cdot \poly \parens{\log \frac 1 h}} \text{ gates } \\
 \text { and } O& \parens{ \poly \parens{\log \frac 1 h}} \text{ ancillas}.
\end{aligned}
\end{equation}
\item \label{item:FTLF} Next, we use this block encoding in \cite[Theorem 6.3]{deiml2025quantumrealizationfiniteelement} to obtain a block encoding of \(F^TLF\):
\begin{equation}
\begin{aligned}
& \parens{ \Theta \parens{ \log \frac 1 h},  \poly\parens{\log \frac 1 h}, \, {0} } \text{-block encoding} \\
 \text{ using } O& \parens{ z \cdot \poly \parens{\log \frac 1 h}} \text{ gates } \\
 \text { and } O& \parens{ \poly \parens{\log \frac 1 h}} \text{ ancillas}.
\end{aligned}
\end{equation}
\item From the normalization and subnormalization bounds in \cite[Theorem 6.3]{deiml2025quantumrealizationfiniteelement}, we can infer that the spectral norm of $F^TLF$ is $\Omega(1)$. Because the effective condition number of $F^TLF$ is $O(1)$ \cite{deiml2025quantumrealizationfiniteelement}, this implies the lowest eigenvalue outside the nullspace is also $\Omega(1)$. Thus, we can set $\eta$ in the Moore-Penrose pseudoinverse theorem (\Cref{thm:moore_penrose_pseudoinverse_qsvt}) as $\frac{\Omega(1)}{O(\log(1/h))}$. Thus, the value of \(m\) there is \(O(\log (1/h) \log (1/ \delta))\) for \(\delta\) error. This gives a block encoding of $\parens{F^TLF}^+$:
\begin{equation}
\begin{aligned}
& \parens{ O(1), O \parens{ \poly\parens{\log \frac 1 h}}, O(\delta)} \text{-block encoding} \\
 \text{ using } O & \parens{ z \cdot \poly \parens{\log \frac 1 h, \log \frac 1 {\delta}}} \text{ gates } \\
 \text { and } O & \parens{ \poly \parens{\log \frac 1 h, \log \frac 1 {\delta}}} \text{ ancillas}.
\end{aligned}
\end{equation}

\item For the block encodings of \(h^{3/2}F\) and \(h^{3/2}F^T\), we use the preconditioner block-encoding constructed in \Cref{thm:final_preconditioner} of  \Cref{sec:preconditioner_construction}:
\begin{equation}
\begin{aligned}
& \parens{ O \parens { \log \frac 1 h}, O \parens{ \poly\parens{\log \frac 1 h}}, 0} \text{-block encoding} \\
 \text{ using } O& \parens{\poly\parens{\log \frac 1 h}} \text{ gates } \\
 \text { and } O& \parens{ \poly\parens{\log \frac 1 h}} \text{ ancillas}.
\end{aligned}
\end{equation}
\item Finally, we use the multiplication lemma \cite{Gilyen_19_QSVT_and_beyond} to combine the above block encodings into a block encoding of \(\parens{h^{3/2}F} (F^T L F)^+ \parens{ h^{3/2} F}^T\):
\begin{equation}
\begin{aligned}
& \parens{ O \parens{ \log^2 \frac 1 h}, O \parens{ \poly\parens{\log \frac 1 h}}, O \parens{  {\delta \poly\parens{\log \frac 1 h}}}} \text{-block encoding} \\
 \text{ using } O & \parens{ z \cdot \poly \parens{\log \frac 1 h, \log \frac 1 {\delta}}} \text{ gates } \\
 \text { and } O & \parens{ \poly \parens{\log \frac 1 h, \log \frac 1 {\delta}}} \text{ ancillas}.
\end{aligned}
\end{equation}
\end{enumerate}
\item For the second term:
\begin{enumerate}
\item We obtain the block encoding of \(A/h^3\) using {\Cref{lem:A_and_C_block_encoding}}:
\begin{equation}
\begin{aligned}
& \parens{O(1),{O(1)}, {0}}\text{-block encoding} \\
 \text{ using } O& \parens{ z \cdot \poly \parens{\log \frac 1 h}} \text{ gates } \\
 \text { and } O& \parens{ \poly \parens{\log \frac 1 h}} \text{ ancillas}.
\end{aligned}
\end{equation}
\item After multiplication with the block encoding of the third term and adding \(I\) \cite{Gilyen_19_QSVT_and_beyond}, we have a block encoding of \(I + \parens{h^{3/2} F} (F^T L F)^+ \parens{h^{3/2} F }^T \frac A {h^3}\):
\begin{equation}
\begin{aligned}
& \parens{ O \parens{ \log^2 \frac 1 h}, O \parens{ \poly\parens{\log \frac 1 h}}, O \parens{ {\delta \poly\parens{\log \frac 1 h}}}} \text{-block encoding} \\
 \text{ using } O & \parens{ z \cdot \poly \parens{\log \frac 1 h, \log \frac 1 {\delta}}} \text{ gates } \\
 \text { and } O & \parens{ \poly \parens{\log \frac 1 h, \log \frac 1 {\delta}}} \text{ ancillas}.
\end{aligned}
\end{equation}
\item Because \(I\) and \(L^{-1} A\) of \((I + L^{-1}  A)\) have positive singular values, we know that \((I + L^{-1}  A)^{+} = (I + L^{-1}  A)^{-1}\). Applying the Moore-Penrose inverse \Cref{thm:moore_penrose_pseudoinverse_qsvt} along with the robustness lemma of \cite[Lemma 22]{Gilyen_19_QSVT_and_beyond}, we obtain a block encoding of \(\parens{I + \parens{h^{3/2} F} (F^T L F)^+ \parens{h^{3/2} F }^T \frac A {h^3}}^{-1}\):
\begin{equation}
\begin{aligned}
&\parens{O(1), O \parens{ \poly\parens{\log \frac 1 h}}, O \parens { \delta^{1/{2}} \cdot \poly \parens { \log \frac 1 \delta, \log \frac 1 h} }}\text{-block encoding} \\
 \text{ using } O & \parens{ z \cdot \poly \parens{\log \frac 1 h, \log \frac 1 {\delta}}} \text{ gates } \\
 \text { and } O & \parens{ \poly \parens{\log \frac 1 h, \log \frac 1 {\delta}}} \text{ ancillas.}
 \end{aligned}
\end{equation}

\end{enumerate}
\item Finally, we combine the above four block encodings to obtain a block encoding of the {full Hamiltonian, H}:
\begin{equation}
\begin{aligned}
&\parens{O \parens{ \log^2 \frac 1 h}, O \parens{ \poly \parens{ \log \frac 1 h, \log \frac 1 {\delta}}}, O \parens { \delta^{1/{2}} \cdot \poly \parens { \log \frac 1 \delta, \log \frac 1 h} }}\text{-block encoding} \\
 \text{ using } O & \parens{ z \cdot \poly \parens{\log \frac 1 h, \log \frac 1 {\delta}}} \text{ gates } \\
 \text { and } O & \parens{ \poly \parens{\log \frac 1 h, \log \frac 1 {\delta}}} \text{ ancillas } \\
 \end{aligned}
\end{equation}
\end{enumerate}

\end{proof}

\FloatBarrier

 {
\section{Block Encoding of Diagonal Material Constant Matrices} \label{sec:diagonal_material_matrix_BE}

This section describes the block-encoding strategy for a diagonal matrix, \(S_{\Sigma}\), containing the material coefficients for a 3D domain. We assume that the domain is partitioned into a \(2^{r_m} \cross 2^{r_m} \cross 2^{r_m}\) grid, where we call each of the \(2^{3r_m}\) regions in this grid a ``material cell.'' The boundaries of the material cell with index \(i,j,k\) are
\begin{align}
    ih_m &\leq x \leq (i+1) h_m \\
    jh_m &\leq y \leq (j+1) h_m \\
    kh_m &\leq z \leq (k+1) h_m
\end{align}
with \(i,j,k \in \{0, ... , 2^{r_m}-1\}\) and \(h_m = \frac{1}{2^{r_m}}\).
We then define the diagonal matrix
\begin{equation}
S_{\Sigma}
= \sum_{i,j,k = 0}^{2^{r_m}-1} \Sigma_{i,j,k} \ket{i,j,k} \bra{i,j,k} = 
\begin{bmatrix}
    \Sigma_{0,0,0} & 0 & \cdots & 0 \\
    0 & \Sigma_{0,0,1} & \ddots & \vdots \\
    \vdots & \ddots & \ddots & 0 \\
    0 & \cdots & 0 & \Sigma_{2^{r_m}-1,2^{r_m}-1,2^{r_m}-1} \\
\end{bmatrix}
\in \mathbb R^{2^{3r_m}\times 2^{3r_m}}.
\end{equation}

\begin{lemma}
Let \(\{ \mathbf{v_l}\}_{l=0}^{z-1}\) be a disjoint partition of the material cell index set into \(z\) cuboidal ``material regions'' on each of which \(\Sigma\) has a constant value \(\Sigma_{v_l}\) and \begin{equation}
    \Sigma_{\max} = \max_{l \in \{0, \ldots, z-1\}} \Sigma_{v_l}.
\end{equation} Then, \(S_{\Sigma}\) can be \((\Sigma_{\max}, 1, 0)\)-block encoded using \(O(z \poly(\log(\frac{1}{h_m}))\) one- and two-qubit gates and \(O(\poly(\log(\frac{1}{h_m}))\) additional reusable ancilla qubits.
\end{lemma}

\begin{proof}
The cuboid decomposition of the geometry allows each material region to be defined using 6 total values: an upper and lower bound material cell index for each dimension. We define \(\mathbf{v_l}\) as the set of all valid material cell index groups satisfying these conditions for the partition of index \(l\):
\begin{align}
    \mathbf{v_l} = \left\{(i,j,k)\in \{0, \ldots, 2^{r_m}-1 \}^3 \;\middle|\; 
    \begin{aligned}
        v_{l,x}^{-} &\le i < v_{l,x}^{+},\\
        v_{l,y}^{-} &\le j < v_{l,y}^{+},\\
        v_{l,z}^{-} &\le k < v_{l,z}^{+}
    \end{aligned}
    \right\}.
\end{align}

Because each dimension is assumed to have \(2^{r_m}\) valid material cell indices, these index comparisons can be performed separately across the qubit registers representing each spatial dimension. Once flag qubits representing membership of a material cell index in a material region are computed, they are used to perform a conditional \(R_y(2 \arccos(\Sigma_{v_l} / \Sigma_{\max}))\) rotation upon a single post-selection qubit. The ancilla qubits are then uncomputed and this process is repeated for each material region. \Cref{fig:three-dimensional-region-comparator} gives a representation of the circuit applying this operation, \(S_{\Sigma, v_l}\), for one material region. The circuit uses a GEQ gate which is composed of \(O(\poly(r_m)) = O(\poly(\log(\frac{1}{h_m})))\) elementary gates and reusable ancillas and performs the operation
\begin{equation}
    \mathrm{GEQ}_{i,v}\colon
    \ket{i}\ket{f}
    \longmapsto
    \ket{i}
    \ket{
        f \oplus
        \!\left[
            i \geq v
        \right]
    }.
\end{equation}

\begin{figure}[ht]
\centering
\resizebox{\textwidth}{!}{%
\begin{quantikz}[
    row sep=0.28cm,
    column sep=0.22cm
]
\lstick{$\ket{i}$}
    & \slice{$\ket{\psi_1}$}
    & \gate{\mathrm{GEQ}_{i,v_{l,x}^{-}}}
    & \gate{\mathrm{GEQ}_{i,v_{l,x}^{+}}}
    & \qw
    & \qw
    & \qw
    & \slice{$\ket{\psi_2}$}
    & \slice{$\ket{\psi_3}$}
    & \qw
    & \qw
    & \qw
    & \qw
    & \gate{\mathrm{GEQ}_{i,v_{l,x}^{+}}^\dagger}
    & \gate{\mathrm{GEQ}_{i,v_{l,x}^{-}}^\dagger}
\\
\lstick{$\ket{j}$}
    & \qw
    & \qw
    & \qw
    & \gate{\mathrm{GEQ}_{j,v_{l,y}^{-}}}
    & \gate{\mathrm{GEQ}_{j,v_{l,y}^{+}}}
    & \qw
    & \qw
    & \qw
    & \qw
    & \qw
    & \gate{\mathrm{GEQ}_{j,v_{l,y}^{+}}^\dagger}
    & \gate{\mathrm{GEQ}_{j,v_{l,y}^{-}}^\dagger}
    & \qw
    & \qw
\\
\lstick{$\ket{k}$}
    & \qw
    & \qw
    & \qw
    & \qw
    & \qw
    & \gate{\mathrm{GEQ}_{k,v_{l,z}^{-}}}
    & \gate{\mathrm{GEQ}_{k,v_{l,z}^{+}}}
    & \qw
    & \gate{\mathrm{GEQ}_{k,v_{l,z}^{+}}^\dagger}
    & \gate{\mathrm{GEQ}_{k,v_{l,z}^{-}}^\dagger}
    & \qw
    & \qw
    & \qw
    & \qw
\\
\lstick{$\ket{0}_{\mathrm{post}}$}
    & \qw
    & \qw
    & \qw
    & \qw
    & \qw
    & \qw
    & \qw
    & \gate{
        R_y\!\left(
            2\arccos(\Sigma_{v_l} / \Sigma_{\max})
        \right)
      }
    & \qw
    & \qw
    & \qw
    & \qw
    & \qw
    & \qw
\\
\lstick{$\ket{0}_{f_i}$}
    & \qw
    & \targ{} \vqw{-4}
    & \targ{} \vqw{-4}
    & \qw
    & \qw
    & \qw
    & \qw
    & \ctrl{-1}
    & \qw
    & \qw
    & \qw
    & \qw
    & \targ{} \vqw{-4}
    & \targ{} \vqw{-4}
\\
\lstick{$\ket{0}_{f_j}$}
    & \qw
    & \qw
    & \qw
    & \targ{} \vqw{-4}
    & \targ{} \vqw{-4}
    & \qw
    & \qw
    & \ctrl{-2}
    & \qw
    & \qw
    & \targ{} \vqw{-4}
    & \targ{} \vqw{-4}
    & \qw
    & \qw
\\
\lstick{$\ket{0}_{f_k}$}
    & \qw
    & \qw
    & \qw
    & \qw
    & \qw
    & \targ{} \vqw{-4}
    & \targ{} \vqw{-4}
    & \ctrl{-3}
    & \targ{} \vqw{-4}
    & \targ{} \vqw{-4}
    & \qw
    & \qw
    & \qw
    & \qw
\end{quantikz}%
}

\caption{
    Reversible circuit applying the unitary \(S_{\Sigma, v_l}\), which applies the material-dependent rotation to the post-selection qubit when \((i,j,k) \in \mathbf{v_l}\). Each pair of comparator operations XORs its result into a single flag qubit. The rotation is controlled by all three flags, after which the comparisons are uncomputed in reverse order. Some ancilla qubits are omitted from the diagram for simplicity.
}
\label{fig:three-dimensional-region-comparator}
\end{figure}

The intermediate states of the circuit applying \(S_{\Sigma, v_l}\) are
\begin{align}
    \ket{\psi_1} = \ket{i}\ket{j}\ket{k}&\ket{0}\ket{0}\ket{0}\ket{0}, \\
    \ket{\psi_2} = \ket{i}\ket{j}\ket{k}&\ket{0}\ket{v_{l,x}^{-} \leq i < v_{l,x}^{+}}\ket{v_{l,y}^{-} \leq j < v_{l,y}^{+}}\ket{v_{l,z}^{-} \leq k < v_{l,z}^{+}}, \\
    \ket{\psi_3} = \ket{i}\ket{j}\ket{k}&  \left[
\left(1-\mathbf{1_{v_l}}(i,j,k)\right)\ket{0} + \mathbf{1_{v_l}}(i,j,k) \left(\frac{\Sigma_{v_l}}{\Sigma_{\max}}\ket{0} + \sqrt{1-\left(\frac{\Sigma_{v_l}}{\Sigma_{\max}}\right)^2}\ket{1}\right)\right] \nonumber \\ & \ket{v_{l,x}^{-} \leq i < v_{l,x}^{+}}\ket{v_{l,y}^{-} \leq j < v_{l,y}^{+}}\ket{v_{l,z}^{-} \leq k < v_{l,z}^{+}}
\end{align}
where \(\mathbf{1_{v_l}}(i,j,k)\) is an indicator function that is equal to 1 if \((i,j,k)\) is in \(\mathbf{v_l}\) and 0 otherwise. The final quantum state of the circuit in \Cref{fig:three-dimensional-region-comparator} is
\begin{equation}
\ket{\psi_4}=\ket{i}\ket{j}\ket{k}\left[
\left(1-\mathbf{1_{v_l}}(i,j,k)\right)\ket{0} + \mathbf{1_{v_l}}(i,j,k) \left(\frac{\Sigma_{v_l}}{\Sigma_{\max}}\ket{0} + \sqrt{1-\left(\frac{\Sigma_{v_l}}{\Sigma_{\max}}\right)^2}\ket{1}\right)\right]\ket{0}\ket{0}\ket{0}.
\end{equation}

Repeating this circuit to apply \(S_{\Sigma, v_l}\) for each of the \(z\) material regions gives us the final state
\begin{align}
    \ket{\psi_{\mathrm{final}}}  &= \prod_l S_{\Sigma, v_l} \left[ \ket{i} \ket{j} \ket{k}\ket{0}\ket{0}\ket{0}\ket{0} \right] \\
    &= \ket{i}\ket{j}\ket{k}  \left[ \frac{\Sigma_{i,j,k}}{\Sigma_{\max}} \ket{0} +  \sqrt{1 - \left(\frac{\Sigma_{i,j,k}}{\Sigma_{\max}}\right)^2 } \ket{1} \right] \ket{0}\ket{0}\ket{0},
\end{align}
with \(\Sigma_{i,j,k}\) being the material coefficient at the material cell with indices \(i\), \(j\), and \(k\). This makes \(\left(\prod_l S_{\Sigma, v_l}\right)\) a \((\Sigma_{\max}, 1, 0)\)-block encoding of \(S_{\Sigma}\). This involves \(z\) applications of \(S_{\Sigma, v_l}\), each requiring \(O(\poly(\log(\frac{1}{h_m})))\) one- and two-qubit gates and \(O(\poly(\log(\frac{1}{h_m})))\) reusable ancilla qubits. This results in \(S_{\Sigma}\) requiring \(O(z \poly(\log(\frac{1}{h_m})))\) one- and two-qubit gates and \(O(\poly(\log(\frac{1}{h_m})))\) reusable ancilla qubits. There is no dependence on \(z\) in the ancilla qubits because they can be reused for each \(S_{\Sigma, v_l}\) application.
\end{proof}
}

\section{Block Encoding of the Projector onto Zero-Fission-Supported FEM nodes}

In this section, we construct a block encoding of a projector onto the zero-fission region, which will later be used in the block encoding of our Hamiltonian. 

\begin{lemma}[Projectors for \(C_0\)]
    \label{lem:projectors_for_C0}
    Let \(\Pi_{C0}\) be a projector onto the indices in the set \(B_0\) defining the zero-fission regions (\Cref{thm:properties_of_C}). We can implement a \((1, O(1), 0)\)-block encoding of \(\Pi_{C0}\) using \(O(z \cdot \poly(\log(1/h)))\) one- and two-qubit gates and \(O(\poly(\log(1/h)))\) ancillas.
\end{lemma}

\begin{proof}
    \(\Pi_{C0}\) is a \((\frac{1}{h} - 1)^3 \cross (\frac{1}{h} - 1)^3\) diagonal matrix where \(\Pi_{C0,ijk,ijk} = 1\) if all 8 FEM cells constituting the support of \(\varphi_{ijk}\) have a zero fission cross section and \(\Pi_{C0,ijk,ijk} = 0\) otherwise. Specifically,
    \begin{equation} \label{eq:Pi_C0_definition}
        \Pi_{C0,ijk,ijk} = \begin{cases} 
        1, \quad \nu \Sigma_{f,i+\delta_i, j+\delta_j, k+\delta_k} = 0 \quad \forall \, \delta_i, \delta_j, \delta_k \in \{0,1\} \\
        0, \quad \text{otherwise}.
    \end{cases}
    \end{equation}
    We define a new constant for each FEM cell   
    \begin{equation}
        f_{ijk} = 
    \begin{cases}
        1, \quad \nu\Sigma_{f,ijk} = 0 \\
        0, \quad \nu\Sigma_{f,ijk} \neq 0
    \end{cases}
    \end{equation}
    with \(i,j,k \in \left[0, \frac{1}{h}-1 \right]\). We can \((1, 1, 0)\)-block encode the diagonal projector \(P^f\) defined with    
    \begin{equation}
        P^{f}_{ijk,ijk} = f_{i,j,k}
    \end{equation}
    using \(O(z \poly(\log(\frac{1}{h})))\) one- and two-qubit gates and \(O(\poly(\log(\frac{1}{h})))\) reusable ancilla qubits using the exact same strategy as from \Cref{sec:diagonal_material_matrix_BE} because \(f_{i,j,k}\) is defined using the same material region boundary sets, \(\mathbf{v_l}\), as are used to define \(\nu\Sigma_{f,ijk}\), \(\Sigma_{a,ijk}\), and \(D_{ijk}\).
    
    We now define eight \((\frac{1}{h} - 1)^3 \cross (\frac{1}{h} - 1)^3\) diagonal projector matrices, \(P^{\delta_i, \delta_j, \delta_k}\), one for each combination of binary states of \(\delta_i\), \(\delta_j\), and \(\delta_k\). The elements of these are
    \begin{equation}
        P^{\delta_i, \delta_j, \delta_k}_{ijk,ijk} = f_{i+\delta_i,j+\delta_j,k+\delta_k}
    \end{equation}
    with \(i,j,k \in \left\{0, \ldots, \frac{1}{h}-2 \right\}\). \(P^{\delta_i, \delta_j, \delta_k}\) contains a subset of the values of \(P^f\), so we can show that
    \begin{equation} \label{eq:zero_fission_permutation_projector}
        \Pi_n Q_{\delta_i, \delta_j, \delta_k}^T P^f Q_{\delta_i, \delta_j, \delta_k} \Pi_n
    \end{equation}
    is a \((1, O(1), 0)\)-block encoding of \(P^{\delta_i, \delta_j, \delta_k}\). Here \(Q_{\delta_i, \delta_j, \delta_k}\) is a permutation matrix such that \(Q_{\delta_i, \delta_j, \delta_k}^T P^f Q_{\delta_i, \delta_j, \delta_k}\) is a diagonal matrix that contains the elements of \(P_f\) corresponding to the indices \(i = \{\delta_i, \frac{1}{h} - 1 + \delta_i\}\), \(j = \{\delta_j, \frac{1}{h} - 1 + \delta_j\}\), and \(k = \{\delta_k, \frac{1}{h} - 1 + \delta_k\}\) in the top left corner. Also, \(\Pi_n = \sum_{t=0}^{(\frac{1}{h} - 1)^3} \ket{t} \bra{t}\) is a projector onto the first \((\frac{1}{h} - 1)^3\) elements of the \(\frac{1}{h}^3\)-sized system. The unitary operation\(Q_{\delta_i, \delta_j, \delta_k}\) can be applied using elementary arithmetic operations (that will not be explicitly defined here) using \(\poly(\log(\frac{1}{h}))\) one- and two-qubit gates and ancillas, and the projector \(\Pi_n\) can be \((1,1,0)\) block-encoded using simple integer comparison operations also requiring \(\poly(\log(\frac{1}{h}))\) one- and two-qubit gates and ancillas. Therefore \Cref{eq:zero_fission_permutation_projector} is a \((1, O(1), 0)\) zero-padded block-encoding of \(P^{\delta_i, \delta_j, \delta_k}\) using \(O(z \poly(\log(\frac{1}{h})))\) one- and two-qubit gates and \(O(\poly(\log(\frac{1}{h})))\) reusable ancilla qubits.

We then use this to define
\begin{equation} \label{eq:multiplication_of_zero_fission_operators}
    \Pi_{C0} = \prod_{\delta_i,\delta_j,\delta_k \in \{0,1\} } P^{\delta_i, \delta_j, \delta_k},
\end{equation}
where we define each element on this diagonal projector as 
\begin{align}
    \Pi_{C0,ijk,ijk} = \prod_{\delta_i,\delta_j,\delta_k \in \{0,1\} } f_{i+\delta_i,j+\delta_j,k+\delta_k} &= 
    \begin{cases}
        &1, \quad f_{i+\delta_i,j+\delta_j,k+\delta_k} = 1 \quad \forall \, \delta_i,\delta_j,\delta_k \in \{0,1\} \\
        &0, \quad \text{otherwise}
    \end{cases}
    \\
     &= 
    \begin{cases}
        &1, \quad \nu\Sigma_{f,i+\delta_i,j+\delta_j,k+\delta_k} = 0 \quad \forall \, \delta_i,\delta_j,\delta_k \in \{0,1\} \\
        &0, \quad \text{otherwise},
    \end{cases}
\end{align}
which is exactly the definition of \(\Pi_{C0}\) from \Cref{eq:Pi_C0_definition}. Using the multiplication lemma from \cite{Gilyen_19_QSVT_and_beyond} to apply \Cref{eq:multiplication_of_zero_fission_operators}, \(\Pi_{C0}\) can be \((1,O(1),0)\)-block encoded using \(O(z \poly(\log(\frac{1}{h})))\) one- and two-qubit gates and \(O(\poly(\log(\frac{1}{h})))\) ancilla qubits as claimed.
\end{proof}

\begin{lemma}[Projectors for C]
    \label{lem:projectors_for_C}
    Let \(\Pi_{C}\) be a projector onto the non-zero subspace of \(C\) (\Cref{thm:properties_of_C}). We can implement a \((1, O(1), 0)\)-block encoding of \(\Pi_{C}\) using \(O(z \cdot \poly(\log(1/h)))\) one- and two-qubit gates and \(O(\poly(\log(1/h)))\) ancillas.
\end{lemma}

\begin{proof}
    This can be done by first recognizing that \(\Pi_C = I - \Pi_{C0}\). Let \(U\) the \((1, q = O(1), 0)\)-block encoding of \(\Pi_{C0}\) found in \Cref{lem:projectors_for_C0}. For an arbitrary input state \(\ket{\psi}\), we can separate it into components within the space that \(\Pi_{C0}\) projects onto and all other components:
    \begin{align}
        \ket{\psi} = \underbracket{\Pi_{C0}\ket{\psi}}_{\ket{\psi_p}} + \underbracket{(I - \Pi_{C0})\ket{\psi}}_{\ket{\psi_{\perp}}}.
    \end{align}
    These two components are orthogonal and generally unnormalized.

    Since \(U\) is an exact normalization-\(1\) block encoding and \(\Pi_{C0}\ket{\psi_p}=\ket{\psi_p}\), the component of \(U(\ket{0^q}_Q\ket{\psi_p})\) having all zeros in the \(Q\) register is \(\ket{0^q}_Q\ket{\psi_p}\). Unitarity therefore leaves no component supported on nonzero states of the \(Q\) register, so
    \begin{equation}
        U(\ket{0^q}_Q\ket{\psi_p})=\ket{0^q}_Q\ket{\psi_p}.
    \end{equation}
    Define
    \begin{equation}
        \ket{G_{\perp}} := U(\ket{0^q}_Q\ket{\psi_{\perp}}).
    \end{equation}
    Since \(\Pi_{C0}\ket{\psi_{\perp}}=0\), this component satisfies
    \begin{equation}
        (\bra{0^q}_Q\otimes I)\ket{G_{\perp}}=0.
    \end{equation}
    Thus, \(\ket{G_{\perp}}\) is supported only on states with at least one \(\ket{1}\) in the \(Q\) register.

    Applying \(U\) to \(\ket{\psi}\) and the post-selection register \(Q\) gives us
    \begin{align}
        U(\ket{0^q}_Q \ket{\psi})
        &= U\bigl(\ket{0^q}_Q (\Pi_{C0}\ket{\psi} + (I - \Pi_{C0})\ket{\psi})\bigr) \\
        &= \ket{0^q}_Q \ket{\psi_p} + \ket{G_{\perp}}.
    \end{align}

    We then add another single qubit \(\ket{0}_s\) and apply a multi-controlled \(X\) gate, denoted by \(\mathrm{MCX}_{Q=0^q\to s}\), controlled on all of the qubits in the \(Q\) register being \(\ket{0}\) and targeting the qubit \(s\). This gives us the state
    \begin{equation}
        \ket{1}_s\ket{0^q}_Q \ket{\psi_p} + \ket{0}_s\ket{G_{\perp}}.
    \end{equation}

    Applying \(I_s\otimes U^{\dagger}\) then gives us the state
    \begin{equation}
        \ket{1}_s\ket{0^q}_Q \ket{\psi_p} + \ket{0}_s\ket{0^q}_Q\ket{\psi_{\perp}},
    \end{equation}
    because \(U^{\dagger}(\ket{0^q}_Q\ket{\psi_p})=\ket{0^q}_Q\ket{\psi_p}\) and \(U^{\dagger}\ket{G_{\perp}}=\ket{0^q}_Q\ket{\psi_{\perp}}\).

    Therefore, the unitary
    \begin{equation}
        W := (I_s\otimes U^{\dagger})\,
        \mathrm{MCX}_{Q=0^q\to s}\,
        (I_s\otimes U)
    \end{equation}
    satisfies
    \begin{equation}
        (\bra{0}_s\bra{0^q}_Q\otimes I)
        W
        (\ket{0}_s\ket{0^q}_Q\otimes I)
        = I-\Pi_{C0}=\Pi_C.
    \end{equation}
    Hence, \(W\) is a \((1,q+1,0)\)-block encoding of \(\Pi_C\), with \(q+1=O(1)\). It uses one application each of \(U\) and \(U^{\dagger}\), together with a multi-controlled \(X\) gate on \(q\) controls. The total cost is therefore \(O(z \poly(\log(1/h)))\) one- and two-qubit gates and \(O(\poly(\log(1/h)))\) ancillas, as claimed.
\end{proof}

{
\section{Block Encoding of the Absorption and Fission Matrices}
\label{sec:block-encoding-of-absorption-and-fission-matrices}
\label{sec:appendix_absorption_fission}

In this section, we present a strategy (based on \cite{deiml2025quantumrealizationfiniteelement}) for efficiently block encoding the
absorption and fission matrices \(A\) and \(C\) defined in
\Cref{prob:discrete_problem}. Recall that
\begin{align}
A_{ijk,i'j'k'}
&=
\int_{\Omega}
\Sigma_a(\bx)
\varphi_{ijk}(\bx)
\varphi_{i'j'k'}(\bx)
\,\dd{\bx},
\\
C_{ijk,i'j'k'}
&=
\int_{\Omega}
\nu\Sigma_f(\bx)
\varphi_{ijk}(\bx)
\varphi_{i'j'k'}(\bx)
\,\dd{\bx}.
\end{align}
Here \(\varphi_{ijk}\) are the nodal basis functions introduced in
\Cref{def:nodal_basis}.

The domain contains
\begin{equation}
M=\parens{\frac{1}{h}}^3
\end{equation}
mesh cells and
\begin{equation}
N=\parens{\frac{1}{h}-1}^3
\end{equation}
interior nodes. By assumption, each mesh cell is contained entirely within a
single material region, so \(\Sigma_a\) and \(\nu\Sigma_f\) are constant on
each cell.

Let
\begin{equation}
S_a
=
\!\left(
\Sigma_a|_{c_{111}},
\Sigma_a|_{c_{112}},
\ldots
\right)
\in \mathbb R^{M\times M}
\end{equation}
be the diagonal matrix containing the value of the absorption cross section
on each mesh cell. Similarly, define
\begin{equation}
S_f
=
\!\left(
\nu\Sigma_f|_{c_{111}},
\nu\Sigma_f|_{c_{112}},
\ldots
\right)
\in \mathbb R^{M\times M}.
\end{equation}

These matrices can be \((\Sigma_{a,\max}, 1, 0)\)-block encoded and \((\nu \Sigma_{f,\max}, 1, 0)\)-block encoded, respectively, using \(O(z \poly(\log \frac{1}{h}))\) elementary gates and \(O(\poly(\log\frac{1}{h}))\) additional reusable ancilla qubits using the method from \Cref{sec:diagonal_material_matrix_BE}.

We first construct a change of basis in one dimension and then take its
tensor product to obtain the corresponding three-dimensional transformation.

\subsection{One-Dimensional Change of Basis}

Let the one-dimensional mesh cells be
\begin{equation}
e_p=[(p-1)h,ph],
\qquad
p=1,\ldots,\frac{1}{h}.
\end{equation}
On each cell \(e_p\), define two functions
\begin{align}
\Psi_{p,0}(x)
&=
\begin{cases}
h^{-1/2},
& x\in e_p,\\
0,
& \text{otherwise},
\end{cases}
\\
\Psi_{p,1}(x)
&=
\begin{cases}
\displaystyle
\sqrt{\frac{3}{h}}
\left(
\frac{2(x-(p-1)h)}{h}-1
\right),
& x\in e_p,\\
0,
& \text{otherwise}.
\end{cases}
\end{align}

These functions are orthonormal in \(L^2([0,1])\):
\begin{equation}
\int_0^1
\Psi_{p,s}(x)\Psi_{p',s'}(x)\,\dd{x}
=
\delta_{p,p'}\delta_{s,s'}.
\end{equation}

Let
\begin{equation}
Q_h^1
=
\spanop{
\Psi_{p,s}
:
p=1,\ldots,\frac{1}{h},
\quad
s\in\{0,1\}
}.
\end{equation}

The \(\Psi_{p,s}\) basis functions are represented graphically in Figure \ref{fig:gradient_basis_functions}.

\begin{figure}[h]
\centering
\begin{tikzpicture}[scale=1.2]

    \draw[->] (0,0) -- (8.5,0) node[right] {$x$};
    \draw[->] (0,-2.0) -- (0,2.0) node[above] {$\Psi_{k}(x)$};

    \foreach \x/\lab in {0/$0$,2/$0.25$,4/$0.5$,6/$0.75$,8/$1$} {
        \draw[dotted] (\x,-1.732) -- (\x,1.732);
        \draw (\x,0.05) -- (\x,-0.05)
            node[below] {\lab};
    }

    \draw[dashed] (0,1.732) -- (8.1,1.732);
    \draw[dashed] (0,-1.732) -- (8.1,-1.732);

    \node[left] at (0,1.732) {$\sqrt{3} h^{-1/2}$};
    \node[left] at (0,1) {$h^{-1/2}$};
    \node[left] at (0,0) {$0$};
    \node[left] at (0,-1.732) {$\sqrt{3} h^{-1/2}$};

    \foreach \i in {0,2,4,6} {
        \draw[blue, thick] (\i,1) -- (\i+2,1);
        \node[blue, above] at (\i+1,1)
            {$\Psi_{\i}$};
    }

    \foreach \i in {0,2,4,6} {
        \pgfmathtruncatemacro{\idx}{\i+1}
        \draw[red, thick] (\i,-1.732) -- (\i+2,1.732);
        \node[blue] at (\i+0.7,-1.25)
            {};
    }

\end{tikzpicture}
\caption{Basis functions \(\Psi_{p,s}\) spanning the \(Q_h^1\) space for \(h = \frac{1}{4}\).}
\label{fig:gradient_basis_functions}
\end{figure}

Since every one-dimensional piecewise-linear finite element function is
linear on each mesh cell, we have
\begin{equation}
V_0^h([0,1])\subseteq Q_h^1.
\end{equation}

For a function \(f_h\in V_0^h([0,1])\), write
\begin{equation}
f_h(x)
=
\sum_{i=1}^{1/h-1} f_i\varphi_i(x).
\end{equation}
Define the vectors
\begin{align}
\bm\varphi(x)
&=
\begin{bmatrix}
\varphi_1(x) &
\varphi_2(x) &
\cdots &
\varphi_{1/h-1}(x)
\end{bmatrix}^T,
\\
\bm\Psi(x)
&=
\begin{bmatrix}
\Psi_{1,0}(x) &
\Psi_{1,1}(x) &
\Psi_{2,0}(x) &
\Psi_{2,1}(x) &
\cdots &
\Psi_{1/h,1}(x)
\end{bmatrix}^T,
\\
f
&=
\begin{bmatrix}
f_1 &
f_2 &
\cdots &
f_{1/h-1}
\end{bmatrix}^T.
\end{align}

We define \(R_{L,1D}\) by
\begin{equation}
\label{eq:R_h_definition}
\bm\varphi(x)^T
=
\bm\Psi(x)^T R_{L,1D}.
\end{equation}
Therefore,
\begin{equation}
\label{eq:R_h_function_expansion}
f_h(x)
=
\bm\varphi(x)^T f
=
\bm\Psi(x)^T R_{L,1D} f.
\end{equation}

The matrix \(R_{L,1D}\in
\mathbb R^{(2/h)\times(1/h-1)}\) is
\begin{align} \label{eq:R_L_1D_matrix}
    R_{L,1D} &= \sqrt{h} \underbracket{\begin{bmatrix}
        \frac{1}{2} & 0 & \cdots & 0  \\
        \frac{1}{2 \sqrt{3}} & 0 & \cdots & 0 \\
        \frac{1}{2} & \frac{1}{2} & \cdots & 0 \\
        -\frac{1}{2 \sqrt{3}} & \frac{1}{2 \sqrt{3}} & \cdots & 0 \\
        0 & \frac{1}{2} & \cdots & 0 \\
        0 & -\frac{1}{2 \sqrt{3}} & \cdots & 0 \\
        \vdots & \vdots & \ddots & 0 \\
        0 & 0 & \cdots & \frac{1}{2} \\
        0 & 0 & \cdots & \frac{1}{2 \sqrt{3}} \\
        0 & 0 & \cdots & \frac{1}{2} \\
        0 & 0 & \cdots & -\frac{1}{2 \sqrt{3}} \\
    \end{bmatrix}}_{\quad (2/h) \cross (1/h - 1)} \nonumber =  \sqrt{h} \parens{I_{1/h} \otimes \begin{bmatrix}
        \frac{1}{2} &  \frac{1}{2} \\
        -\frac{1}{2 \sqrt{3}} &  \frac{1}{2 \sqrt{3}} \\
    \end{bmatrix}} \underbracket{\begin{bmatrix}
        0 & 0 & 0 & \cdots & 0 \\
        1 & 0 & 0 & \cdots & 0 \\
        1 & 0 & 0 & \cdots & 0 \\
        0 & 1 & 0 & \cdots & 0 \\
        0 & 1 & 0 & \cdots & 0 \\
        \vdots & \vdots & \vdots & \ddots & 0 \\
        0 & 0 & 0 & \cdots & 1 \\
        0 & 0 & 0 & \cdots & 1 \\
        0 & 0 & 0 & \cdots & 0 \\
    \end{bmatrix}}_{M_2, \quad (2/h) \cross (1/h - 1)}.
\end{align}

It is useful to remove the \(\sqrt h\) factor and define
\begin{equation}
\widetilde R_{L,1D}
\coloneqq
h^{-1/2}R_{L,1D}.
\end{equation}

\begin{lemma} \label{lemma:R_L1D_block_encoding}
An 
\(\left(O(1), 1, 0 \right) \)-block encoding of \(\widetilde R_{L,1D}\) can be implemented using
\(\poly(\log(1/h))\) one- and two-qubit gates and ancillas.
\end{lemma}

\begin{proof}
The \(M_2\) matrix performs the simple arithmetic operation
\begin{equation}
    M_2 \ket{i} = \ket{2i + 1} + \ket{2i + 2}
\end{equation}
for all \(0 \leq i < \frac{1}{h} - 1\). This allows \(M_2\) to be \((O(1), 1, 0)\)-block encoded using \(\poly(\log(1/h))\) one- and two-qubit gates and ancillas. The entries of \(\widetilde R_{L,1D}\) are independent of \(h\), and \(\norm{\widetilde R_{L,1D}}=O(1)\). \(I_N\) and the \(2 \cross 2\) matrix in the \(\widetilde R_{L,1D}\) decomposition are both unitaries that can both be implemented exactly in \(O(1)\) time, so the lemma is proven.
\end{proof}

\subsection{Three-Dimensional Transformation}

We now extend the construction to the three-dimensional finite element space
\(V_0^h(\Omega)\).
For each mesh cell \(c_{pqr}\) and
\(\bm s=(s_1,s_2,s_3)\in\{0,1\}^3\), define
\begin{equation}
\Psi_{pqr,\bm s}(x,y,z)
=
\Psi_{p,s_1}(x)
\Psi_{q,s_2}(y)
\Psi_{r,s_3}(z).
\end{equation}
The resulting functions form an orthonormal basis for the space
\begin{equation}
Q_h^3
=
Q_h^1\otimes Q_h^1\otimes Q_h^1.
\end{equation}
and the vector of basis functions spanning \(Q_{h}^{3}\) is also separable into its 1D components as 
\begin{equation}
    \bm\Psi^3(x) = \bm\Psi^1(x) \otimes \bm\Psi^1(x) \otimes \bm\Psi^1(x).
\end{equation}
In particular,
\begin{equation}
\int_\Omega
\Psi_{pqr,\bm s}(\bx)
\Psi_{p'q'r',\bm s'}(\bx)
\,\dd{\bx}
=
\delta_{p,p'}
\delta_{q,q'}
\delta_{r,r'}
\delta_{\bm s,\bm s'}.
\end{equation}

Let \(\pi_l\) be the permutation that changes the cell-first ordering
\begin{equation}
\ket{p}\ket{q}\ket{r}
\ket{s_1}\ket{s_2}\ket{s_3}.
\end{equation}
to the tensor-product ordering
\begin{equation}
\ket{p}\ket{s_1}
\ket{q}\ket{s_2}
\ket{r}\ket{s_3}.
\end{equation}
Define
\begin{equation}
\label{eq:T_h_definition}
T_h
\coloneqq
\pi_l^T
\left(
R_{L,1D}\otimes R_{L,1D}\otimes R_{L,1D}
\right).
\end{equation}

Thus, if
\begin{equation}
\phi_h(\bx)
=
\sum_{i,j,k=1}^{1/h-1}
u_{ijk}\varphi_{ijk}(\bx)
\end{equation}
has coefficient vector \(u\in\mathbb R^N\), then \(\pi_l T_h u\) is its
coefficient vector in the orthonormal cell-local basis
\(\{\Psi_{pqr,\bm s}\}\):
\begin{equation}
\label{eq:T_h_function_expansion}
\phi_h(\bx)
=
\bm\Psi^3(\bx)^T \pi_l T_h u.
\end{equation}

Since each factor \(R_{L,1D}\) contributes a factor \(\sqrt h\), define
\begin{equation}
\label{eq:T_h_tilde_definition}
\widetilde T_h
\coloneqq
h^{-3/2}T_h
=
\pi_l^T
\left(
\widetilde R_{L,1D}
\otimes
\widetilde R_{L,1D}
\otimes
\widetilde R_{L,1D}
\right).
\end{equation}
\begin{lemma} \label{lemma:T_h_block_encoding}
\(\widetilde T_h\) can be \(\left(O(1), O(1), 0\right) \)-block encoded using \(O(\poly(\log(1/h)))\) one- and two-qubit gates and
ancillas.
\end{lemma}
    
\begin{proof}
    \(\widetilde R_{L,1D}\) can be \(\left(O(1), 1, 0\right) \)-block encoded \(\poly(\log(1/h))\) one- and two-qubit gates and
    ancillas. \(\pi_l^T\) is a simple permutation matrix that can be constructed also using \(O(\poly(\log(1/h)))\) one- and two-qubit gates and no ancillas or post-selection qubits. Therefore, because \(\widetilde T_h\) is constructed using a constant number (three) \(\widetilde R_{L,1D}\) matrices and one \(\pi_l^T\) matrix, the product of block-encodings lemma from
\cite{Gilyen_19_QSVT_and_beyond} can be used and the lemma claim is proven.
\end{proof}

\subsection{Decomposition of \(A\) and \(C\)}

We now express the absorption and fission matrices in terms of \(T_h\).

\begin{lemma}
\label{lem:A_C_transformation_decomposition}
The matrices \(A\) and \(C\) defined in \Cref{prob:discrete_problem} satisfy
\begin{align}
\label{eq:A_transformation_decomposition}
A
&=
T_h^T
\left(
S_a\otimes I_8
\right)
T_h,
\\
\label{eq:C_transformation_decomposition}
C
&=
T_h^T
\left(
S_f\otimes I_8
\right)
T_h.
\end{align}
Equivalently,
\begin{align}
\label{eq:A_scaled_transformation_decomposition}
\frac{A}{h^3}
&=
\widetilde T_h^T
\left(
S_a\otimes I_8
\right)
\widetilde T_h,
\\
\label{eq:C_scaled_transformation_decomposition}
\frac{C}{h^3}
&=
\widetilde T_h^T
\left(
S_f\otimes I_8
\right)
\widetilde T_h.
\end{align}
\end{lemma}

\begin{proof}
We prove \Cref{eq:A_transformation_decomposition}; the proof of
\Cref{eq:C_transformation_decomposition} is identical after replacing
\(\Sigma_a\) by \(\nu\Sigma_f\).

Let \(e_\alpha,e_\beta\in\mathbb R^N\) denote standard basis vectors
corresponding to two finite element nodal basis functions
\(\varphi_\alpha\) and \(\varphi_\beta\). By
\Cref{eq:T_h_function_expansion},
\begin{align}
\varphi_\alpha(\bx)
&=
\bm\Psi^3(\bx)^T \pi_l T_h e_\alpha,
\\
\varphi_\beta(\bx)
&=
\bm\Psi^3(\bx)^T \pi_l T_h e_\beta.
\end{align}
Therefore,
\begin{align}
A_{\alpha,\beta}
&=
\int_\Omega
\Sigma_a(\bx)
\varphi_\alpha(\bx)
\varphi_\beta(\bx)
\,\dd{\bx}
\\
&=
\int_\Omega
\Sigma_a(\bx)
\left(
\bm\Psi^3(\bx)^T \pi_l T_h e_\alpha
\right)
\left(
\bm\Psi^3(\bx)^T \pi_l T_h e_\beta
\right)
\,\dd{\bx}.
\end{align}

Write
\begin{equation}
T_h e_\alpha
=
\sum_{p,q,r}
\sum_{\bm s\in\{0,1\}^3}
t^{(\alpha)}_{pqr,\bm s}
\ket{p,q,r,\bm s},
\end{equation}
and similarly for \(T_he_\beta\). Then
\begin{align}
A_{\alpha,\beta}
&=
\sum_{p,q,r}
\sum_{\bm s}
\sum_{p',q',r'}
\sum_{\bm s'}
t^{(\alpha)}_{pqr,\bm s}
t^{(\beta)}_{p'q'r',\bm s'}
\nonumber\\
&\qquad\times
\int_\Omega
\Sigma_a(\bx)
\Psi_{pqr,\bm s}(\bx)
\Psi_{p'q'r',\bm s'}(\bx)
\,\dd{\bx}.
\end{align}

Each \(\Psi_{pqr,\bm s}\) is supported only on the cell \(c_{pqr}\).
Moreover, \(\Sigma_a\) is constant on that cell. Using orthonormality,
\begin{equation}
\int_\Omega
\Sigma_a(\bx)
\Psi_{pqr,\bm s}(\bx)
\Psi_{p'q'r',\bm s'}(\bx)
\,\dd{\bx}
=
\Sigma_a|_{c_{pqr}}
\delta_{p,p'}
\delta_{q,q'}
\delta_{r,r'}
\delta_{\bm s,\bm s'}.
\end{equation}
It follows that
\begin{align}
A_{\alpha,\beta}
&=
\sum_{p,q,r}
\sum_{\bm s}
\Sigma_a|_{c_{pqr}}\,
t^{(\alpha)}_{pqr,\bm s}
t^{(\beta)}_{pqr,\bm s}
\\
&=
e_\alpha^T
T_h^T
\left(
\mathcal S_a\otimes I_8
\right)
T_h
e_\beta.
\end{align}
Since this equality holds for every pair \(\alpha,\beta\),
\begin{equation}
A
=
T_h^T
\left(
S_a\otimes I_8
\right)
T_h.
\end{equation}

Replacing \(\Sigma_a\) by \(\nu\Sigma_f\) gives
\begin{equation}
C
=
T_h^T
\left(
S_f\otimes I_8
\right)
T_h.
\end{equation}

Finally, since \(T_h=h^{3/2}\widetilde T_h\),
\begin{align}
\frac{A}{h^3}
&=
\widetilde T_h^T
\left(
S_a\otimes I_8
\right)
\widetilde T_h,
\\
\frac{C}{h^3}
&=
\widetilde T_h^T
\left(
S_f\otimes I_8
\right)
\widetilde T_h,
\end{align}
as claimed.
\end{proof}

\begin{lemma} \label{lem:A_and_C_block_encoding}
    \(\frac{A}{h^3}\) can be \((O(\Sigma_{a,\max}), O(1), 0)\)-block encoded and \(\frac{C}{h^3}\) can be \((O(\nu\Sigma_{f,\max}), O(1), 0)\)-block encoded, both using \(O(z \poly(\log(1/h)))\) one- and two-qubit gates and \(O(\poly(\log(1/h)))\) ancillas.
\end{lemma}
\begin{proof}
We can combine the block encodings of \(S_f\) and \(S_a\) from \Cref{sec:diagonal_material_matrix_BE} with the block encoding of
\(\widetilde T_h\) and its transpose from \Cref{lemma:T_h_block_encoding} and use the multiplication lemma from
\cite{Gilyen_19_QSVT_and_beyond} to achieve the stated result.
\end{proof}

\begin{lemma}[Block encoding of \(C^{(p)}/h^3\)]
    \label{lem:loading_C_p}
     \(C^{(p)}/h^3\) is the same as \(C/h^3\) except the zero block (described in \Cref{thm:properties_of_C}) is replaced by an identity block. An \((O(\nu\Sigma_{f,\max}), O(1), 0)\)-block encoding of \(C^{(p)}/h^3\) can be constructed using \(O(z \cdot \poly(\log(1/h)))\) one- and two-qubit gates and \(O(\poly(\log(1/h)))\) ancillas. 
\end{lemma}
\begin{proof}
\(C^{(p)}/h^3\) is simply the sum of the \(C/h^3\) matrix block encoded in \Cref{lem:A_and_C_block_encoding} and the \(\Pi_{C0}\) matrix block encoded in \Cref{lem:projectors_for_C0}. \(C/h^3\) can be \((O(\nu\Sigma_{f,\max}), O(1), 0)\)-block encoded using \(\poly(\log(1/h))\) one- and two-qubit gates and ancillas and \(\Pi_{C0}\) can be \((1, O(1), 0)\)-block encoded using \(O(z \cdot \poly(\log(1/h)))\) one- and two-qubit gates and \(O(\poly(\log(1/h)))\) ancillas. Combining these two matrices using LCU gives the stated result.
\end{proof}

}

\section{Initial State Preparation} \label{sec:uniform_superposition_initial_state_proof}
In this section, we show that the uniform superposition state \(s\) can be used as the initial state for phase estimation as the overlap between \(s\) and \(\hat v_{h_f}\) tends to a constant \(\theta\) independent of the fine mesh size \(h\) as \(h \rightarrow 0\). We then show via numerical experiments in \Cref{sec:uniform_superposition_overlap_experiments} that (1) the success probability \(|\theta|^2\) is large in practice (above \(0.2\) in all our experiments) and (2) appears to be converging to a constant value within the range of degrees of freedom tested.

It is convenient to first prove the following.
\begin{lemma} 
    \label{lem:simple_positive_eigenfunction}
    Assuming the minimum eigenvalue \(\lambda\) of \Cref{prob:orig_prob_weak_form} exists, it is simple, and the corresponding eigenfunction \(\phi\) can be chosen to be positive almost everywhere.
\end{lemma}
\begin{proof}
Consider the bilinear form $c(u, v) \coloneq a(u, v) - \lambda b(u, v)$, where $a$ and $b$ are as defined in \Cref{prob:orig_prob_weak_form}. Note that $0$ is an eigenvalue as $c(\phi, \phi) = a(\phi, \phi) - \lambda b(\phi, \phi) = 0$. Furthermore, since $\lambda$ is the minimum eigenvalue of \Cref{prob:orig_prob_weak_form}, we have $c(u, u) = a(u, u) - \lambda b(u, u) \ge 0$ for all $u$, so all eigenvalues of $c$ are nonnegative. Thus, $0$ is the minimum eigenvalue of $c$. By \cite[Theorem 8.38]{Gilbarg2001}, the minimum eigenvalue of $c$ is simple and has a positive eigenfunction.
Since the equations $a(\phi, \psi) = \lambda b(\phi, \psi)$ and $c(\phi, \psi) = 0$ are equivalent, the $\lambda$-eigenspace of the former is equal to the $0$-eigenspace of the latter. In particular, they are both one-dimensional eigenspaces with a positive eigenfunction.
\end{proof}

Now we state the main theorem of this section, which shows that it is sufficient to take a uniform superposition state as the initial state for phase estimation. 

\begin{theorem}
    \label{thm:uniform_initial_state_preparation}
    Let \(s = \frac {1}{ \sqrt N} \sum_{i=0}^{N-1} \ket i\) be the uniform superposition over the \(N = \parens{\frac 1 h - 1}^d\) fine mesh nodes. Let \(u_h\) be an eigenvector of \( (L_{h} + A_{h})^{-1} C_{h}\) with maximum eigenvalue \(k_h = 1/\lambda_h\); thus \(\hat v_h = \frac {C^{1/2}_{h}  u_h}{\norm{C^{1/2}_{h}  u_h}}\) is an eigenvector of \(C^{1/2}_{h} (L_{h} + A_{h})^{-1} C^{1/2}_{h}\) with maximum eigenvalue \(k_h\). Let \(f = \nu \Sigma_f\), and let
    \begin{equation}
        \theta = \frac{\norm{\sqrt {f} \phi}_{L^1(\Omega)}}{\norm{\sqrt {f} \phi}_{L^2(\Omega)}},
    \end{equation}
    where \(\phi\) is the minimum eigenfunction of \Cref{prob:orig_prob_weak_form} as given in \Cref{lem:simple_positive_eigenfunction}, chosen to be positive and normalized such that \(\norm{\phi}_a = 1\).

    Then, we have 
    \begin{equation}
     \abs{\abs{\braket{s}{\hat v_h}} - \theta} = O \parens{h^{\gamma/(2\pi)}}.
    \end{equation}
\end{theorem}
\begin{proof}
Let \(\phi_h\) be the eigenfunction of \Cref{prob:finite-element-version} corresponding to coefficient vector \(u_h\), normalized such that \(\norm{\phi_h}_a = 1\), and taken in the direction such that \(a(\phi, \phi_h) \geq 0\). Let us define \(\theta_h = \frac{\
\int_\Omega \sqrt f \phi_h}{\sqrt{\int_\Omega \left ( \sqrt f \phi_h \right )^2}} \).

We prove the theorem in two steps. First, we show in \Cref{lem:overlap_theta_h_convergence} that
\begin{equation}
\abs{\braket{s}{\hat v_h} - \theta_h} = O \parens{h^{1/2}}.
\end{equation}
Next, we show in \Cref{lem:theta_theta_h_convergence} that
\begin{equation}
\abs{\theta - \theta_h} = O \parens{ h^{\gamma/(2\pi)}}.
\end{equation}
Combining these two, we have
\begin{equation}
\abs{\abs{\braket{s}{\hat v_h}} - \theta} \leq \abs{\braket{s}{\hat v_h} - \theta} \leq \abs{\braket{s}{\hat v_h} - \theta_h} + \abs{\theta_h - \theta} = O \parens{h^{1/2}} + O \parens{ h^{\gamma/(2\pi)}} = O \parens{ h^{\gamma/(2\pi)}}.
\end{equation}
This proves the theorem. 
\end{proof}

The following lemma shows the first step of \Cref{thm:uniform_initial_state_preparation}:
\begin{lemma}
    \label{lem:overlap_theta_h_convergence}
    Let \(s\), \(\hat v_h\), \(\theta\) and \(\theta_h\) be as defined in \Cref{thm:uniform_initial_state_preparation}. Then, we have
    \begin{equation}
        \abs{\braket{s}{\hat v_h} - \theta_h} = O \parens{h^{1/2}}.
    \end{equation}
\end{lemma}

\begin{proof}
    We have
    \begin{equation}
    \braket{s}{\hat v_h} = \frac{1}{\sqrt N} \frac{\mathbf{1}^T C^{1/2}_h u_h}{\norm{C^{1/2}_h u_h}_2} = \frac{ \frac{1}{\sqrt N} \parens{C_{h}^{1/2} \mathbf 1 }^T u_h }{\norm{C^{1/2}_h u_h}_2}.  
    \end{equation}
    
    Let \(I_r\) be the set of indices corresponding to mesh nodes in region \(r\). For nodes at the boundary of multiple regions, we assign them to one of the regions arbitrarily. Let \(e_r = \sum_{\mathbf i \in I_r} \ket {\mathbf i}\) be the vector of all ones corresponding to region \(r\). Let \(R\) be the set of all regions. Then, we have
    \begin{equation}
        \sum_{r \in R} e_r = \mathbf 1.
    \end{equation}

    Next, let \(\mathbf{g}_h \in \mathbb R^N\) be the vector defined by \((\mathbf{g}_h)_{\mathbf i} = \sqrt{f(x_{\mathbf i})}\). For \(\mathbf i\) at a region boundary, we use the value \(f_r\) of the region to which \(\mathbf i\) was assigned, so that \(\mathbf{g}_h = \sum_{r \in R} \sqrt{f_r} e_r\).

    To prove the lemma, we separately consider the numerator and denominator. The outline of the proof is as follows: 
    \begin{itemize}
        \item Step 1(a): \(C_h^{1/2} e_r \approx h^{d/2}\sqrt{f_r} e_r\) for each region \(r\), where \(f_r\) is the value of \(f\) in region \(r\).
        \item Step 1(b): We can then show that combining all regions, \(C_h^{1/2} \mathbf 1 = \sum_{r \in R} C_h^{1/2} e_r \approx h^{d/2}\sum_{r \in R} \sqrt{f_r} e_r = h^{d/2}\mathbf{g}_h\).
        \item Step 1(c): This is then used to show that the numerator \(\frac{1}{\sqrt N} \parens{C_{h}^{1/2} \mathbf 1 }^T u_h \approx \frac{1}{\sqrt N} h^{d/2} \mathbf{g}_h^T u_h \approx \int_\Omega \sqrt f \phi_h\).
        \item Step 2: For the denominator, we have \(\norm{C_h^{1/2} u_h}_2 = \sqrt{u_h^T C_h u_h} = \sqrt{\int_\Omega f \phi_h^2}\).
    \end{itemize}

    We then combine these steps to show the lemma statement. We now proceed to show each of these steps in detail. 

    \textbf {Step 1(a):} Here we show that for every region \(r\), including those with \(f_r = 0\), we have
    \begin{equation}
        \norm{C_h^{1/2} e_r  - h^{d/2} \sqrt{f_r} e_r}_2 = O(\sqrt{f_{\max}} h^{1/2})
    \end{equation}

    Let \(J_r \subseteq I_r\) be ``clean indices'' that lie at least \(2h\) inside region \(r\). Define
    \begin{equation}
    e_r = e_r^o + e_r^\partial, \quad e_r^o = \sum_{i \in J_r} \ket i, \quad e_r^\partial = \sum_{i \in I_r\backslash J_r} \ket i.
    \end{equation}
    
    We first consider the clean part. Define the ``weight'' \(w_r^o(\bx) = \sum_{j \in J_r} \varphi_j(\bx)\). By the definition of \(C_h\) in \Cref{prob:discrete_problem}, we have
    \begin{equation}
        \parens{C_h e_r^o}_{i} = \sum_{j \in J_r} \parens{C_h}_{i, j} = \sum_{j \in J_r} \int_\Omega f(\bx) \varphi_{i}(\bx) \varphi_{j}(\bx) d\bx = f_r \int_\Omega   \varphi_{i}(\bx) w_r^o(\bx) d\bx.
    \end{equation}

    Since \(0 \leq w_r^o \leq 1\), we therefore have, for \textit{every} index \(i\), 
    \begin{equation}
        \label{eq:clean_region_bound}
    0 \leq (C_h e_r^o)_{i} \leq f_r \int_\Omega   \varphi_{i}(\bx) d\bx = f_r h^d.
    \end{equation}

    Let \(b_r^o = C_he_r^o - h^d f_r e_r^o\). Because \((e_r^o)_{\mathbf i} \in \{0, 1\}\), we have \(\abs{(b_r^o)_{\mathbf i}} \leq f_r h^d\) for all \(\mathbf i\). Observe that except for the boundary layer of the clean region, for all other nodes \(\mathbf i\), we have \(\parens{{b_r^o}}_{\mathbf i} = 0\). This is because sufficiently inside the clean region \(\parens{C_h e_r^o}_{\mathbf i} = f_r h^d\), and outside the clean region both \(\parens{C_h e_r^o}_{\mathbf i} = 0\) and \((e_r^o)_{\mathbf i} = 0\). Since the clean region boundary consists of \(O \parens{ \parens{\frac 1 h}^{d-1}} \) nodes, we have
    \begin{equation}
        \label{eq:clean_region_boundary_norm_bound}
        \norm{b_r^o}_2 = O \parens{f_r h^d h^{-(d-1)/2}} = O(f_r h^{(d+1)/2}).
    \end{equation}

    Now consider the case where \(f_r > 0\). We can apply the identity 
    \begin{equation}
        \label{eq:identity_for_f_r_positive}
        C_h^{1/2} - h^{d/2} \sqrt{f_r} I = \parens{C_h^{1/2} + h^{d/2} \sqrt{f_r} I}^{-1} \parens{C_h - h^d f_r I}. 
    \end{equation}
    Because \(C_h\) is positive semidefinite, the inverse norm is at most \(1/(h^{d/2} \sqrt{f_r})\). Combining with \Cref{eq:clean_region_boundary_norm_bound}, we have
    \begin{equation}
        \label{eq:clean_region_final_bound}
        \norm{C_h^{1/2} e_r^o - h^{d/2} \sqrt{f_r} e_r^o}_2 = O(\sqrt{f_r} h^{1/2}).
    \end{equation}
    If \(f_r = 0\), we have \(C_h e_r^o = 0\) and hence \(C_h^{1/2} e_r^o = 0\).

    Next, we consider the boundary part. Using \Cref{thm:properties_of_C} for \(d\) dimensions, we have \(\norm{C_h}_2 = O(f_{\max} h^d)\). Since the boundary layer consists of \(O \parens{ \parens{\frac 1 h}^{d-1}} \) nodes, we have
    \begin{equation}
    \begin{aligned}
        \label{eq:boundary_region_final_bound}
        \norm{C_h^{1/2} e_r^\partial - h^{d/2} \sqrt{f_r} e_r^\partial}_2 &\leq \parens{\norm{C_h^{1/2}}_2  + h^{d/2} \sqrt{f_r}} \norm{e_r^\partial}_2 = h^{d/2} \parens{\sqrt{f_{\max}} + \sqrt{f_r}} O(h^{-(d-1)/2})\\
        &= O(\sqrt{f_{\max}} h^{1/2}).
    \end{aligned}
    \end{equation}

    Combining \Cref{eq:clean_region_final_bound} and \Cref{eq:boundary_region_final_bound}, we have
    \begin{equation}
        \label{eq:final_bound_for_each_region}
        \norm{C_h^{1/2} e_r - h^{d/2} \sqrt{f_r} e_r}_2 = O(\sqrt{f_{\max}} h^{1/2}).
    \end{equation}

    \textbf{Step 1(b):} We now combine all the regions:
    \begin{equation}
    \label{eq:combining_all_regions}
        \begin{aligned}
        \norm{C_h^{1/2} \mathbf 1 - h^{d/2}\mathbf{g}_h}_2
        &= \norm{\sum_{r \in R} \parens{C_h^{1/2} e_r - h^{d/2} \sqrt{f_r} e_r}}_2 \\
        &\leq \sum_{r \in R} \norm{C_h^{1/2} e_r - h^{d/2} \sqrt{f_r} e_r}_2
        = O \parens{z \sqrt{f_{\max}} h^{1/2}},
        \end{aligned}
    \end{equation}
    where \(z\) is the number of regions.

    \textbf{Step 1(c):} We now use the above to bound the numerator
    \begin{equation}
    \frac{1}{\sqrt N} h^{d/2} \mathbf{g}_h^T u_h = \frac{h^d \mathbf{g}_h^T u_h}{(1 - h)^{d/2}} .
    \end{equation}

    Let \(p \in \mathbb R^N\) be such that \(p_i = \int_\Omega \sqrt{f(\bx)} \varphi_{\mathbf i}(\bx) d\bx\). Then we have \(\int_\Omega \sqrt{f(\bx)} \phi_h(\bx) d\bx = p^T u_h\). Next, we define \(q \in \mathbb R^N\) such that \(q_{\mathbf i} = h^d (\mathbf{g}_h)_{\mathbf i} - p_{\mathbf i}\). We will proceed to bound \(\norm{q}_2\). Observe that for all clean nodes within each region, we have \(q_{\mathbf i} = 0\). There are at most \(O(z h^{-(d-1)})\) dirty nodes across all regions, and for each such node \(\abs{q_{\mathbf i}} = O(\sqrt{f_{\max}} h^d)\). Thus, \(\norm{q}_2 = O(\sqrt{z f_{\max}} h^{(d+1)/2})\), with constants depending on the region geometry. In the bounds that follow, we will repeatedly use the following relation between \(\norm{u_h}_2\) and \(\norm{\phi_h}_{L^2(\Omega)}\):
    \begin{equation}
        \label{eq:u_h_L2_norm}
        \norm{u_h}_2 = O(h^{-d/2}\norm{\phi_h}_{L^2(\Omega)}).
    \end{equation}

    This norm relation follows from the generalization of \Cref{lem:mass_matrix_3d_eigs} to \(d\) dimensions. We first show that \(h^d \mathbf{g}_h^T u_h\) is close to \(p^Tu_h = \int_\Omega \sqrt{f(\bx)} \phi_h(\bx) d\bx\), and then that \(\frac{h^d \mathbf{g}_h^T u_h}{(1 - h)^{d/2}}\) is close to \(h^d \mathbf{g}_h^T u_h\).

    For the former, we have 
    \begin{equation}
        \label{eq:s_p_bound}
        \abs{h^d \mathbf{g}_h^T u_h - p^T u_h} = \abs{q^T u_h} \leq \norm{q}_2 \norm{u_h}_2 = O(\sqrt{z f_{\max}} \norm{\phi_h}_{L^2(\Omega)}h^{1/2}).
    \end{equation}

    For the latter, we have
    \begin{equation}
        \label{eq:s_h_bound}
        \begin{aligned}
        \abs{\frac{h^d \mathbf{g}_h^T u_h}{(1 - h)^{d/2}} - h^d \mathbf{g}_h^T u_h} = \abs{h^d \mathbf{g}_h^T u_h} \abs{\frac{1}{(1 - h)^{d/2}} - 1} &= \abs{h^d \mathbf{g}_h^T u_h} O(h) \leq C \sqrt{f_{\max}} \norm{\phi_h}_{L^2(\Omega)} \cdot O(h) \\
        &= O \parens{\sqrt{f_{\max}} \norm{\phi_h}_{L^2(\Omega)} h}.
        \end{aligned}
    \end{equation}

    Thus, by triangle inequality, combining \Cref{eq:s_p_bound} and \Cref{eq:s_h_bound}, we have
    \begin{equation}
        \abs{\frac{1}{\sqrt N} h^{d/2} \mathbf{g}_h^T u_h - \int_\Omega \sqrt{f(\bx)} \phi_h(\bx) d\bx} = O(\sqrt{z f_{\max}} \norm{\phi_h}_{L^2(\Omega)} h^{1/2}).
    \end{equation}

    Therefore
    \begin{equation}
        \label{eq:numerator_bound}
        \begin{aligned}
        &\abs{\frac{1}{\sqrt N} \parens{C_{h}^{1/2} \mathbf 1 }^T u_h - \int_\Omega \sqrt{f(\bx)} \phi_h(\bx) d\bx} \\
        &= \abs{\frac{1}{\sqrt N} \parens{C_{h}^{1/2} \mathbf 1 }^T u_h - \frac{1}{\sqrt N} h^{d/2} \mathbf{g}_h^T u_h + \frac{1}{\sqrt N} h^{d/2} \mathbf{g}_h^T u_h - \int_\Omega \sqrt{f(\bx)} \phi_h(\bx) d\bx} \\
        &\leq \abs{\frac{1}{\sqrt N} \parens{C_{h}^{1/2} \mathbf 1 }^T u_h - \frac{1}{\sqrt N} h^{d/2} \mathbf{g}_h^T u_h} + \abs{\frac{1}{\sqrt N} h^{d/2} \mathbf{g}_h^T u_h - \int_\Omega \sqrt{f(\bx)} \phi_h(\bx) d\bx} \\
        &\leq \frac{1}{\sqrt N} \norm{C_h^{1/2} \mathbf 1 - h^{d/2}\mathbf{g}_h}_2 \norm{u_h}_2
        + O(\sqrt{z f_{\max}} \norm{\phi_h}_{L^2(\Omega)} h^{1/2}) \\
        &\leq O \parens{z \sqrt{f_{\max}} h^{1/2}} \frac{1}{\sqrt N} \norm{u_h}_2
        + O(\sqrt{z f_{\max}} \norm{\phi_h}_{L^2(\Omega)} h^{1/2}) \\
        &= O \parens{z \sqrt{f_{\max}} h^{1/2}} \frac{h^{-d/2}}{\sqrt N} \norm{\phi_h}_{L^2(\Omega)}
        + O(\sqrt{z f_{\max}} \norm{\phi_h}_{L^2(\Omega)} h^{1/2}) \\
        &= O \parens{z \sqrt{f_{\max}} \norm{\phi_h}_{L^2(\Omega)} h^{1/2}}.
        \end{aligned}
    \end{equation}
    Here we used Cauchy--Schwarz, the Step 1(b) bound \Cref{eq:combining_all_regions}, the norm relation \Cref{eq:u_h_L2_norm}, and \(h^{-d/2}/\sqrt N = (1-h)^{-d/2} = O(1)\). Since \(z \geq 1\), the \(\sqrt{z}\) term is absorbed into the \(z\) term. This bounds the numerator without any dependence on a positive lower bound for \(f\).

    \textbf{Step 2:} Next, we rewrite the denominator \(\norm{C_h^{1/2} u_h}_2\).

    We have
    \begin{equation}
    \norm{C_h^{1/2} u_h}_2^2 = u_h^T C_h u_h = \int_\Omega f(\bx) \phi_h^2(\bx) d\bx = \norm{\sqrt{f} \phi_h}_{L^2(\Omega)}^2.
    \end{equation}
    Thus, 
    \begin{equation}
        \label{eq:denominator_bound}
    \norm{C_h^{1/2} u_h}_2 = \norm{\sqrt{f} \phi_h}_{L^2(\Omega)}.
    \end{equation}

Finally, it remains to combine the bounds on the numerator and denominator. 
Using \Cref{eq:numerator_bound} and \Cref{eq:denominator_bound}, we have
\begin{equation}
    \begin{aligned}
\abs{\braket{s}{\hat v_h} - \theta_h} &= \frac{\abs{\frac{1}{\sqrt N} \parens{C_{h}^{1/2} \mathbf 1 }^T u_h - \int_\Omega \sqrt{f(\bx)} \phi_h(\bx) d\bx}}{\norm{C_h^{1/2} u_h}_2} \\
&= \frac{\abs{\frac{1}{\sqrt N} \parens{C_{h}^{1/2} \mathbf 1 }^T u_h - \int_\Omega \sqrt{f(\bx)} \phi_h(\bx) d\bx}}{ \norm{\sqrt{f} \phi_h}_{L^2(\Omega)}} \\
&= O \parens{z \sqrt{f_{\max}} \parens{\frac{\norm{\phi_h}_{L^2(\Omega)}}{\norm{\sqrt{f} \phi_h}_{L^2(\Omega)}}} \cdot  h^{1/2}} \\
&= O(h^{1/2})
    \end{aligned}
\end{equation}
where \(z\), \(f_{\max}\), and the region geometry are fixed independently of \(h\). The last line follows from the fact that the norm ratio is uniformly bounded: the discrete eigenvalue equation and \Cref{eq:def_a} give
\begin{equation}
    \Sigma_{a,\min} \norm{\phi_h}_{L^2(\Omega)}^2
    \leq a(\phi_h,\phi_h)
    = \lambda_h \norm{\sqrt f \phi_h}_{L^2(\Omega)}^2,
\end{equation}
so the ratio is at most \(\sqrt{\lambda_h/\Sigma_{a,\min}} = O(1)\), since \(\lambda_h \to \lambda\). This remains valid when \(f\) vanishes in some regions.
\end{proof}

The following lemma is used in the second step of \Cref{thm:uniform_initial_state_preparation} and bounds the difference between  \(\theta\) and \(\theta_h\).
\begin{lemma}
    \label{lem:theta_theta_h_convergence}
    Let \(\theta\) and \(\theta_h\) be as defined in \Cref{thm:uniform_initial_state_preparation}. Then, we have
    \begin{equation}
        \abs{\theta_h - \theta} = \abs{\frac{\int_\Omega \sqrt{f(\bx)} \phi_h(\bx) d\bx}{\norm{\sqrt{f} \phi_h}_{L^2(\Omega)}} - \frac{\norm{\sqrt{f} \phi}_{L^1(\Omega)}}{\norm{\sqrt{f} \phi}_{L^2(\Omega)}}} = O \parens{h^{\gamma/(2\pi)}}.
    \end{equation}
\end{lemma}
\begin{proof}
    
First, taking the denominator inside the integral, and using Cauchy-Schwarz on the expression inside the integral as \(\abs{\int_\Omega q} \leq \norm{1}_{L^2} \norm{q}_{L^2}\) = \(\norm{q}_{L^2}\) (where we are using \(L^2\) to mean \(L^2(\Omega)\))), we have 
\begin{equation}
\begin{aligned}
    \abs{\theta - \theta_h} &= \abs{\frac{\int_\Omega \sqrt{f(\bx)} \phi_h(\bx) d\bx}{\norm{\sqrt{f} \phi_h}_{L^2}} - \frac{\norm{\sqrt{f} \phi}_{L^1(\Omega)}}{\norm{\sqrt{f} \phi}_{L^2}}} \\
    &= \abs{\int_\Omega \parens{\frac{\sqrt{f(\bx)} \phi_h(\bx)}{\norm{\sqrt{f} \phi_h}_{L^2}} - \frac{\sqrt{f(\bx)}{\phi(\bx)}}{\norm{\sqrt{f} \phi}_{L^2}}}} \\
    &\leq \norm{\frac{\sqrt{f} \phi_h}{\norm{\sqrt{f} \phi_h}_{L^2}} - \frac{\sqrt{f} \phi}{\norm{\sqrt{f} \phi}_{L^2}}}_{L^2}.
\end{aligned}
\end{equation}
Applying \Cref{lem:normalized-vec-diff}, we have
\begin{equation}
    \begin{aligned}
        \abs{\theta_h - \theta} 
        &\leq \frac{2 \norm{\sqrt{f} (\phi_h - \phi)}_{L^2}}{\norm{\sqrt{f} \phi}_{L^2}}.
    \end{aligned}
\end{equation}
Consider the denominator, 
\begin{equation}
\begin{aligned}
    \norm{\sqrt{f} \phi}_{L^2}^2 &= b(\phi, \phi) = \frac{a(\phi, \phi)}{\lambda} = \frac 1 {\lambda}.
\end{aligned}
\end{equation}
Substituting back, we find
\begin{equation}
    \label{eq:l2_bound_final}
\begin{aligned}
    \abs{\theta_h - \theta} &\leq 2 \sqrt{\lambda} \norm{\sqrt{f} (\phi_h - \phi)}_{L^2}.
\end{aligned}
\end{equation}

Combining Theorem 6.1 and Theorem 8.1 of \cite{babuska_1991_eigenvalue_problems}, and using \(\beta(h) = \min(D_{\min}, \Sigma_{a, \min})\) as shown in \Cref{thm:babuska_osborn_satisfied}, we have (see \Cref{def:eigenspaces})
\begin{equation}
    \label{eq:delta_bound_old}
    \max \parens{ \delta \parens{\widebar M(\lambda), \widebar M_h(\lambda)}, \delta \parens{ \widebar M_h(\lambda), \widebar M(\lambda)}} = O(\eps_h(\lambda)),
\end{equation}
where
\begin{equation}
\label{eq:delta_def_old}
\delta(X, Y) \coloneqq \sup_{x \in X, \norm{x}_{H^1} = 1} \inf_{y \in Y} \norm{x - y}_{H^1}.
\end{equation}

Let us denote \(\eps_h \coloneqq \eps_h(\lambda)\). Then using the above bound,
we know that there exists \(\xi \in \widebar M(\lambda)\) such that
\begin{equation} 
    \label{eq:initial_coarse_to_true_old}
    \norm{ \frac{\phi_c}{\norm{\phi_c}_{H^1}} - \xi}_{H^1} = O(\eps_{h_c}).
\end{equation}

From \Cref{eq:initial_coarse_to_true_old}, and due to the simplicity of the minimum eigenvalue \Cref{lem:simple_positive_eigenfunction}, we have, for \(\beta_h > 0\) for sufficiently small \(h\), 
\begin{equation}
    \begin{aligned}
        \norm{\frac{\phi_h}{\norm{\phi_h}}_{H^1} - \beta_h \phi}_{H^1} &= O(h^{\gamma/(2\pi)}) \\
        &\implies \norm{\frac{\phi_h}{\norm{\phi_h}}_{H^1} - \beta_h \phi}_a = O(h^{\gamma/(2\pi)}) \\
        &\implies \norm{\phi_h - {\beta_h} \norm{\phi_h}_{H^1} \phi}_a = O\parens{\norm {\phi_h}_{H^1} h^{\gamma/(2\pi)}} = O\parens{ \norm{\phi_h}_a h^{\gamma/(2\pi)}} = O(h^{\gamma/(2\pi)}).
    \end{aligned}
\end{equation}

Then applying \Cref{lem:normalized-vec-diff} below, we have for some constant \(c\),
\begin{equation}
    \label{eq:l2_bound_theta_lemma}
    \begin{aligned}
        \norm{\phi_h - \phi}_{L^2} \leq \norm{\phi_h - \phi}_{H^1} \leq c \norm{\phi_h - \phi}_a \leq 2c\frac{\norm{\phi_h - \beta_h \norm{\phi_h}_{H^1} \phi}_a}{\norm{\phi_h}_a} = O\parens{h^{\gamma/(2 \pi)}}.
    \end{aligned}
\end{equation}

Combining \Cref{eq:l2_bound_final} and \Cref{eq:l2_bound_theta_lemma}, we obtain the result. 
\end{proof}

We now establish the following lemma on the difference of normalized functions used in the proof above.
\begin{lemma}
\label{lem:normalized-vec-diff}
For vectors $u$ and $v$ with $\hat{u} = \frac{u}{\norm{u}}$ and $\hat{v} = \frac{v}{\norm{v}}$,
\begin{equation}
\norm{\hat{u} - \hat{v}} \le \frac{2\norm{u-v}}{\norm{u}}.
\end{equation}
\end{lemma}

\begin{proof}
We add and subtract $\frac{v}{\norm{u}}$ and apply the triangle inequality, giving
\begin{align}
\norm{\hat{u} - \hat{v}} = \norm{\frac{u}{\norm{u}} - \frac{v}{\norm{v}}}
&=  \norm{\frac{u}{\norm{u}} - \frac{v}{\norm{u}} + \frac{v}{\norm{u}} - \frac{v}{\norm{v}}} \\
&\le \norm{\frac{u}{\norm{u}} - \frac{v}{\norm{u}}} + \norm{\frac{v}{\norm{u}} - \frac{v}{\norm{v}}} \\
&= \norm{\frac{u-v}{\norm{u}}} + \norm{v} \norm{\frac{\norm{v} - \norm{u}}{\norm{u}\norm{v}}} \\
&\le \frac{2\norm{u-v}}{\norm{u}},
\end{align}
where in the last line we apply the reverse triangle inequality $\norm{u-v} \ge \norm{\norm{u} - \norm{v}}$.
\end{proof}
\FloatBarrier
{
\clearpage\section{Uniform Superposition Overlap}\suppressfloats[t]
\label{sec:uniform_superposition_overlap_experiments}
In this section, we present experimental results for the success probability, \(|\braket{s}{ \hat{v}_h}|^2\), which is the squared overlap between the uniform superposition state \(\ket{s}\) and the normalized eigenvector \(\ket{\hat{v}_h} = \frac{C^{1/2} \hat{u}_h}{\norm{C^{1/2} \hat{u}_h}}\).

The results using the 2D, \(4 \cross 4\), 2-material checkerboard and 3D, \(4 \cross 4 \cross 4\), 2-material checkerboard cases are presented in \Cref{fig:uniform_superposition_success_prob_2D_checkerboard} and \Cref{fig:uniform_superposition_success_prob_3D_checerkboard}, respectively. For both cases, \(\nu \Sigma_f(\bx)\) and \(\Sigma_a(\bx)\) are set to be uniformly equal to 1 throughout the domain while \(D(\bx) = D_{\max}\) in material 1 and \(D(\bx) = D_{\min}\) in material 2. For all of the cases tested, the success probability for the uniform superposition state exhibited convergence to a constant probability as predicted in \Cref{sec:uniform_superposition_initial_state_proof}. This suggests that the uniform superposition state can be used as a reasonable initial guess for these checkerboard geometries for arbitrarily fine FEM discretizations.

\begin{figure}
    \centering
    \includegraphics[width=0.6\textwidth]{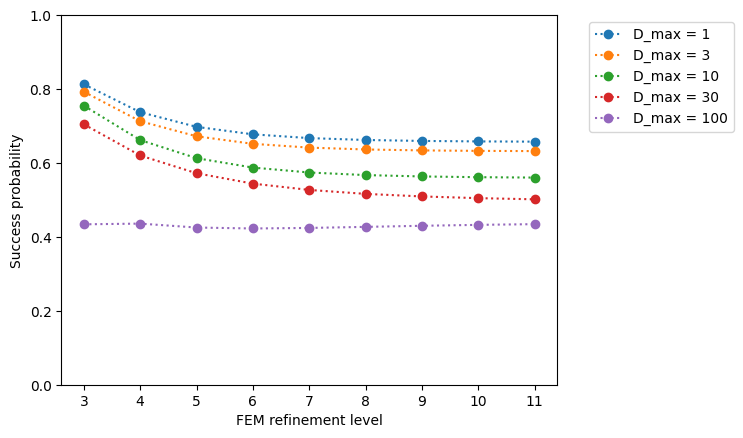}
    \caption{Squared overlap between the uniform superposition state and the eigenvector solution of the 2D checkerboard for varying maximum diffusion coefficients (\(D_{\min} = 1\)) and levels of the fine FEM discretization.}
    \label{fig:uniform_superposition_success_prob_2D_checkerboard}
\end{figure}

\begin{figure}
    \centering
    \includegraphics[width=0.6\textwidth]{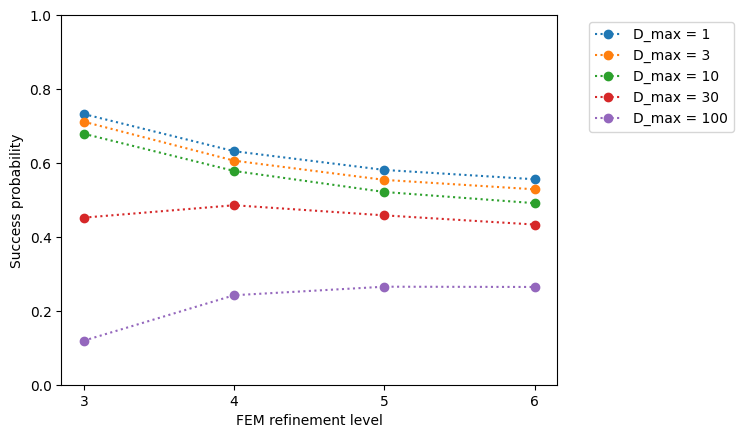}
    \caption{Squared overlap between the uniform superposition state and the eigenvector solution of the 3D checkerboard for varying maximum diffusion coefficients (\(D_{\min} = 1\)) and levels of the fine FEM discretization.}
    \label{fig:uniform_superposition_success_prob_3D_checerkboard}
\end{figure}

We also present the uniform superposition success probabilities for a one-speed solution of the 2D C5G7 benchmark nuclear reactor geometry \cite{Lewis_2001_C5G7}. The C5G7 reactor material data were originally given in 7-energy-group form, but they were collapsed to one-group constants for this simulation. This was done by averaging each of the total, fission production, and absorption 7-group cross sections weighted by the Monte Carlo multigroup scalar flux solutions. Because the C5G7 reactor geometry is assumed to have isotropic neutron scattering (\(\Sigma_{s1} = 0\)), the diffusion coefficients were found using \(D = \frac{1}{3 \Sigma_t}\). The C5G7 geometry's larger material grid (\(128 \cross 128\) compared to the \(4 \cross 4\)) and the presence of zero fission cross section regions caused its success probability convergence to be more gradual than the checkerboard cases. However, it still appears to be approaching a constant, non-zero success probability. This indicated that even for more realistic reactor geometries like C5G7, the asymptotic success probability is still reasonably high. Similar poor convergence can be seen in the checkerboard case when zero fission cross section regions and larger material grid sizes are included. \Cref{fig:uniform_superposition_success_prob_2D_checkerboard_large_grid} shows a \(2^{r_m} \cross 2^{r_m}\) checkerboard case with varying material grid refinement level \(r_m\). One of the two material regions has a 0 fission cross section. Larger \(r_m\) values seem to lead to slower convergence towards the success probability asymptote as well as a larger value for the asymptote.

\begin{figure}
    \centering
    \includegraphics[width=0.6\textwidth]{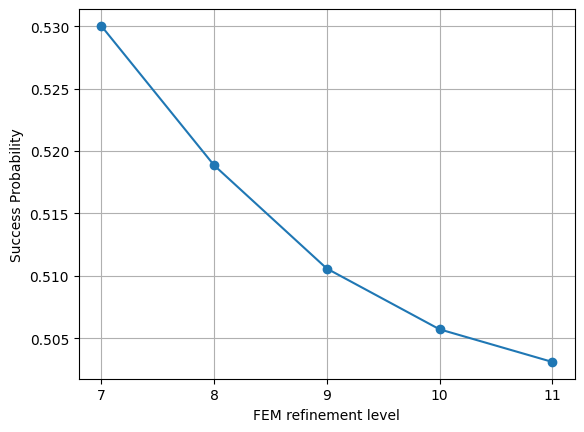}
    \caption{Squared overlap between the uniform superposition state and the eigenvector solution \(\ket{\hat{v}_h}\) of the C5G7 one-speed reactor geometry for varying levels of the fine FEM discretization }
    \label{fig:uniform_superposition_success_prob_C5G7}
\end{figure}

\begin{figure}
    \centering
    \includegraphics[width=0.8\textwidth]{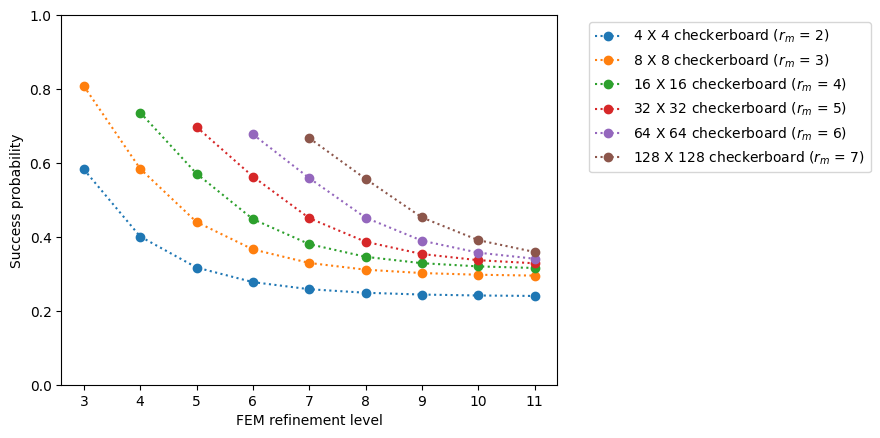}
    \caption{Success probability between the uniform superposition state and the eigenvector solution of the 2D checkerboard for varying checkerboard grid sizes and FEM refinement levels. \(\nu \Sigma_{f,min} = 0\), \(D_{\min} = 1\), \(D_{\max} = 10\). }
    \label{fig:uniform_superposition_success_prob_2D_checkerboard_large_grid}
\end{figure}

}
\FloatBarrier

{
\section{Nonuniform Grid Classical Scaling Simulations} \label{sec:nonuniform_classical_grid_scaling}

In this section, we show numerical experiments for the scaling of nonuniform meshing schemes for classical algorithms. 

The poor scaling of degrees of freedom \(N\) with respect to eigenvalue error \(\epsilon_k\) seen in the checkerboard grids with high diffusion coefficient ratios in \Cref{sec:numerical_experiments} were due to low eigenvector solution regularity at material vertices \cite{Petzoldt2001}. A common strategy to reduce these poor-regularity effects is to use a nonuniform mesh that decreases the discretization length near these material discontinuities \cite{Li_2022}. For the nonuniform numerical experiments shown in this section, we use a polynomially-graded mesh with FEM cell boundaries, \(x_{j,i}\), in each dimension defined as

\begin{align}
x_{j,i} &= \frac{j}{M} + \frac{1}{M}
\begin{cases}
\displaystyle
\tfrac{1}{2}
\left(
\tfrac{2i}{n}
\right)^\gamma,
& 0\le i\le \tfrac{n}{2}, \\
1-\frac{1}{2}
\left(
2-\frac{2i}{n}
\right)^\gamma,
& \tfrac{n}{2}\le i\le n
\end{cases}
&
j &= 0, \ldots, M-1 \quad \quad i = 0, \ldots, n
\end{align}
where, in a single dimension, \(M\) is the number of material regions, \(n = (1/h) / M\) is the number of FEM cells per material region, and \(\gamma\) is a parameter that determines the order of the polynomial mesh grading, with higher \(\gamma\) being better suited for lower-regularity solutions. The cell boundaries for a multidimensional system are simply the tensor product of the single-dimension cell boundaries. 
Given the nonuniform and uniform meshes have the same number of degrees of freedom, this results in the nonuniform mesh being coarser than the uniform mesh in interior material region nodes but finer near material boundaries. This is illustrated with a \(2D\), \(4 \cross 4\) checkerboard example in \Cref{fig:mesh_scaling}.

\begin{figure}[h]
\centering
\begin{subfigure}{.5\textwidth}
  \centering
  \includegraphics[width=0.6\textwidth]{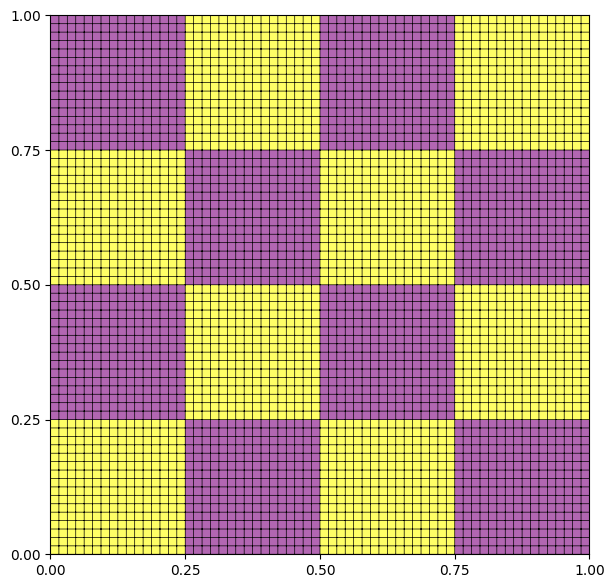}
  \caption{\(4 \times 4\) material grid, \(64 \times 64\) FEM uniform mesh}
\end{subfigure}%
\begin{subfigure}{.5\textwidth}
  \centering
  \includegraphics[width=0.6\textwidth]{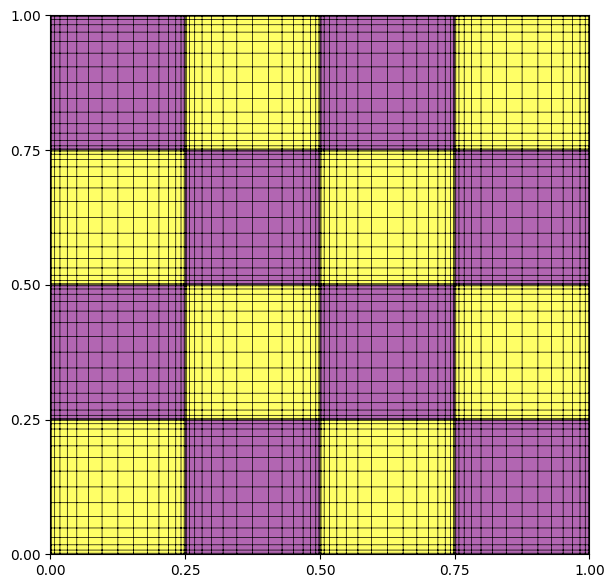}
  \caption{\(4 \times 4\) material grid, \(64 \times 64\) FEM nonuniform mesh}
\end{subfigure}
\caption{Uniform and polynomially graded nonuniform mesh comparisons with the same number of degrees of freedom.}
\label{fig:mesh_scaling}
\end{figure}

As discussed in \Cref{sec:numerical_experiments}, the error in the eigenvalue solution introduced from uniform FEM discretization is \( \epsilon = \abs{\lambda - \lambda_h} = O(h^{2 \chi'})\). Combining this with the relationship between \(h\) and the number of degrees of freedom of the \(d\)-dimensional discretized system \(N = \Theta(1/h^d)\) implies that \(N = O(\epsilon^{-d/2 \chi'})\). We define this exponent as \(p' = d/(2 \chi')\). Since \(\chi' \leq 1\), \(p' \geq d/2\) regardless of the discretization scheme. We use \Cref{eq:observed_order_convergence} to determine the values of \(p'\) for uniform and nonuniform meshing to compare their respective orders of convergence. The results of this analysis for a \(4 \cross 4\) checkerboard pattern in 2D and a \(4 \cross 4 \cross 4\) checkerboard pattern in 3D are shown in \Cref{fig:nonuniform_convergence_exponents_2D} and \Cref{fig:nonuniform_convergence_exponents_3D}, respectively. For both the 2D and 3D cases, the scaling exponent from the polynomially-graded nonuniform mesh case appears to approach the optimal value of \(p'=1\) and \(p'=3/2\), respectively, as the number of degrees of freedom increases. In contrast, the scaling exponent from the uniform mesh case appears to approach significantly larger \(p'\) values. This implies that, at least for the cases tested but likely for many practical geometries, it is possible to construct a nonuniform mesh that allows classical methods to solve \Cref{eq:diffusion_equation_first} with a number of degrees of freedom close to the optimal \(N = O(\epsilon^{-d/2})\) rather than the \(N = O(\epsilon^{-d/(2\chi)})\) bound given in \Cref{sec:numerical_experiments}. In cases with large diffusion coefficient jumps at material boundaries, this significantly decreases the potential quantum advantage. Code for the experiments presented in \Cref{sec:uniform_superposition_overlap_experiments} and  \Cref{sec:nonuniform_classical_grid_scaling} is available at \url{https://github.com/Tinkidinki/diffusion-fem-codes}.

\begin{figure}
\centering
\begin{subfigure}{.5\textwidth}
  \centering
  \includegraphics[width=0.8\textwidth]{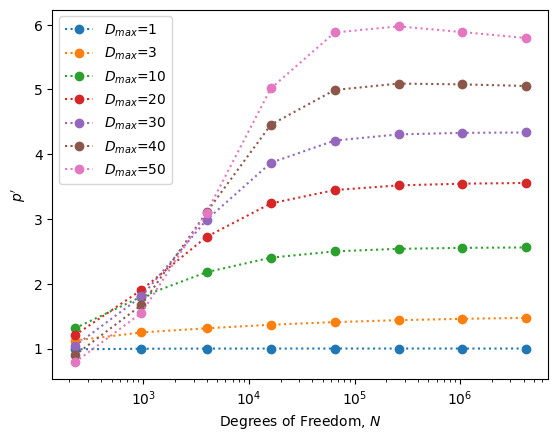}
  \caption{Uniform 2D FEM discretization convergence rates for various \(D_{\max}\) values}
  \label{fig:uniform_mesh_scaling}
\end{subfigure}%
\begin{subfigure}{.5\textwidth}
  \centering
  \includegraphics[width=0.83\textwidth]{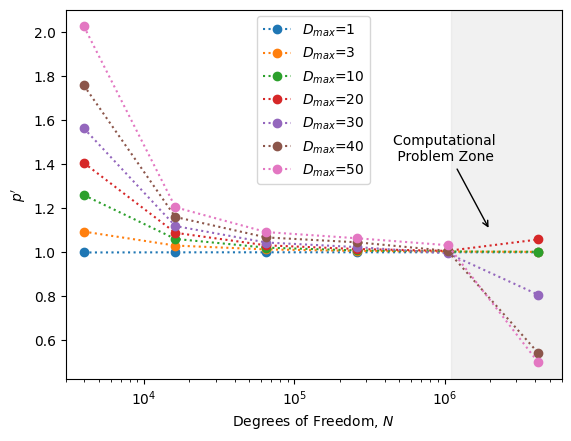}
  \caption{Nonuniform 2D mesh convergence rates (\(\gamma\) parameter optimized for each \(D_{\max}\)). For high \(\gamma\) and \(N\) values, nonphysical results likely resulting from floating-point precision errors are seen in the ``Computational Problem Zone.''}
  \label{fig:nonuniform_mesh_scaling}
\end{subfigure}
\caption{Uniform vs.\ nonuniform eigenvalue convergence rates for a 2D checkerboard geometry.}
\label{fig:nonuniform_convergence_exponents_2D}
\end{figure}

\begin{figure}
\centering
\begin{subfigure}{.5\textwidth}
  \centering
  \includegraphics[width=0.8\textwidth]{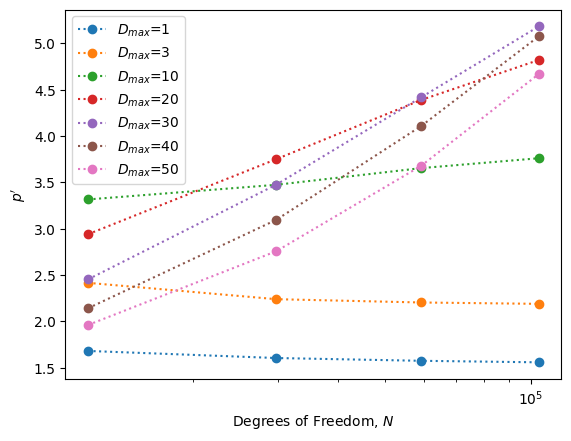}
  \caption{Uniform 3D FEM discretization convergence rates for various \(D_{\max}\) values}
  \label{fig:uniform_mesh_scaling_3D}
\end{subfigure}%
\begin{subfigure}{.5\textwidth}
  \centering
  \includegraphics[width=0.83\textwidth]{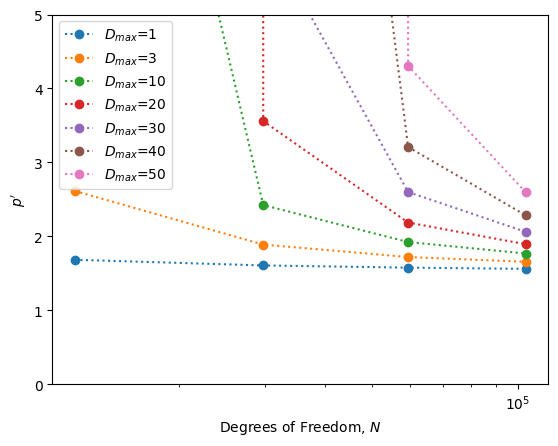}
  \caption{Nonuniform 3D mesh convergence rates (\(\gamma\) parameter optimized for each \(D_{\max}\))}
  \label{fig:nonuniform_mesh_scaling_3D}
\end{subfigure}
\caption{Uniform vs.\ nonuniform eigenvalue convergence rates for a 3D checkerboard geometry.}
\label{fig:nonuniform_convergence_exponents_3D}
\end{figure}

}

\end{document}